\documentclass[11pt]{article}

\usepackage{algorithm}
\usepackage{algorithmic}
\usepackage[]{graphicx}
\usepackage{amsmath}
\usepackage{latexsym}
\usepackage{amsfonts}
\usepackage{fullpage}
\usepackage{amstext}
\usepackage{amsthm}
\usepackage{url}

\usepackage{color}
\usepackage{enumerate}
\usepackage{setspace}
\usepackage{epsfig}

\usepackage{fullpage}

\newtheorem{theorem}{Theorem}
\newtheorem{definition}{Definition}

\newtheorem{property}{Property}

\newtheorem{corollary}{Corollary}
\newtheorem{lemma}{Lemma}

\newtheorem{assumption}{Assumption}
\newtheorem{constraint}{Constraint}
\newtheorem{fact}{Fact}

\newcommand{\xiu}{x^u_i}

\newcommand{\xbij}{\overline{x}_{i_j}}
\newcommand{\wtij}{\widetilde{w}_{i_j}}

\newcommand{\tis}{{\tau_i^s}}

\newcommand{\ybi}{\overline{y}_i}
\newcommand{\xbi}{\overline{x}_i}
\newcommand{\zbi}{\overline{z}_i}

\newcommand{\loa}{\lambda \alpha_1}
\newcommand{\sgn}{\mbox{~\!sign}}

\newcommand{\wti}{\widetilde{w}_i}

\newcommand{\sis}{s_i^*}

\newcommand{\lE}{\lmbd E}

\renewcommand{\span}{\operatorname{span}}

\newcommand{\eps}{\lambda}

\newcommand{\kmin}{\kappa}
\newcommand{\kmax}{\kappa}

\newcommand{\xbar}{\overline x}

\newcommand{\Del}{\Delta}
\newcommand{\del}{\delta}
\newcommand{\lmbd}{\lambda}
\newcommand{\gam}{\gamma}

\newcommand{\boldp}{\mathbf{p}}

\newcommand{\kpi}{\kappa}

\newcommand{\alo}{\alpha_1}
\newcommand{\alt}{\alpha_2}
\newcommand{\alr}{\alpha_3}
\newcommand{\alf}{\alpha_4}

\newcommand{\argmax}{\mathrm{argmax}}

\newcommand{\elas}{E}  

\newcommand{\Xomit}[1]{}
\Xomit{
     \newcommand{\qed}{\nobreak \ifvmode \relax \else
     \ifdim\lastskip<1.5em \hskip-\lastskip
     \hskip1.5em plus0em minus0.5em \fi \nobreak
     \vrule height0.75em width0.5em depth0.25em\fi}
}

\begin{document}

\title{Discrete Price Updates Yield Fast Convergence in Ongoing Markets with Finite Warehouses%
\footnote{The work of Richard Cole was
supported in part by NSF grants CCF-0515127 and CCF-0830516;
the work of Lisa Fleischer was supported in part
by NSF grants CCF-0515127, CCF-0728869 and CCF-1016778;
the work of Ashish Rastogi while a Ph.D. student at NYU was supported
in part by NSF grant CCF-0515127.}}

\author{Richard Cole \\
 Courant Institute \\ New York University
\and Lisa Fleischer \\
Dartmouth College
\and Ashish Rastogi \\
Goldman Sachs \& Co. \\
New York}

\maketitle

\begin{abstract}
This paper shows that in suitable markets,
even with out-of-equilibrium trade allowed,
a simple price update rule leads to rapid convergence toward the equilibrium.
In particular, this paper considers a Fisher market
repeated over an unbounded number of time steps,
with the addition of finite sized warehouses to enable non-equilibrium trade.

The main result is that suitable tatonnement style price updates
lead to convergence in a significant subset of markets
satisfying the Weak Gross Substitutes property.
Throughout this process the warehouse are always
able to store or meet demand imbalances
(the needed capacity depends on the initial imbalances).

Our price update rule is robust in a variety of regards:

\begin{itemize}
\item
The updates for each good depend only on information about that good
(its current price, its excess demand since its last update)
and occur asynchronously from updates to other prices.
\item
The process is resilient to error in the excess demand data.
\item
Likewise, the process is resilient to discreteness,
i.e.~a limit to divisibility, both of goods and money.
\end{itemize}
\end{abstract}

\thispagestyle{empty}

\newpage
\section{Introduction}
\label{sec:intro}

This paper investigates when a tatonnement-style price update in a dynamic
market setting could lead to fast convergent behavior.

The impetus for this work comes from the following question: why might
well-functioning markets be able to stay at or near equilibrium
prices?
\Xomit{
\footnote{We are not concerned with the question of whether
this assertion is indeed correct.}
}%
This raises two issues: what are
plausible price adjustment mechanisms and in what types of markets
are they effective?

This question was considered by Walras in 1874, when he
suggested that prices adjust by tatonnement: upward if there is too
much demand and downward if too little~\cite{Walras1874}. Since
then, the study of market equilibria, their existence, stability,
and their computation has received much attention in Economics,
Operations Research, and most recently in Computer Science.
McKenzie~\cite{mckenzie} gives a fairly recent account of the
classic perspective in economics.
The recent activity in Computer Science
has led to a considerable number of polynomial time algorithms
for finding approximate and exact equilibria in a variety of markets
with divisible
goods; we cite a selection of these
works~\cite{CodenottiMV05, leontief, DevanurV04, DevanurPSV02, DevKan08, GargKapoor04,
JainVaz07, Orlin10, vazirani10, VazYan10}.
 However, these algorithms do not seek to, and
do not appear to provide methods that might plausibly explain these
markets' behavior.

We argue here for the relevance of this question from a computer
science perspective. Much justification for looking at the market
problem in computer science stems from the following argument: If
economic models and statements about equilibrium and convergence are
to make sense as being realizable in economies, then they should be
concepts that are computationally tractable. Our viewpoint is that
it is not enough to show that the problems are computationally
tractable; it is also necessary to show that they are tractable in a
model that might capture how a market works.
It seems implausible that markets with many interacting players
(buyers, sellers, traders) would perform overt global
computations, using global information.

Analogous goals have arisen when
considering convergence to equilibria in game theory. Perhaps the
most successful approach has been to use regret minimizing
procedures, which has yielded convergence in several
settings~\cite{roughgarden, kleinberg,even-dar,blum}.
Other works have considered best response~\cite{fanelli08,nisan11} and proportional
response dynamics~\cite{wu07}.

A central concern with the tatonnement model of price updates,
and other models, is whether they cause prices to converge to equilibria.
Such results were shown for some types of
markets; early examples include~\cite{ArrowBH59,ArrowH58,NikaidoUzawa60,Uzawa60}.
However, there is no demonstration that these proposed update models converge
reasonably quickly.  Indeed, without care in the specific details,
they will not.\footnote{Of the referenced papers, only one
formalization \cite{Uzawa60} is a discrete algorithm, and, as we
make more specific later, it may not converge quickly.}

It has long been recognized that the tatonnement price adjustment
model is far from a realistic model: for example, in 1972, Fisher \cite{fisher72}
wrote
``such a model of price adjustment $\cdots$ describes nobody's actual behavior.''
This has led to to work on other price dynamics ranging from the
Hahn process~\cite{fisher72}
to non-equilibrium dynamics in the trading post model~\cite{shapleyShub77, Balasko},
and others~\cite{CrockettSS04}.
At the same time, there has been a continued interest in the
plausibility of tatonnement, and indeed its predictive accuracy
in a non-equilibrium trade setting has been shown in some experiments~\cite{hirota05}.

Plausibly, in many consumer markets buyers are myopic:
based on the current prices, goods are assessed on a take
it or leave it basis.
It seems natural that this would lead to out of equilibrium trade.
This is the type of setting in which we wish to consider our main question:
under what conditions can tatonnement style price updates  lead
to convergence?

\Xomit{
At first sight, it is not clear why these models are realistic. The
most studied, due to Walras, is the auctioneer model: an auctioneer
announces prices, receives the market demand at these prices in the
form of buy and sell requests, \emph{but with no trade actually
occurring}, adjusts prices according to the tatonnement procedure,
and iterates. Only when prices reach equilibrium is trade allowed.
In reality, it is trade that reveals demand and hence needed price
adjustments.
Perhaps as a result, there is a widespread belief that tatonnemment
is merely an algorithmic tool to
compute equilibria and that it does not describe actual market behavior
(see \cite{Balasko} for example).
Clearly, any realistic model has to enable trade in
disequilibrium.
}

Accordingly, in this paper we propose a simple market model in which the market
extends over time and trading occurs out of equilibrium (as well as
at equilibrium). We call this the \emph{Ongoing Market}. Here, the
market repeats from one time unit to the next; we call the basic
unit a \emph{day}.  The link from one day to the next is that goods
unsold one day are available the next day, in addition to the new
supply, which for simplicity, we take as being the same from day to
day. This appears to provide a simple and natural way of allowing
out-of-equilibrium trade. The algorithmic task is to converge to
equilibrium prices while clearing unsold stocks. We use the version of
tatonnment in which price updates are proportional both to the current
price and to the excess demand (modulo some minor details
to prevent very large price changes).
We show that it results in
rapid convergence toward equilibrium prices in this market.

Our analysis of the algorithm for the Ongoing Market relies in part on
understanding which we develop by
analyzing convergence in the traditional market problem
(used in the auctioneer model for example).
In this paper, we call this the \emph{One-Time Market}.
The algorithmic technique of iteratively computing prices for the
One-Time Market can be seen as a plausible approximation to the
Ongoing Market (but with no carry over of unsold goods).  Our work
can be seen as a formal justification for this approach.
Further, in
our opinion, the intuitive understanding that markets are usually
similar from one time period to the next has been a factor in the
previous appeal of iterative price update algorithms, including, in
the Computer Science literature, the recent tatonnement algorithm of
Codenotti et al.~\cite{CodenottiMV05} and the auction algorithms of
Garg et al.~\cite{GargKapoor04}.

Our proposed price update protocols capture some of the
characteristics of trading as proposed in the trading post
model \cite{Balasko},
features that are lacking from previous algorithms subject to
asymptotic analysis. Namely, our algorithm consists of price updates
satisfying the following criteria: the price update for a good
depends only on the price, demand, and supply for that good, and on
no other information about the market; in addition, the price update for each
good occurs distributively and asynchronously
(i.e.\ at independent times).
Further desirable features are that: the algorithm can
start with an arbitrary set of prices; the algorithm tolerates inaccuracy in
demand data; finally, it can tolerate discreteness: the fact that neither
goods nor money are infinitely divisible.
We show that our update
protocols converge quickly in many
markets that satisfy the weak gross substitutes property. In the
process, we identify several natural parameters characterizing these markets,
parameters which govern the rate of convergence.

\subsection{The Market Problems}
\label{sec:problem}

{\bf The One-Time Fisher Market}\footnote{The market we describe here is
often referred to as the
Fisher market.  We use a different term because we
consider this problem in a new Ongoing Market model as described below.}
A market comprises a set of goods
$G$, with $|G| = n$, and two sets of agents, buyers $B$, with $|B|=m$, and sellers $S$.
The sellers bring the goods $G$ to market and the buyers bring money with which to buy them.
The trade is
driven by a collection of prices $p_i$ for good $i$, $1\le i\le n$.
For simplicity, we assume that there is a distinct seller for each good;
further it suffices to have one seller per good.
The seller of good $i$ brings a supply $w_i$ of this good to market.
Each seller seeks to sell its goods for money at the prices $p_i$.

Each buyer $b_j\in B$ comes to market with money $v_j$;
buyer $b_j$ has a utility function $u_j(x_{1j},\cdots,x_{nj})$
expressing its preferences: if $b_j$ prefers a basket with $x_{ij}$
units (possibly a real number) of good $i$, to the basket with
$y_{ij}$ units, for $1\le i\le n$, then $u_l(x_{1j},\cdots,x_{nj})>
u_j(y_{1j},\cdots,y_{nj})$. Each buyer $b_j$ intends to buy goods so
as to achieve a personal optimal combination (basket) of goods given
their available money.

Prices ${\bf
p}=(p_1,p_2,\cdots,p_n)$ are said to provide an \emph{equilibrium}
if, in addition, the demand for each good is bounded by the supply:
$\sum_{j=1}^m x_{ij} \le  w_i$, and the total cost
of the goods is bounded by the available money:
$\sum_i p_i w_i \le \sum_j v_j$. The market problem is
to find equilibrium prices.\footnote{Equilibria exist under quite
mild conditions (see \cite{MWG95} \S 17.C, for example).}

\medskip

The Fisher market is a special case of the more general
 \emph{Exchange} market or \emph{Arrow-Debreu}  market.

While we define the market in terms of a set of buyers $B$, all that matters for
our algorithms is the aggregate demand these buyers generate, so we will
tend to focus on properties of the aggregate demand rather than properties
of individual buyers' demands.
This will prove significant when we consider the case of
(somewhat) indivisible goods.

\medskip \noindent {\bf Standard notation}~
$x_i=\sum_l x_{il}$
is the demand for good $i$, and $z_i=x_i-w_i$ is the excess demand
for good $i$ (which can be positive or negative).
Note
that while $w$ is part of the specification of the market, $x$
and $z$ are functions of the vector of prices as determined by individual buyers maximizing their
utility functions given their available money.
We will assume that $x_i$ is a function of the prices $p$, that is a set of prices
induce unique demands for each good.

\Xomit{
We follow standard practice\footnote{See Varian~\cite{Varian} \S
21.5.} and view the actions of individual buyers and sellers as
being encapsulated in the price adjustments for each good.  More
specifically, we imagine that there is a separate, ``virtual'' price
setter for each good in the market.  Henceforth, for ease of
exposition, we describe price setters as if they were actual
entities, although in reality they are virtual entities induced by
agents' trades.
}

\smallskip\noindent
{\bf The Ongoing Fisher Market}~ In order to have non-equilibrium trade, we
need a way to allocate excess supply and demand. To this end, we
suppose that for each good there is a \emph{warehouse} which can
store excess demand and meet excess supply.   The seller has a warehouse of
finite capacity to enable it to cope with fluctuations in demand. It
will change prices as needed to ensure its warehouse neither
overfills nor runs out of goods.

The market consists of a set $G$ of $n$ goods and a set $B$ of $m$
buyers.
The market repeats over an unbounded number of time intervals called
days. Each day, the seller of good $i$ (called seller $i$) receives
$w_i$ new units of good $i$, and buyer $\ell$ is given $v_\ell$
money, $1\le \ell \le m$. As before, each buyer $\ell$ has a utility
function $u_\ell(x_{1\ell},\cdots,x_{n\ell})$ expressing its
preferences. Each day, buyer $\ell$ selects a maximum utility basket
of goods $(x_{1\ell},\cdots,x_{n\ell})$ of cost at most $v_\ell$.
Each seller $i$ provides the demanded goods $\sum_{\ell=1}^m
x_{i\ell}$. The resulting excess demand or surplus, $\sum_{\ell=1}^m
x_{i\ell}-w_i$, is taken from or added to the warehouse stock.
Seller $i$ has a warehouse of capacity $c_i$.

Given initial prices $p^\circ_i$, warehouse stocks $s^\circ_i$,
where $0<s^\circ_i<c_i$, $1\le i\le n$, and ideal warehouse stocks
$s^*_i$, $0<s^*_i<c_i$, the task is to repeatedly adjust prices so
as to converge to equilibrium prices with the warehouse stocks
converging to their ideal values. We let $s_i$ denote the current
contents of warehouse $i$, and $s^e_i=s_i-s^*_i$ denote the
\emph{excess warehouse reserves}.

\Xomit{
The difficulty with the problem as stated is that the initial prices
could be arbitrarily low and hence demand arbitrarily high, thereby
causing the seller(s) to run out of stock.
}

We suppose that it is the sellers that are adjusting the prices of their goods.
In order to have progress, we require them to change prices at least once a day.
However, for the most part, we impose no upper bound on the  frequency of price changes.
Indeed, we will see that having more frequent price changes is helpful in
the event that very low prices are present.
This entails measuring
demand on a finer scale than day units. We take a very simple
approach: we assume that each buyer spends their money at a uniform
rate throughout the day. (Equivalently, this is saying that buyers
with collectively identical profiles occur throughout the day,
though really similar profiles suffice for our analysis.)~
Likewise, if one supposes there is a limit to the granularity, this
imposes a limit on the frequency of price changes.
\Xomit{
and as we will see, also imposes a limit on how extreme the initial prices can be for
convergence to be assured.
}

\smallskip\noindent {\bf Market Properties}~ In an effort to capture the
distributed nature of markets and the possibly limited knowledge
of individual  price setters, we impose several constraints on procedures
we wish to consider:

\begin{enumerate}
\item
Limited information: the price setter for good $i$ knows
only the price, supply, and excess demand of good $i$, both current
and past history.  Thus the price updates can depend on this
information only. Notably, this precludes the use not only of other
prices or demands, but also of any information about the specific
form of utility functions.

\item
Simple actions: The price setters' procedures should be simple.

\item
Asynchrony: Price updates for different goods are allowed to be
asynchronous.

\item
Fast Convergence: The price update procedure should converge quickly
toward equilibrium prices from any initial price vector.

\item
Robustness to Inaccuracy: Even if the demand data is somewhat inaccurate,
the procedure should still converge to approximately the
equilibrium prices.

\end{enumerate}
We call procedures that satisfy the first three constraints
\emph{local}, by contrast with centralized procedures that use more
complete (global) information about the market.

Next, we discuss the motivations for these constraints.

Constraint (1) stems from the plausible assertion that not
everything about the market will be known to a single price setter.
While no doubt some information about several goods is known to a
price setter, it is a conservative assumption to assume less is
known, for any convergence result carries over to the broader
setting.  Further, it is far from clear how to model the broader
setting.

Constraint (2), simplicity, is in the eye of the beholder. Its
presence reflects our view that without further information, this is
both generally applicable and plausible.

Constraint (3), asynchrony, is an inherent property of independent
price adjustments.  Since the price setter of good $i$ reacts only
to trade in good $i$, the price adjustment of good $i$ occurs
independently of other price adjustments.

Constraint (4) arises in an effort to recognize the dynamic nature
of real markets, which are subject to changing supplies and demands
over time. However, surely much of the time, markets are changing
gradually, for otherwise there would be no predictability. A natural
approximation is to imagine fixed conditions and seek to come close
to an equilibrium in the time they prevail --- hence the desire for
rapid convergence.

Constraint (5) recognizes the reality that in practice data is always is a little dated
and inaccurate, and it seems more realistic if procedures tolerate this.

A final issue is that it is an approximation to allow goods and money to be infinitely
divisible. It is useful to understand to what extent this limits
convergent behavior. This entails a variety of changes to the
standard definitions used for continuous (divisible)
markets: it turns out it is no longer
sensible to constrain demands (or utilities) at the individual level,
but only in the aggregate, and the notion of weak gross substitutes needs
to be modified (relaxed in fact), but not in the way used in the literature
for (discrete) matching markets.
We call this modified setting the indivisible goods market, but we are intending
a setting in which there are still many copies of each good, rather than a setting
in which each good is unique or close to unique.

\subsection{Previous Work}
\label{sec:previous} To the best of our knowledge, asynchronous
price update algorithms have not been considered previously.
Further, there has been no complexity analysis of even synchronous
tatonnement algorithms with this type of limited information. While
Uzawa \cite{Uzawa60} gave a synchronous algorithm of this type, he
only showed convergence, and did not address speed of convergence.

The existence of market equilibria has been a central topic of
economics since the problem was formulated by Walras in 1874
\cite{Walras1874}.
Tatonnement was described more precisely as a differential equation
by Samuelson~\cite{Samuelson47}:
\begin{equation}
\label{eq:tatondif} dp_i/dt =\mu_i z_i.
\end{equation}
The $\mu_i$ are arbitrary positive constants that represent rates of
adjustment for the different prices; they need not all be the same.
Arrow, Block, and Hurwitz, and Nikaido and
Uzawa~\cite{ArrowBH59,ArrowH58,NikaidoUzawa60} showed that for
markets of gross substitutes the above differential equation will
converge to an equilibrium price.

Unfortunately, for general utility functions (i.e. that do not lead
to gross substitutability), the equilibrium need not be stable and
the differential equation (and thus also discretized versions) need
not converge~\cite{Scarf60}.  Partly in response, Smale described a
convergent procedure that uses the derivative matrix of excess
demands with respect to prices~\cite{Smale76}.  Following this,
Saari and Simon~\cite{SaariSimon78} showed that any price update
algorithm which uses an update that is a fixed function of excesses
and their derivatives with respect to prices needs to use
essentially all the derivatives in order to converge in all markets.
However, this is viewed as being an excessive amount of information,
in general.

There are really two questions here.  The first is how to find an
equilibrium, and the second is how does the market find an
equilibrium. The first question is partially addressed by the work
of Arrow et al.\ and Smale, and addressed further in papers in
operations research (notably Scarf~\cite{Scarf69} gives a
(non-polynomial) algorithm for computing equilibrium prices), and
theoretical computer science, where there are a series of very nice
results demonstrating equilibria as the solutions to convex
programs, or describing combinatorial algorithms to compute such
equilibria exactly or approximately. (An early example of a
polynomial algorithm for computing market equilibria for restricted
settings is~\cite{DevanurPSV02}.  An extensive list of references is
given in the
surveys~\cite{Codenotti-AGT,Vazirani-AGT}.)

We are interested in the second question.  The differential
equations provide a start here, but they ignore the discrete nature
of markets: prices typically change in discrete increments, not
continuously. In 1960, Uzawa showed that there is a choice of $\eps$
for which an obvious discrete analog of (\ref{eq:tatondif}) does
converge~\cite{Uzawa60}. However, determining the right $\eps$
depends on knowing properties of the matrix of derivatives of demand
with respect to price, or in other words, this requires global
information.

In relatively recent work, three separate groups have proposed three distinct
discrete update algorithms for finding equilibrium prices and showed
that their algorithms converge in markets of gross
substitutes~\cite{Kitti2010,CrockettSS04,CodenottiMV05}.  However, all
of these algorithms use global information.  With the exception
of~\cite{CodenottiMV05}, none of this work gives (good) bounds on
the rate of convergence. The algorithm in Codenotti et
al.~\cite{CodenottiMV05} describes a tatonnement algorithm (albeit
not asynchronous); however, it begins by modifying the market by
introducing a fictitious player with some convenient properties that
capture global information about the market and have a profound
effect on market behavior. Even in this transformed setting, the
price update step uses a global parameter based on the desired
approximation guarantee, and starts with an initial price point that
is restricted to lie within a bounded region containing the
equilibrium point. Translating their algorithm back into the real
market, one can see that it does not meet our definition of
simplicity or locality.

There are some auction-style algorithms for finding approximate
equilibria which also have a distributed flavor but depend on buyer
utilities being separable over the set of goods~\cite{GargKapoor04,
GargKapoorV04}. However, these algorithms are not seeking to explain
market behavior and not surprisingly do not obey natural properties
of markets.\footnote{These algorithms start their computation at a
non-arbitrary set of artificially low prices; global information is
used for price initialization; and they work with a global
approximation measure
--- each price update uses the goal approximation guarantee in its
update.}

The one work in the Computer Science literature considering the
indivisible goods setting is by Deng et al.~\cite{papadeng02}; this paper
was also the one to start the study of market equilibria in the CS
theory community. They showed that it is APX-Hard to approximate
equilibrium prices and allocations for indivisible exchange markets.
They also gave an exhaustive algorithm for computing an approximate
equilibrium in polynomial time for markets with a constant number of
distinct goods.

Markets of indivisible goods have been studied by mathematical economists
also. Ausubel, Gul and Stacchetti \cite{ausubel05} introduce an ``individual
substitutes'' property so as to ensure that equilibria exist.
Milgrom and Strulovici \cite{milstru06} also consider this setting, replacing
the individual substitutes constraint with the standard WGS
constraint.
\Xomit{As an equilibrium may then not exist they propose a
notion of pseudo-equilibrium; they argue that a pseudo-equilibrium
price is also an approximate equilibrium price.}
Both works present
exhaustive algorithms to find the equilibrium prices.

The design and analysis of procedures and convergence to equilibria
has been a recent topic of study for game theoretic problems, using
the technique of regret minimization in particular \cite{roughgarden,
kleinberg,even-dar,blum}.
Other work has studied convergence in network routing
and network design games
\cite{chien_sinclair,fischer_racke,goemans_vahab,vahab_2}.
In partial contrast, it is known that finding equilibria via local
search (e.g., via best response dynamics) is PLS-complete in many
contexts~\cite{JPY88,FabrikantPT04}. Recently, Hart and
Mansour~\cite{HartMansour07} gave communication complexity lower
bounds to show that in general games, players with limited
information require an exponential (in the number of players) number
of steps to reach an equilibrium.

The design
and analysis of convergent asynchronous distributed protocols has also
arisen in network routing, for example \cite{low},
and in high latency parallel computing
\cite{bertsekas}.
These lines of work are perhaps the most similar to ours.

\subsection{Our Contribution}

As  devising a tatonnement-style price update in unsolvable for general markets,
our goal has been to devise plausible constraints that enable rapid convergence.
Our overall task has been to devise a reasonable model, a price update
algorithm, a measure of closeness to equilibrium, and then to analyze
the system to demonstrate fast convergence. This also entails
identifying appropriate parameterized constraints on the market.

Loosely speaking, our measure computes the cost of resource misallocation
(compared to the equilibrium allocation).
Roughly, we show that the price
update process runs in time\\
$O(\frac{1}{\kpi}\log(\mbox{initial cost}/\mbox{final cost}))$,
where $\kpi$ is a parameter of the price update protocol,
a parameter that depends on the properties of the market demand.
Further, in the
case of indivisible goods, we give nearly tight bounds on the achievable minimum for the final
cost.
We also bound the needed warehouse size in terms of the initial (imbalanced)
state of the market.

A precursor of this work \cite{cole-fleischer} introduced the Ongoing Market and analyzed
it in the context of divisible markets.  The analysis given here is quite distinct and
conceptually simpler; it also yields improved bounds in terms of the parameter constraints.
(It is a little difficult to compare the two works, as the earlier work required more constraints
and parameters, and used slightly different rules for the price updates --- the
rules in the present paper strike us as more realistic.)

An initial analysis of markets of indivisible goods appeared in Rastogi's thesis \cite{ashish-thesis}.

The main technical tool for the analysis is the use of a potential function, whose value
corresponds to the market resource misallocation. We show the potential decreases continuously
at or above a suitable minimum rate, and never increases.

This work identifies three parameters which govern the rate of convergence and the needed
warehouse sizes.
The first parameter, which we call the \emph{market demand elasticity}, or the
\emph{market elasticity} for short, governs the rate of convergence of the
one-time markets: it is a lower bound on the fractional
rate of change of the demand for any good with respect to its own price: $\min_i \frac{dx_i}{dp_i}/\frac{x_i}{p_i}$.

The second parameter is the standard elasticity of wealth\footnote{This is the the fractional rate
of change of demand with respect to changes in wealth; it is usually defined for individual agents;
we will use an aggregate definition, obtained when all buyers have the same fractional change in
wealth, or equivalently all prices change by the same fraction.} it will
be used in the analysis of convergence rates when rapid price rises are used
to mitigate the effects of large demands.

The third parameter is somewhat less intuitive; it governs the needed warehouse sizes.
We call it the \emph{equilibrium flex}, and it is defined as follows.
Let $p_i^{(c)}$ denote the equilibrium prices for supplies $cw_i$, for $1\le i \le n$.
Then we define the equilibrium flex, $e(c)$, by
\[
e(c) = \ln \max_i \left\{ \frac {p_i^*} {p_i^{(c)}}, \frac {p_i^{(1/c)}} {p_i^*} \right\}
\]
where $p^*_i = p_i^{(1)}$ are the equilibrium prices for the market at hand.
We remark that the carry-forward of demand imbalances in Ongoing Markets will create
temporary changes in effective supply levels and hence in the desired equilibrium prices.
It is plausible that the size of these changes affects the convergence rate and not surprising that it
affects the needed warehouse size also.

Bounds on $e(c)$ follow from bounds on the elasticity of wealth, but it
is not clear that this is a tight connection.

For a market of buyers all having CES utilities, the values of all these parameters
are modest, as we will see.

\section{Results}

\subsection{The 1-Good Case in a One-Time Market}

We begin by considering a tatonnement process in the context of
a one-time Fisher market with a single good,
in order to determine reasonable constraints on the demand function,
namely constraints that enable a fairly rapid convergence of the price
toward its equilibrium value.

Let $x$ denote the demand for the good and $p$ its price.
We assume that $x(p)$ is a strictly decreasing function.
Let $w$ denote the supply of the good.
Our goal is to update $p$ repeatedly so as to cause $x$
to converge toward $w$.
We will use the following update rule:
\[
\label{eqn:p-update}
p'=p \left(1+ \lmbd \min\left\{ 1, \frac{x-w}{w} \right\} \right)
\]
where $0 < \lambda \le \frac12$ is a suitable fixed parameter.

We are going to argue that this update function is appropriate
for demands $x$ satisfying:
\begin{eqnarray}
\label{eqn:deriv-bound}
\frac{x(p)}{p} \le -\frac{dx}{dp} \le E\frac{x(p)}{p},
       ~~\mbox{where } 1\le E\le \frac{1}{2\lambda}.
\end{eqnarray}
We will also explain why the update function
can be viewed as an approximation of Newton's method,
but an approximation quite unlike the standard secant method.

\begin{definition}
\label{def:same-side}
Let $x'$ denote $x(p')$.
Price $p'$ is said to be on the \emph{same side} as $p$ if either both $x,x' \ge w$
or both $x,x' \le w$. We also say that $p'$ is a \emph{same-side update}.
\end{definition}
Let $p^*$ denote the value of $p$ for which $x(p^*)=w$.
Note that $p'$ is on the same side as $p$ exactly if both $p,p'\ge p^*$ or if both $p,p'\le p^*$.

While same side updates are not important in the 1d setting, in the  multi-dimensional
setting they are needed to ensure progress.
And so in the next lemma we investigate how large a step size is possible
while having a same side update.
\begin{lemma}
\label{lem:same-side-update}
Suppose that $-p\frac{dx}{dp} \le Ex(p)$, for some $E \ge 1$.
Then
the update $p'=p (1+ \frac{1}{2E} \min\{ 1, \frac{x-w}{w} \} )$ is
on the \emph{same side} as $p$,
but the update $p'=p (1+ \frac{1}{E} \frac{x-w}{w} )$ need not
be.
\end{lemma}

Next we investigate the rate of convergence.
It might seem natural to determine how quickly $|x^* - x(p)|$
decreases, but in higher dimensions, this proves a less effective measure.
Instead, we will measure the ``distance'' from equilibrium by the measure, or potential function,
$\phi(x,p)=|x-w|p$.
This can be thought of as the cost in money of the current amount of disequilibrium.

\begin{theorem}
\label{lem:1d-rate-of conv}
If $x\le 2w$ initially, then in one iteration $\phi$ reduces by at least $w|p'-p| = \lambda \phi$,
where $p$ is the price before the update and $p'$ the price after the update.
Hence in $O(\frac{1}{\lmbd}\log\frac{\phi_I}{\phi_F})$
iterations, $\phi$ reduces from $\phi_I$ to at most $\phi_F$.
If $x\ge 2w$ initially,  then in one interation $\phi$ reduces by at least $w|p'-p| = \lambda \frac{w}{x-w}\phi$.
Hence in $O(\frac{1}{\lmbd} \frac{x-w}{w}\log\frac{\phi_I}{\phi_F})$
iterations, $\phi$ reduces from $\phi_I$ to at most $\phi_F$.
\end{theorem}

\subsection{The Multidimensional Case, Synchronous Updates, One-Time Market}

We start by investigating the multidimensional case in the event that all prices are updated simultaneously
as the analysis is relatively simple.
We will see that the rate of convergence is very similar to that which was obtained in the 1d case.

Recall that there are $n$ goods $g_1, g_2,  \cdots , g_n $.
Let $p= (p_1 ,p_2 , \cdots,p_n)$ be the $n$-vector of prices
for these goods,
and let  $x_i(p)$, or  $x_i$ for short,
denote the demand for the $i$th good at prices $p$
(that is, it is a function of all the prices,
not just the price of the $i$th good).

Again, we use the update rule
\[
\label{eqn:p-update-i}
p_i'=p_i \left(1+ \lmbd \min\left\{ 1, \frac{x_i-w_i}{w_i} \right\} \right).
\]
Again, we assume that the demands satisfy the bounded rate of change property
which we term \emph{bounded elasticity} in the current context.
\begin{definition}
\label{def:bdd-elas}
Good $i$ has \emph{bounded elasticity} $E$, if for any collection $p$ of prices:
\[
\frac{x_i(p)}{p_i} \le -\frac{dx_i}{dp_i} \le E\frac{x_i(p)}{p_i},
       ~~\mbox{where } 1\le E\le \frac{1}{2\lambda}.
\]
The market has \emph{(demand) elasticity} $E$ if every good has bounded elasticity $E$.
\end{definition}

We are thinking of $E$'s value as a constant.
To motivate this, we give the value of $E$ for some standard utility functions:
if all buyers have Cobb-Douglas utilities \cite{MWG95} (p.~612), then $E=1$; if all buyers have CES
utilities \cite{MWG95} (p.~97) with parameter $\rho$, $0 \le \rho \le 1$, then $E=1/(1 - \rho)$.

We will also need the standard assumption of Weak Gross Substitutes (WGS) to somewhat constrain
how one price change affects demands for other goods.

\begin{definition}
\label{def:WGS}
A market satisfies the \emph{gross substitutes property} \emph{(}GS\emph{)} if for
any good $i$, increasing $p_i$ leads to increased demand for all
other goods. The market satisfies \emph{weak gross substitutes} if
the demand for every other good increases or stays the same.
\end{definition}

We use the potential function $\phi = \sum_i \phi_i $,
where $\phi_i = p _i |x _i - w _i| $.
So the effect of an update to $p_i$ on $\phi_i$ is just as in the 1-d case.
However the update may also affect the demands $x_j$,
and hence the potentials $\phi _j$, for goods $j \ne i$.

\begin{theorem}
\label{cor:sim-progress}
Suppose that the market obeys WGS and has elasticity $E$.
When all the prices are updated simultaneously, if $\lmbd(2E-1) \le \frac12$,
$\phi=\sum_i \phi_i$ reduces by at least
$\sum_i \lmbd \phi_i^I\min\left\{1, \frac{w_i}{|x_i - w_i|}\right\}$.
If the demands satisfy $x_i \le dw_i$ throughout, then
$\phi$ reduces from $\phi_I$ to at most $\phi_F$
in $O(\frac{1}{\lmbd} (d-1)\log\frac{\phi_I}{\phi_F})$
iterations.
\end{theorem}

Bounding $d$ is non-trivial, for an update to $x_j$, $j  \ne i$, may cause $x_i$ to
increase.  We return to this issue later.

\subsection{The Multidimensional Case, Asynchronous Updates, One-Time Market}

But we don't want to assume that all the price updates are simultaneous.
For this to make sense we have to reinterpret the meaning of the parameters for supply and demand.
We now view these as rates.

To make this concrete, let us call the basic time unit a day.
$w_i$  is the daily supply of good $i$, which is assumed to arrive continuously throughout the day.
$x_i$ will be the instantaneous rate of demand given the current prices.
The demand over some time period  $[t_1 , t_2] $ is given by $\int _{t_1 } ^{t_2 } x_i (t) dt $.
 We are interested in scaling this demand by the time period;
we call this the average demand, denoted  ${\overline x}_i [ t_1 , t_2 ] $;
i.e., ${\overline x}_i [ t_1 , t_2 ] = \frac {1} {t_2 -t_1 }\int _{t_1 }^ {t_2 } x_i (t)dt$.
We are particularly interested in the average demand for good $i$ since the time of last update to $p_i$,
at time $\tau_i$ say:
we define the average demand at the current time by
${\overline x}_i = {\overline x}_i [\tau_i ,  t] =\frac {1} {t_2 -t_1 }\int _{t_1 } ^{t_2 } x_i (t) dt $.
As we will see later when we introduce warehouses,
the average demand can be measured in our setting and
so this is the value for the demand that will be used in the price update function:
\[
p_i'=p_i \left(1+ \lmbd \min\left\{ 1, \frac{{\overline x}_i - w_i}{w_i} \right\} \right).
\]
To ensure progress we require that each price updates at least once a day.

We use the following potential function: $\phi=\sum_i \phi_i$, where
\begin{eqnarray}
\label{eqn:phi-asynch-def}
\phi_i(x_i,{\overline x}_i,w_i) =
p_i \left[\span(x_i,{\overline x}_i,w_i) -\alpha_1 \lmbd|w_i - \xbi|(t - \tau_i)\right],
\end{eqnarray}
with $\alpha_1  > 0 $ being a suitable constant and
$\span(x,y,z)$
denoting the length of the interval spanned by its arguments, i.e.\
$\max\{x,y,z\} - \min\{x,y,z\}$.

We now show a bound on the convergence rate which is broadly the same as the one that applies when the updates are synchronous.
\begin{theorem}
\label{lem:progr-bdd-dem}
Suppose that the market obeys WGS and has elasticity $E$.
If $x_i \le d w_i$, for all $i$, where $d \ge 2$,
$\alpha_1 (d-1) \le 1$,
$\loa + \lambda \left(1 + \frac{2Ed}{1-\lmbd E } \right)\le 1$,
and each price is updated at least once every day,
then $\phi$ decreases by at least a $1 - \frac{\loa}{2}$ factor daily.
Hence $\phi$ reduces from $\phi_I$ to at most $\phi_F$
in $O(\frac{1}{\loa} (d-1)\log\frac{\phi_I}{\phi_F})$
days.
\end{theorem}

In the current setting, the natural measure of misspending on good $i$ is
$|\xbi -w_i|p_i$, for as we will see when we consider warehouses, the difference
$|\xbi -w_i|$ is what we can measure, and this is used as the indicator of how far from equilibrium $p_i$ currently is,
and hence how to update it.
However, we also have to take account of the current rate of excess demand,
$|x_i -w_i|$. Thus we define the misspending $S$ to be
\[
S =  \sum_i [|x_i(t') - w_i| +|\xbi -w_i|]p_i.
\]
Clearly, $ \phi =\theta( S)$.

\Xomit{
\begin{lemma}
\label{lem:missp-pot}
$ \phi =\theta( S)$.
\end{lemma}
}

We can conclude that, while it is not necessarily decreasing at all times,
over time the misspending decreases at the same rate as $\phi$.

\subsection{The Ongoing Market, or Incorporating Warehouses}
\label{sec:intro-war}

For the market model to be self-contained, we need to explain how excess
demand is met and what is done with excess supply. The solution is simple:
we provide finite capacity warehouses (buffers in computer science terminology)
that can store excess supply and meet excess demand.
There is one warehouse per good.
The price-setter for a good changes
prices as needed to ensure the corresponding warehouse
neither overfills nor runs out of goods.

Just as the demand is a rate, we imagine the supply to be a rate, which
for the purposes of analysis we treat as being a fixed rate.
Each $\Del t$ long instant, the resulting excess demand or surplus, $(x_i-w_i)\Del  t$,
is taken from or added to the warehouse stock.

Let $c_i$ be the capacity of the warehouse for good $i$.
Each warehouse is assumed to have a target ideal
content of $s^*_i$ units (perhaps the most natural value is
$s^*_i=c_i/2$).

The goal is to repeatedly adjust prices so as to converge to
near-equilibrium prices with the warehouse stocks converging to
near-ideal values.
A further issue is to determine what size warehouse suffices,
which we defer to a later section.

The price update rule needs to take account of the current state of the
warehouse, namely whether it is relatively full or empty.
To this end, let $\tau_i$ be the time of
the previous update to $p_i$. let $t$ be the current time,
and let $s_i$ denote the current contents of
warehouse $i$.
Then the target excess demand, $\overline z_i$,
is given by $\xbi [\tau_i, t]
- w_i  + \kappa_i (s_i - s_i^*)$ where $\kappa_i > 0$ is
a suitable (small) parameter.
So $\overline z_i = \frac{s_i (\tau_i) - s_i
(t)}{t - \tau_i} + \kappa_i (s_i - s_i^*)$ and is readily calculated by
monitoring warehouse stocks.
We let $\wti$ denote $w_i - \kappa_i (s_i - s_i^*)$,
which we call the \emph{target demand}.

We will need the following constraint on $\kappa_i$.
\begin{constraint}
\label{cst:wti-bdd}
$\kappa_i (s_i -s^*_i) = |\wti - w_i| \le \frac13 w_i$.
\end{constraint}

For simplicity, henceforth we assume that $\kappa=\kappa_i$ for all $i$.

The price of good $i$ is updated according to the following rule:
\begin{equation}
\label{eqn:pr_change}
p'_i \leftarrow  p_i \left(1 +
     \lambda~ \mbox{median}\left\{-1, \frac{ \overline z_i
         (p)}{w_i}, 1 \right\}
   \right)
\end{equation}
This rule ensures that the change to $p_i$ is bounded by
$\pm \lmbd p_i$.

We redefine the potential $\phi_i$ to also take account of the imbalance in the warehouse
stock as follows:
\[
\phi_i = p_i[\span(x_i, \xbi, \wti) -
\loa (t - \tau_i) | \xbi- \wti |
+ \alpha_2 | \wti - w_i |]
\]
where $1< \alpha_2 < 2$ is a suitable constant;
this is simply Equation \ref{eqn:phi-asynch-def}, with
$\wti$ replacing $w_i$ and with the additional term
$ \alpha_2 p_i | \wti - w_i |$.

We redefine the misspending $S$ to take account of the warehouse imbalance:
\[
S =  \sum_i [|x_i(t') - w_i| +|\xbi -w_i|]p_i. + | \wti(t') - w_i |] p_i .
\]
Again, $ \phi =\theta( S)$.

\begin{theorem}
\label{lem:war-progr-bdd-dem}
If Constraint \ref{cst:wti-bdd} holds and $x_i \le d \wti$,
$\alt = \frac 32$, $\alo = \frac{1}{16}$, $\lE \le \frac{1}{17}$, $\lE d \le \frac{5}{17}$,
$\lmbd \le \frac{1}{14}$, $\kpi \le \frac{\loa}{10}$,
\Xomit{
$\frac{\alt}{2} + \alpha_1  \max\{\frac32,(d-1)\} \le 1$,
$\loa + \frac43 \lmbd  \left(1 + \frac{2E d}{1-\lE}  + \frac12 \alt \right) \le 1$,
$ 4\kpi(1 + \alt) \le \loa \leq \frac{1}{2}$,
$\frac{\kmin(\alt -1)}{2} \le 1$,
}
and each price is updated at least once every day,
then $\phi$ decreases by at least a $1 - \frac{\kmin(\alt -1)}{4}$ factor daily.

In fact, if
$\phi \ge 2(1 + 2\alt) \sum_i |\wti - w_i| p_i$,
then $\phi$ decreases by at least a $1 - \frac{\loa}{8(1 +\alt)}$ factor daily.
\end{theorem}

Once the potential is largely due to the warehouse imbalances, the daily rate
 of decrease of $\phi$ drops from $1 - \theta(\loa)$ to $1 - \theta(\kpi)$,
which seems unavoidable as the improvement to the warehouse contents may be
only $w_i$, which contributes an $\alt\kpi w_i p_i$ reduction to $\phi_i$, but $\phi_i$
could be of size $\theta(w_i p_i)$.

We do not mean to suggest that the above values for the parameters are tight or even nearly tight.
Rather the result should be understood as indicating the order of magnitude rate of convergence.

\subsection{Bounds on Prices and Demands}
\label{sec:bounds-using-w-elas}

Here, we
determine an bounds on prices and demands that hold
throughout, given initial bounds on demands, and hence what value of $d$
can be used in Theorem{lem:war-progr-bdd-dem}.

\begin{definition}
\label{def:bdd-prices-wrhs}
Prices are $f$-bounded if $p_i$ always remains in the range $p^*_ie^{\pm f}$ for all $i$,
where $p^*_i$ is the equilibrium value for $p_i$,
and demands are $d$-bounded if $x_i \le d w_i$ for all $i$.
\end{definition}

\begin{definition}
\label{def:var-dem-equil}
$p^{(c)}$ denotes the equilibrium prices for supplies $c w_i$.
\end{definition}

\begin{definition}
\label{def:dmd-bdd-prices}
Prices $p$ are \emph{$c$-demand bounded} if $p_i\in [p^{(c)}_i, p^{(1/c)}_i] $
for all $i$.
The equilibrium flex, $e(c)$, is defined to be
\[
e(c) = \ln \max_i \left\{ \frac {p_i^*} {p_i^{(c)}}, \frac {p_i^{(1/c)}} {p_i^*} \right\}
\]
where $p^*_i = p_i^{(1)}$ are the equilibrium prices for the market at hand.
Note that $c$-demand bounded prices are $e(c)$-bounded.
\end{definition}
For CES utilities (with parameter $\rho$) $e(c) = \ln c$ (it is independent of $\rho$).

\begin{lemma}
\label{prices-stay-dmd-bdd}
If $\frac{\lE}{1 - \lE} \le \frac 16$, and if
initially the prices have all been $c$-demand bounded for a full day
for some $c \ge 2$, they remain $c$-demand bounded thereafter.
\end{lemma}

\begin{definition}
\label{def:d-bound}
Suppose that the prices are always $c$-demand bounded.
Let $f=f(c)$ be the corresponding $f$-bound on the prices and given the prices
are $f$-bounded, let $d(f)$ be the demand bound \emph{(}conceivably, $d \gg c$\emph{)}.
\end{definition}

\begin{lemma}
\label{lem:weak-d-bdd}
$d(f) \le e^{2Ef}$.
\end{lemma}

It follows that given a $c$-demand bound for the first day of the update process, there is
a value for $f=f(c)$, implied by the uniqueness of equilibria in this setting, which in turn
yields a bound on $d(f)$.

Next, we obtain a bound on $e(c)$ given the assumption of 0-homegeneity\footnote{i.e.\ if the money of each agent
and prices all increase by a multiplicative factor $\gamma > 0$ this leaves demands unchanged.}
and a lower bound of 0 on the elasticity of wealth.

\begin{definition}
\label{def:elas-wealth}
Suppose that the money of each buyer increases by a multiplier $\gamma$.
We define $\xi^w_i$, the elasticity of good $i$ with respect to wealth, to be
$\xi^w_i=\frac{d x_i}{d \gamma}/x_i$.
We say that the market has a bounded wealth elasticity $E^w \ge 0$ if $\xi^w_i \ge -E^w$ for all $i$.
\footnote{It is usual to define the elasticity of wealth for each individual separately. It is not hard to see that if each
individual has wealth elasticity $E^w$ then this implies the $E^w$ bound in our definition.}
When $E^w = 0$, the
market is said to have \emph{normal} demands
\emph{(}see \cite{MWG95} page 25\emph{)}.
\end{definition}
For CES utilities, $E^w = -1$ (this is a stronger bound than is used in any of our analyses.)

\noindent
{\bf Notation}. Let $ \rho = \max_{i,j} \frac{w_ip_i^*}{w_jp_j^*} $.

\begin{lemma}
\label{lemma:price-lower-upper-bound}
If the demands are all normal, then
$e(c) \le \ln [c (\rho n)^{(c-1)}]$.
\end{lemma}

We can still obtain bounds on $e(c)$ for positive $E^w$, but they are much larger.

We note that even the bound of Lemma \ref{lemma:price-lower-upper-bound}  is large
and implies that a small change in supplies could have a huge
effect on the values of some of the equilibrium prices.
This seems implausible as a practical matter, and
as the bound of Lemma \ref{lem:weak-d-bdd}
and hence the convergence rate depend on the value of $e(c)$ it suggests
that our bounds based on the elasticity of wealth,
rather than the equilibrium flex, may be unduly pessimistic.

Prior work \cite{cole-fleischer} made the stronger assumption that $E^w < 0$ (though this was
expressed in a different way).

\subsection{Faster Updates with Large Demands, Ongoing Market}

As we have seen, when demands are large initially, the rate of convergence depends
inversely on a bound on a parameter
$d$ where for each $i$, $dw_i$ bounds the demand for good $i$.
As we will see the bound on the warehouse sizes also depends on $d$.

We now show how to avoid this dependence, by a plausible and modest change
to the frequency of the price updates.

It seems reasonable that when demands are large, the
seller will observe this quickly (due to stock being drawn from its warehouse)
and consequently will quickly adjust its price.
Accordingly, we introduce a new rule for the frequency of updates: in addition
to the once a day update, whenever $w_i$ units of good $i$ have been sold
since the last update, price $p_i$ is updated.

In addition, we will need the $E'$ bound on the elasticity of wealth.

We significantly modify our basic potential function for this analysis, as will be seen when the
analysis is carried out. Here the relationship between $S$ and $\phi$
is a bit looser: $ S  =O (\phi) = O(S +M)$, where $ M$ is the
daily supply of money.

\begin{theorem}
\label{lem:fst-updt-progr-bdd-dem}
If Constraint \ref{cst:wti-bdd} holds,
$d=5$, $\alt = \frac32$, $\lmbd (E + E') \le \frac{1}{17}$, $\alo \le \frac{1}{16}$,
$\loa + \frac43 \lmbd \left(\frac74 +\frac{10 E }{1-\lE}  \right) \le 1$,
$\kpi \le \frac{\loa}{13}$,
and each price is updated at least once every day,
then $\phi$ decreases by at least a $1 - \frac14 \kpi$ factor daily.
\end{theorem}

Again, we could prove a $1 -\theta(\loa)$ decrease when $\phi = \Omega(\sum_i w_ip_i)$.

\subsection{Bounds on Warehouse Sizes}

For simplicity, we will assume that $c_i/w_i$
is the same for all $i$.
Also, for simplicity, we suppose that $s_i^* = \frac12 c_i$,
that is the target fullness for each warehouse is half full.

We view each warehouse as having 8 equal sized zones of fullness, with the goal being to bring
the warehouse into its central four zones. The role of the outer zones is to provide the buffer
to cope with initial price imbalances.

\begin{definition}
\label{def:wrhs-zones}
The four zones above the half way target are called the \emph{high} zones, and the other four
are the \emph{low} zones.
Going from he center outward, the zones are called the \emph{safe} zone, the \emph{inner buffer},
the \emph{middle buffer}, and the \emph{outer buffer}.
\end{definition}

\begin{theorem}
\label{lem:good-wrhs}
Suppose that the prices are always $f$-bounded and let $d= d(f)$.
Also suppose that each price is updated at least once a day.
Suppose further that the warehouses are initially all in their safe or inner buffer zones.
Finally, suppose that $\lmbd \left( 1 + \frac{1} {\alf} \right) \le \frac12$,
where $\frac{\alf} = \frac{ \kpi c_i} {8w_i} $.
Then the warehouse stocks never go outside their outer buffers \emph{(}i.e.\ they never overflow
or run out of stock\emph{)} if
$\frac{\alf}{\kpi} = \frac{c_i}{8w_i}  \ge
  \max \left\{ (d-1) D, 2\left( 1 + \frac{4}{\alf} \right)f + \frac{8 \lmbd} {\alf} \right\}$;
furthermore,
after $D + 2\left( 1 + \frac{4}{\alf} \right) f + \frac{8 \lmbd} {\alf} + \frac{8}{\kpi}$
days the warehouses will be in their safe or inner buffer zones thereafter, where
\[
D = \frac{16(1 + \alt)}{\loa}\log\frac{\phi_{\text{init}}} {\frac{1 - \loa}{2}\min_i w_i p^*_i  },
\]
and $\phi_{\text{init}}$ is the initial value of $\phi$.

If the fast updates rule is followed, then it suffices to have
$\frac{\alf}{\kpi} = \frac{c_i}{8w_i}  \ge  2\left( 1 + \frac{4}{\alf} \right)f + \frac{8 \lmbd} {\alf} $,
and then after $\left( 1 + \frac{4}{\alf} \right) f + \frac{8 \lmbd} {\alf} \lmbd + \frac{8}{\kpi}$
days the warehouses will be in their safe or inner buffer zones thereafter.
\end{theorem}

\subsection{The Effect of  Inaccuracy}

Next, we investigate the robustness of the tatonnement process
with respect to inaccuracy in the demand data.
To mitigate the intricacy of the analysis, we return to the setting without fast updates.

Specifically, we assume that there may be an error of up to $\rho w_i$
in the reported values of $s_i(t)$ and $s_i(\tau_i)$,
where $\rho > 0$ is a constant parameter.
Recall that these are the values which are used to calculate
$\zbi$ ($= \frac{s_i (\tau_i) - s_i (t)}{t - \tau_i} + \kappa_i (s_i(t) - s_i^*)$).
Let $\zbi^c$ denote the correct value for $\zbi$, and
$\zbi^r$ the reported value.

To enable us to control the effect of erroneous updates, we will place a
lower bound on the frequency of updates to a given price. Specifically,
successive updates are at most 1 day apart (as before), and at least $1/b$
days apart, where $b \ge 1$ is a parameter.

We consider two scenarios:

\smallskip
\noindent
(i) The parameter $\rho$ is not known to the price-setters, who then
perform updates as before. \\
We show that for $\phi \ge\rho b^2 \frac{\lmbd}{\kpi}EM$,
$\phi$ reduces by a $(1 -\Theta(\kpi))$ factor daily.

\smallskip
\noindent
(ii) The parameter $\rho$ is known to the price-setter for each good $i$,
who performs an update only if the possible error is at most half of
the reported value $\zbi^r$, i.e.\ if $|\zbi^r - \zbi^c| \le \frac12 \zbi^r$.
\\
Then, we show that for $\phi \ge\rho bM $, $\phi$ reduces by a
$(1-\Theta(\kpi))$ factor daily.

It may seem more appealing to allow multiplicative errors
of up to $1 \pm \rho$ in the reported $s_i$.
However, this seems a little unreasonable in the case
that the actual $s_i(t) - s_i(\tau)$ is relatively small.
Also, later we will consider a scenario in which more frequent updates are required
when $s_i(t)$ changes rapidly,
and then a multiplicative rule for the error in the \emph{change} to the warehouse stock
would give an error no larger than the additive rule.

Contrariwise, one might argue that if the warehouse stock is changing only slightly,
then the perceived error ought to be small.
But once one considers that there is a daily supply of $w_i$ units of the $i$th good,
and that possibly the error reflects fluctuations in the selling of these $w_i$ units,
an error of $\pm \rho w_i$ seems reasonable.

Our analysis for Case (i) uses the potential function $\phi = \sum_i \phi_i$,
from Section \ref{sec:intro-war}.

\begin{theorem}
\label{lem:non-acc-progr-bdd-dem}
If $x_i \le d\wti$ for and $i$ and Constraint \ref{cst:wti-bdd} hold,
$\frac{\alt}{2} + \alpha_1  \max\{\frac32,(d-1)\} \le 1$,
$\loa + \frac43 \lmbd \left(1 + \frac{2E d}{1-\lE}  + \frac12 \alt \right) \le 1$,
$\frac{\kpi (\alt - 1)}{2} \le 1$,
$ 4\kpi(1 + \alt) \le \loa \leq \frac{1}{2}$,
and each price is updated at least once every day, and at most every $1/b$ days,
and if
$\phi \ge \frac{16 \mu M} {\kpi (\alt - 1)}
 \frac{1 -\loa}{1 - \loa - \mu}$
at the start of the day,
then $\phi$ decreases by at least a $1 - \frac{\kmin(\alt -1)}{8}$ factor by the end of the day,
where $M$ is the daily
supply of money
and $\mu = \frac43\lmbd \rho b (2b +\kmax)\left(1 + \frac{2E d}{1-\lE}  + \frac12 \alt \right)$,
supposing that $\kmin(\alt -1) \ge 16\mu/[1-\loa -\mu]$.
\end{theorem}

Case (ii) yields a less stringent constraint on $\rho$.
We  use a slightly different potential for which again $\phi(t) \le S(t) =O(\phi(t-1)$.

\begin{theorem}
\label{lem:non-acc-progr-bdd-dem-2}
If $x_i \le d\wti$ for and $i$ and Constraint \ref{cst:wti-bdd} hold,
$\frac{\alt}{2} + \alpha_1  \max\{3,2(d-1)\} \le 1$,
$\loa + 2\lambda \left(1 + \frac{2E d}{1-\lE}  + \frac12 \alt \right) \leq 1$,
$ 4\kpi(1 + \alt) \le \loa \leq \frac{1}{2}$ for all $i$,
$\frac{\kpi (\alt - 1)}{2} \le 1$,
$\mu \left[ \frac{1 + \mu/(1 - \loa)} {1 - \mu/(1 - \loa)} + \frac{1}{(1 - \loa)} \right]
\le \frac{\kpi(\alt - 1)}{2}$,
if $\phi \ge \frac{32\mu M} {[1 - \frac{\mu}{1 -\loa}] [(\kmin(\alt - 1)]}$
at the start of the day,
where $\mu = 8\kmax(1 + \alt) (2b +\kmax) \rho $,
and if each price is updated at least once every day,
then $\phi$ decreases by at least a $1 - \frac{\kmin(\alt -1)}{8}$ factor daily.
\end{theorem}

\noindent
{\bf Remark}. For both update rules in this section,
to obtain bounds on the needed warehouse sizes,
we need to use the fast update variant;
this is left to the interested reader.

\subsection{Discrete Goods and Prices}

Now we investigate the effect of only allowing
integer-valued prices and finitely divisible goods:
this is implemented by requiring each $w_i$
to be an integer and goods to be sold in integral quantities.

\smallskip
\noindent{\bf Demands as a Rate}.
With limited divisibility, we need to look again at our interpretation of demands
as a rate.

$x_i ({\bf p})$, the
{\em daily demand} for
good $i$ at prices ${\bf p}$, is simply the demand
were all the prices to remain unchanged over the course of a day.
The {\em ideal  demand} for good $i$ over time interval $[t_1, t_2]$
with unchanged prices is defined to be $x_i (p) (t_2 - t_1)$.
Ideal demand $x_2^I (t_1, t_2)$ for good $i$ over time
interval $[t_1, t_2]$ with possibly varying prices is given by
$\int_{t_1}^{t_2} x_i ({\bf p}) dt$. Note that this is in fact a sum
as there are only finitely many price changes.

The {\em actual demand} $x_i^A (t_1, t_2)$ over the time interval
$[t_1, t_2]$ is the supply minus the growth in the warehouse stock:
$w_i (t_2 - t_1) - [s_i (t_2) - s_i (t_1)]$.

In order to achieve approximately uniform demand as a function of prices,
we require that $| x_i^A (t_1, t_2) -  x_i^I (t_1, t_2)| < 1$, for all times
$t_1, t_2$ at which price $p_i$ is considered for an update (i.e.\
both actual and null updates).

We will also define ideal warehouse contents.
The ideal content of warehouse $i$, $s_i^I$ is simply the contents of the warehouse
had $x_i^I$ been the demand throughout. Note that $|s_i^A-s_i^I|=|x_i^A-x_i^I| <1$.
$\wti^I$ is defined in terms of $s_i^I$:
$\wti^I  = w_i - \kpi (s_i^I - s_i^*)$;
$\wti^A$ can be defined analogously.
Note that $|\wti^I - \wti^A| < \kpi$.
We  also define $\xbi^A= x_i - \wti^A$ and $\xbi^I = x_i - \wti^I$.
The computation of price updates uses $\xbi^A$.

\smallskip
\noindent{\bf Discrete WGS}.
We need to redefine WGS and the bounds on the rate of change in demand
w.r.t.\ prices so that we can carry out an analysis similar to that for the divisible case.

In the Fisher market context it is not hard to see that WGS
imposes the same constraints on the spending and the demand for each
good. This means that in the discrete setting under WGS, if the
price $p_i$ increases by one unit, then as the spending does not increase,
the demand for good $i$ must
drop, so if the demand at price $p_i$ is $x_i$, at price $p_i+x_i$
it must be zero. This seems unnatural. This impression is reinforced
by considering what happens were half units of money to be
introduced, with WGS remaining in place. Then at price $p_i+x_i/2$
the demand would have to be 0. This suggests that the property ought to
be modified in the discrete setting.

Accordingly, we define a market to satisfy the
\emph{Discrete WGS property} if, for any good $i$,
reducing its price $p_i$ to $p_i-\Delta$ only reduces demand for
all other non-money goods, and the spending on good $i$ is now at
least $p_ix_i(p_i)-[(p_i-\Delta)-1]$, i.e.\ the spending, if
reduced, is reduced by less than the cost of one item.
Note that it need not be that all the money is spent (for there may be left
over money which is insufficient to buy one item of any good).

\smallskip\noindent{\bf Elasticity of Demand and the Parameter $\elas$.}
We define the following bounded analog for discrete markets. Suppose
that the prices of all goods other than good $i$ is set to ${p}_{-i}$.
Then, for all $l_i \leq p_i \leq q_i$,
 \[ \left\lfloor x_i (l_i, p_{-i}) \left( \frac{l_i}{p_i}
   \right)^E \right\rfloor \leq x_i ({p}) \leq \left\lceil x_i
   (q_i, {p}_{-i}) \left( \frac{q_i}{p_i} \right)^E
 \right\rceil. \]

The crucial observation is that there is a fully divisible market with elasticity bound $2E$
that has demands $y_i$ very similar to those for the discrete market:
for every price vector $p$ which induces non-zero demand for every good,
for all $i$, $x_i(p) - 1 < y_i(p) \le x_i(p)$.
Given this correspondence,
the analysis of the discrete case becomes similar to that for inaccurate data.

\Xomit{
\smallskip
\noindent{\bf Sources of Error}.

As $|\xbi^I - \ybi| \le 1$, and $|\wti^A - \wti^I| \le \kpi$
the previous bound of  $2(b + \kpi)\rho w_i$
on the error in calculating $\zbi$ ($=\ybi -\wti$ here)
is replaced by $1 + \kpi$.
This amounts to setting $\rho = \min_i (1 + \kpi)/[2(b + \kpi) w_i]$.
Note that this implies that a price update occurs only if $|\xbi^A -\wti| \ge 2(1 + \kpi)$.
\footnote{We could enforce a condition $|\xbi^A - \ybi| \le \frac1{(1 + \kpi)}$
(or any other convenient
positive bound), and then
price updates would occur so long as $|\xbi^A -\wti| \ge 1$;
however, the bound on the elasticity for the $y_i$ demands would increase correspondingly.}

We also need to take account of the fact that the prices $p_i$
are integral, but the calculated updates need not be.
To avoid overlarge updates, we conservatively round down the
the magnitude of each update.
Note that this reduces the value of the update by at most a factor of 2.
This introduces a second source of error,
which we need to incorporate in the analysis.
In addition, this has the following implication:
no price can be less than $1/\lmbd$,
for a smaller price would never be updated; furthermore, the price $1/\lmbd$
may not be reduced even if the update rule so indicated.
}

\smallskip
\noindent{\bf Indivisibility Parameters}.
We measure the indivisibility of the market in terms of two parameters, $r$ and $s$.
$r=M/\sum_i w_i$, where $M$ is the daily supply of money; it provides an upper bound
on the weighted average price for an item at equilibrium. This can be thought of as the
granularity of money at the equilibrium.
$s=\min_i w_i$ is the minimum size for the daily supply for any item,
and thus indicates the granularity of the least divisible good.

\smallskip

Next, we restate Constraint \ref{cst:wti-bdd} and Lemma \ref{cst:wti-bdd},
replacing $\wti$ with $\wti^I$.

\begin{constraint}
\label{cst:wti-bdd-disc}
$|\wti^I - w_i| \le \frac13 w_i$.
\end{constraint}

\Xomit{
We use a potential $\phi_i$ expressed in terms of $y_i$:
\[
\phi_i = p_i \left[\span(y_i, \ybi, \wti^I) -
4\kpi(1 + \alt) (t - \tau_i) | \ybi- \wti^I |
+ \alpha_2 | \wti^I - w_i | \right].
\]
}
Again, if $\phi$ is large enough that the following theorem guarantees it decreases,
then $\phi = O( S)$.

\begin{theorem}
\label{lem:discr-progr-bdd-dem-2}
If Constraint \ref{cst:wti-bdd-disc} holds and $y_i \le d\wti$ for all $i$,
$\frac{\alt}{2} + \alpha_1  \max\{\frac92,2(d-1)\} \le 1$,
$\loa + \frac{8}{3} \lambda \left(1 + \frac{2E d}{1-\lE}  + \frac12 \alt \right) \le 1$,
each $w_i \ge 6$,
$ 4\kpi(1 + \alt) \le \loa \leq \frac{1}{2}$ for all $i$,\\
$s \ge \frac{48} {(\alt -1) (1 - \loa)}
\left[ 1 + 6(1 + \alt)  + (1 + \alt) \frac{1 - \loa + \frac{18}{s} \kpi (1 + \alt)  + \frac{3\kpi}{s}} { 1 - \loa - \frac{18}{s} \kpi (1 + \alt)}  \right]$,\\
if $\phi(\tau) \ge \frac{48 } {(\alt - 1)}\left[ (1 + \alt) \left( \frac{4} {\lmbd r} + \frac{24}{s} \right) + \frac1s \right]
 \frac{1 - \loa} {1- \loa - \frac{18}{s} \kpi (1  + \alt)} M$
at the start of the day,
and if each price is updated at least once every day,
then $\phi$ decreases by at least a $1 - \frac{\kmin(\alt -1)}{8}$ factor
over the course of the day.
\end{theorem}

\noindent
{\bf Remark}. Again, to obtain bounds on the needed warehouse sizes, we need to use the
fast update rule.

\begin{theorem}
\label{thm:discr-lwr-bdd}
In the discrete setting there are markets with $ \Omega (E/r) $
misspending at any pricing.
\end{theorem}

\subsection{Extensions}

We briefly examine how the analysis can be applied to some divisible markets in which the WGS
constraint is relaxed, and  also indicate a possible extension to a class of markets that interpolate
between the pure Fisher market and an Exchange market.

\section{The Analysis}
\label{sec:analysis}

We now prove the above results.
For readability, we repeat the statements of lemmas and theorems.

From a technical perspective, the novely in this analysis lies in the approach for
coping with asynchrony.
The idea is to have a  potential which decreases either at or faster than the desired
rate of improvement, and whenever an event occurs (a price update), ensure that
the potential either stays the same or decreases further.
Of course, the ``real'' decreases are due to the updates, but because of the interplay
between the different prices and demands, it proves easier to show the desired
rate of progress using our approach.

The most challenging of these analyses is the one used to handle the fast updates.
Because the fast updates may introduce some temporarily ``bad'' events (events that
would increase the potential) these are deferred for the purposes of the analysis.
Tracking the differences between the real market and the market with deferred events
is a delicate matter which has to be done with considerable accuracy.
The resulting potential function is fairly elaborate.

We now proceed to prove the results in the order they were introduced in the previous section.

\Xomit{
We begin by considering a tatonnement process in the context of
a Fisher market with a single good,
in order to determine reasonable constraints on the demand function,
namely constraints that enable a fairly rapid convergence of the price
toward its equilibrium value.

We will then investigate what additional constraints suffice
to achieve similar performance in the multi-good setting,
where the price for each good $g$ is updated
independently of the other prices,
based solely on its current value and the demand for $g$.
In other words, the price updates occur asynchronously.
Even with this, the rate of convergence is similar to
the single good case.
However, the analysis uses significantly more sophisticated arguments.
}

\subsection{A Single Good}

Let $x$ denote the demand for the good and $p$ its price.
We assume that $x(p)$ is a strictly decreasing function.
Let $w$ denote the supply of the good.
Our goal is to update $p$ repeatedly so as to cause $x$
to converge toward $w$.
We will use the following update rule:
\[
\label{eqn:p-update-repeat}
p'=p \left(1+ \lmbd \min\left\{ 1, \frac{x-w}{w} \right\} \right)
\]
where $0 < \lambda \le \frac12$ is a suitable fixed parameter.

We are going to argue that this update function is appropriate
for demands $x$ satisfying:
\begin{eqnarray}
\label{eqn:deriv-bound-repeat}
\frac{x(p)}{p} \le -\frac{dx}{dp} \le E\frac{x(p)}{p},
       ~~\mbox{where } 1\le E\le \frac{1}{2\lambda}.
\end{eqnarray}
We will also explain why the update function
can be viewed as an approximation of Newton's method,
but an approximation quite unlike the standard secant method.

The following fact will be used repeatedly in our analysis.

\begin{fact}
\label{fact:taylor}

\emph{(}a\emph{)} If $\delta\geq -1$ and either $a \leq 0$, or $a \geq 1$,
then $(1+\delta)^a \geq 1 + a\delta$.  \\
\emph{(}b\emph{)} If $0 > \delta\geq -\frac12$ and $0< a < 1$,
then $(1+\delta)^a \geq 1 + 2a\delta$.  \\
\emph{(}c\emph{)} If $\delta\geq -1$ and $0 \leq a \leq 1$, then $(1+\delta)^a \leq 1 + a\delta$,
and the inequality is strict if $a\ne 0,1$. \\
\emph{(}d\emph{)} If $-1<\rho\le
\delta\le 0$ and $-1\le a\le 0$, or if $a\leq -1$ and $0\le a \delta \leq
\rho<1$,
then $(1+\delta)^a\leq 1+a\delta/(1-\rho)$.\\
\emph{(}e\emph{)} If $0 \leq \delta \le \rho < 1$ and $-1\leq a \leq 0$, then
$(1+\delta)^a \leq 1+a\delta(1-\rho)$.
\end{fact}
\begin{proof}
We prove (a)--(c), using a simplified version of Taylor's Theorem:
\begin{theorem}[Taylor]
If $f$ is a twice differentiable function in the interval
$[0,x]$ $($or $[x,0]$, if $x <0)$ then there is a $\xi\in[0,x]$
$($or $\xi\in [x,0]$, if $x < 0)$ such that
\[
f(x) = f(0) + f'(0)*x + \frac{f''(\xi)}{2} x^2.
\]
$\square$
\end{theorem}
Let $f(x) = (1+x)^a$.  Then $f'(x) = a(1+x)^{a-1}$,
$f''(x) = a(a-1)(1+x)^{a-2}$,
$f(0) = 1$, and $f'(0) = a$.  Thus we have that
\[
(1+x)^a =  1+ax+\frac{a(a-1)}{2}(1+\xi)^{a-2} x^2.
\]
If the last term is nonnegative, then we have that
$(1+x)^a \geq  1+ax$.  The last term is nonnegative
provided that $x \geq -1$ (implying $\xi \geq -1$)
and either $a \geq 1$ or $a \leq 0$.  This is (a).

If the last term is nonpositive, then we have that $(1+x)^a \leq
1+ax$.  The last term is nonpositive provided that $x \geq -1$ and
$0 \leq a \leq 1$, and strictly negative if $0< a <1$. This is (c).

(b) holds if the last term is at least $ax$.
This is true since $ax < 0$ and $0< \frac{(a-1)}{2}(1+\xi)^{a-2} x <
1$ when $0\ge \xi \ge x \ge -\frac 12$ and $0<a<1$.

Parts (d) and (e) are obtained by bounding the limit of the infinite
Taylor series for $(1+\del)^a$.
Namely,
\[
(1+\del)^a = \sum_{i\ge 0} \del^i \frac{a(a-1)(a-2)\cdots (a-i+1)}{i!}.
\]
If $-1\le a \le 0$ and $\del \le 0$, then, for $i\ge 0$, $(a-i)/(i+1)\ge -1$, and so $\del(a-i)/(i+1) \le -\del$;
in this case the sum
is bounded by $1 + \sum_{i\ge 1} a \del (-\del)^{i-1} \le 1 + a\del/(1- \rho)$, as claimed for the first
result in (d).
\\
If  $a \le -1$ and $0\le a\del < 1$, then $(a-i)/(i+1)\ge a$, and so $\del (a-i)/(i+1) \le a\del$;
in this case the sum
is bounded by $1 + \sum_{i\ge 1} (a\del)^i \le 1 + a\del/(1- \rho)$, as claimed for the second
result in (d).
\\
If $-1\le a \le 0$ and $0 \le \del < 1$, then, for $i\ge 0$,  $0 \ge (a-i)/(i+1)\ge -1$;
in this case the sum is a sequence of alternating terms,
each successive term being smaller in magnitude by at least a $\del$ factor.
Consequently, the sum is bounded by $1 +a\del -a\del^2 \le 1 +a\del(1-\rho)$ (recall that $a\le 0$),
as claimed for the result in (e).

\end{proof}

\begin{lemma}
\label{lem:x-bdd}
For $\Del>0$,
$\frac{x(p)}{(1+\Del)^E} \le x(p(1+\Del))\le \frac{x(p)}{(1+\Del)}$ and
$\frac{x(p)}{(1-\Del)}\le x(p(1-\Del))\le \frac{x(p)}{(1-\Del)^E}$.
\end{lemma}
\begin{proof}
If $\frac{dx}{dp} = -Ex/p$ for all $p$, then $\frac{d(p^Ex)}{dp} = 0$ for all $p$,
and then $p^E \cdot x(p) = p^E(1+\Del)^E \cdot x(p(1+\Del))$, for all $\Del$
(positive or negative).
Similarly, if $\frac{dx}{dp} = -x/p$ for all $p$,
then $\frac{d(px)}{dp} = 0$ for all $p$,
and then $p\cdot x(p) = p(1+\Del) \cdot x(p(1+\Del))$.
These cases provide the extreme bounds on the growth of $x$
from which the claimed bounds follow.
\end{proof}
\begin{corollary}
\label{cor:f-grth}
$f(p)=p\cdot x(p)$ is a non-increasing function of $p$.
\end{corollary}
\begin{proof}
$p(1+\Del) \cdot x(p(1+\Del)) \le p(1+\Del) \cdot x(p)/(1+\Del) = p \cdot x(p)$.
\end{proof}

\Xomit{
\begin{definition}
\label{def:same-side}
Let $x'$ denote $x(p')$.
Price $p'$ is said to be on the \emph{same side} as $p$ if either both $x,x' \ge w$
or both $x,x' \le w$. We also say that $p'$ is a \emph{same-side update}.
\end{definition}
Let $p^*$ denote the value of $p$ for which $x(p^*)=w$.
Note that $p'$ is on the same side as $p$ exactly if both $p,p'\ge p^*$ or if both $p,p'\le p^*$.

While same side updates are not important in the 1d setting, in our multi-dimensional
applications they are needed to ensure progress.
And so in the next lemma we investigate how large a step size is possible
while having a same side update.
}

\noindent
{\bf Lemma \ref{lem:same-side-update}.}~%
\emph{Suppose that $-p\frac{dx}{dp} \le Ex(p)$, for some $E \ge 1$.
Then
the update $p'=p (1+ \frac{1}{2E} \min\{ 1, \frac{x-w}{w} \} )$ is
on the \emph{same side} as $p$,
but the update $p'=p (1+ \frac{1}{E} \frac{x-w}{w} )$ need not
be.}
\begin{proof}
If $p'=p(1+\Delta)$ is not on the same side as $p$ for $\Delta > 0$, then
$x(p) > w > x(p(1+\Delta))$.
By Lemma \ref{lem:x-bdd}, this implies that
$w> x(p)(1+\Del)^{-E}$; equivalently,
$x(p) < w(1+\Del)^E\le w(1+2E\Del)$ if $E\Del\le\frac12$
(applying Fact \ref{fact:taylor}(d) with $\rho = \frac12$);
thus either $E\Del\ge\frac12$ or $w(1+2E\Del) > x(p)$.
It follows that $\Del\ge \frac{1}{2E} \min\{\frac{x-w}{w},1\}$.

Next, we show that replacing the $1/(2E)$ parameter by $1/E$ can result in updates which
are not on the same side.
Consider the case that $\frac{dx}{dp} = -E\frac{x}{p}$ for all $p$;
then, as noted in the proof of Lemma \ref{lem:x-bdd},
$p^E\cdot x(p)$ is constant.
Let $p(w)$ denote the price at which the demand is $w$.
In this case $p(w) = p\cdot(\frac{x}{w})^{1/E} = p\cdot(1 + \frac{x-w}{w})^{1/E}
< p\cdot (1 + \frac{1}{E} \frac{x-w}{w})$ if $1/E < 1$
(applying Fact \ref{fact:taylor}(c)).
Now suppose that the update
rule being used is $p'=p (1+ \frac{1}{E} \frac{x-w}{w} )$.
Then $p'>p(w)$ and this would not be a same side update.

Finally, suppose that
$p'=p(1-\Delta)$ is not on the same side as $p$ for $\Delta > 0$; then
$x(p) < w < x(p(1+\Delta))$.
By Lemma \ref{lem:x-bdd}, this implies that
$w< x(p)(1+\Del)^{-E}$;
equivalently,
$x(p) > w(1-\Del)^E\ge w(1-E\Del)$ as $| \Del| \le 1$
(applying Fact \ref{fact:taylor}(a));
it follows that $\Del > -\frac{1}{E} \frac{x-w}{w}$.
\end{proof}
Thus, if we want same-side updates,
and we are using the rule  $p'=p (1+ \lambda \frac{x-w}{w} )$
when $\frac{x-w}{w} \le 1$ (i.e.~$x\le 2w$)
then we are obliged to have a derivative bound more or less
of the form given in Equation \ref{eqn:deriv-bound}.

For  $x > 2w$, the update $p'=p(1+\lambda)$ is not the largest possible
same-side update,
but it is simple which is why we use it.

Another way to view the update rule is to see it as
an approximate form of Newton's method:
\[p' = p - \frac{x(p)-w}{dx/dp}.
\]
Instead of $dx/dp$ (which is not available in our setting)
we use a lower bound on its value
(recall that $dx/dp$ is negative by assumption),
for this only reduces the change to $p$.
The lower bound is $-Ex(p)/p$.
This yields
\[p'= p + \frac{p}{E}\frac{x(p) -w}{x(p)}.
\]
When $x(p)< w$, replacing $x(p)$ by $w$ in the denominator
only decreases the lower bound on $dx/dp$ further,
and yields the update rule $p'= p + \frac{p}{E}\frac{x(p) - w}{w}$.
Note that our update rule achieves half this step size when $\lambda=1/(2E)$.
\\
When $x(p) > w$ we have two cases:\\
$x(p) \le 2w$: then replacing $x(p)$ by $2w$ in the denominator
only decreases the lower bound on $dx/dp$,
while our update rule achieves this step size when $\lambda=1/(2E)$.
\\
$x(p) \ge 2w$: then replacing $x(p)$ by $2(x(p) -w)$
only decreases the lower bound, and once again our update rule
achieves this step size when $\lambda=1/(2E)$.


Next we investigate the rate of convergence.
\Xomit{
It might seem natural to determine how quickly $|x^* - x(p)|$
decreases, but in higher dimensions, this proves a less effective measure.
}
Recall that we  measure the ``distance'' from equilibrium by the potential function
$\phi(x,p)=|x-w|p$.

\begin{lemma}
\label{lem:phi-diverges}
 $\phi$ increases as $p$
diverges from $p^*$ \emph{(}in either direction\emph{)}.
\end{lemma}
\begin{proof}
By Corollary \ref{cor:f-grth},
$p \cdot x(p)$ is a non-increasing function of $p$.
For $p>p^*$, $\phi=  pw - px(p)$; $pw$ is strictly increasing, and
$-px(p)$ is non-decreasing; thus in this case $\phi$ is strictly increasing.
For $p<p^*$, $\phi= px(p) - pw$; as a function of $-p$, $px(p)$ is non-decreasing,
and $-pw$ is strictly increasing; thus $\phi$ is a strictly increasing function of
$-p$ in this case.
Together, these observations show that $\phi$ is strictly increasing as
$p$ diverges from $p^*$.
\end{proof}

\smallskip
\noindent
{\bf Notation}.
$\phi_I$ denotes the initial value of $\phi$ and
$\phi_F$ a target final upper bound for $\phi$.

\medskip
\noindent
{\bf Theorem \ref{lem:1d-rate-of conv}.}~%
\emph{%
If $x\le 2w$ initially, then in one interation $\phi$ reduces by at least $w|p'-p| = \lambda \phi$,
where $p$ is the price before the update and $p'$ the price after the update.
Hence in $O(\frac{1}{\lmbd}\log\frac{\phi_I}{\phi_F})$
iterations, $\phi$ reduces from $\phi_I$ to at most $\phi_F$.
If $x\ge 2w$ initially,  then in one interation $\phi$ reduces by at least $w|p'-p| = \lambda \frac{w}{x-w}\phi$.
Hence in $O(\frac{1}{\lmbd} \frac{x-w}{w}\log\frac{\phi_I}{\phi_F})$
iterations, $\phi$ reduces from $\phi_I$ to at most $\phi_F$.
}
\begin{proof}
Case 1: $x<w$.
\\
Let $p'=p(1-\Del)$,
where $\Del= -\lmbd \min\{1, \frac{x-w}{w}\}= -\lmbd\frac{x-w}{w} > 0$.

Then, as the price update is same-sided,
the value of the potential after the price is updated is given by
$\phi'=p(1-\Del)[w-x(p(1 - \Del))]$.
This is bounded as follows:
\begin{eqnarray*}
p(1-\Del)[w-x(p(1 - \Del))] & \le & p(1-\Del)w - p \cdot x(p)\\
           &&            ~~~\mbox{as } p \cdot x(p)\le p(1- \Del) \cdot x(p(1-\Del)) \mbox{ by Corollary \ref{cor:f-grth}}\\
& \le & p(w-x) - \Del wp.
\end{eqnarray*}
Thus $\phi - \phi' \ge \Del wp = w\cdot \lambda p \frac{w-x}{w}=\lambda \phi$.
Also note that $\Del wp = (p-p')w=w|p-p'|$.

\smallskip
\noindent
\emph{Case 2}: $x>w$.\\
Now, let $p'=p(1+\Del)$,
where $\Del=\lmbd \min\{1, \frac{x-w}{w}\}>0$.

Again, as the price update is same-sided, the value of the potential after the price is updated
is given by $\phi'= p(1+\Del)[x(p(1+\Del)) -w]$.
This is bounded as follows:
\begin{eqnarray*}
p(1+\Del)[x(p(1+\Del)) -w] & \le &
p \cdot x(p) - p(1+\Del)w \\
 && ~~~\mbox{as } p \cdot x(p)\ge p(1+ \Del) \cdot x(p(1+\Del)) \mbox{ by Corollary \ref{cor:f-grth}}\\
 & \le & p(x-w) - p\Delta w.
\end{eqnarray*}
Thus $\phi - \phi' \ge \Del wp= w(p'-p)=w|p-p'|$.

\smallskip
\noindent
\emph{Case 2.1}:  $x\le 2w$.

Here $w\Del p = w\cdot \lambda p \frac{x-w}{w}=\lambda \phi$.

\smallskip
\noindent
\emph{Case 2.2}: $x\ge 2w$.

Here $w\Del p = w\cdot \lmbd p = \frac{w}{x-w} \lmbd p (x-w)= \lmbd \frac{w}{x-w} \phi$.

\end{proof}

\noindent
{\bf Remark}.
Changes of variable (and possibly consequent changes to the update rule)
can allow various other potential functions to be modified so as
to have the form $|x- w|p$.\\
e.g. 1.  $\phi' =|\tilde{x}^\alpha - \tilde{w}^\alpha | \tilde{p}^\beta $,
with $\alpha, \beta >0$;
setting $x = \tilde{x}^\alpha$, $w =\tilde{w}^\alpha$, $p = \tilde{p}^\beta$,
yields $\phi' =\phi =|x - w|p $.\\
2. $\phi' = \max \{\frac{\tilde{x}}{\tilde{w}},
\frac{\tilde{w}}{\tilde{x}}\}^{\tilde{p}}$;
setting $p = \tilde{p} $, $x=\log\tilde{x}$,
$w = \log\tilde{w}$, yields $ \log \phi' = \phi = |x- w| p$.
The formulation using $\phi$ seems more natural,
since $\phi (w,w,p) = 0 $ while $\phi ' (w,w,p) = 1 $.
\\
3. $\phi' =(\tilde{x} - \tilde{w})^2 \tilde{p}$;
setting $\phi'' = \sqrt{\phi'} =|\tilde{x} - \tilde{w}| \sqrt{\tilde{p}}$
yields an instance of example 1.

\subsection{Multiple Goods}

\Xomit{
Suppose that there are $n$ goods $g_1, g_2,  \cdots , g_n $.
Let $p= (p_1 ,p_2 , \cdots,p_n)$ be the $n$-vector of prices
for these goods,
and let  $x_i(p)$, or  $x_i$ for short,
denote the demand for the $i$th good at prices $p$
(that is, it is a function of all the prices,
not just the price of the $i$th good).

Again, we use the update rule
\[
\label{eqn:p-update-i}
p_i'=p_i \left(1+ \lmbd \min\left\{ 1, \frac{x_i-w_i}{w_i} \right\} \right).
\]
Again, we assume that the demands satisfy the bounded rate of change property
which we term \emph{bounded elasticity} in the current context.
\begin{definition}
\label{def:bdd-elas}
Good $i$ has \emph{bounded elasticity}, if for any collection $p$ of prices:
\[
\frac{x_i(p)}{p_i} \le -\frac{dx_i}{dp_i} \le E\frac{x_i(p)}{p_i},
       ~~\mbox{where } 1\le E\le \frac{1}{2\lambda}.
\]
\end{definition}

We use the potential function $\phi = \sum_i \phi_i $,
where $\phi_i = p _i |x _i - w _i| $.
So the effect of an update to $p_i$ on $\phi_i$ is just as in the 1-d case.
However the update may also affect the demands $x_j$,
and hence the potentials $\phi _j$, for goods $j \ne i$.
Before we quantify this we introduce some more notation.
}

\noindent
{\bf Notation}.\\
1. Let $\Delta_i p_i=p'_i - p_i$
denote the change to $p_i$ due to its update. \\
2. Let $\Delta_i x_j=x'_j - x_j$ denote the change to $x_j$ due to the update to $p_i$.
(Recall that $x'_j$ denotes $x_j(p_{-i},p'_i)$.)\\
3. Let $I_i$ denote $w_i|\Delta_i p_i|$.\\
4. Let $(p_{- i }, p'_ i )$ denote the price vector $ p $
with the $i$th price replaced by $p'_ i $. \\
5. For all $j$, let
$\Gamma_i  \phi _ j =  \phi _ j ( p _{- i }, p_ i ) - \phi_ j ( p _{- i }, p' _ i ) $
denote the reduction in $ \phi _ j $ when $ p _ i $ is updated.
This may be negative for $j \ne i$.
In addition, let $\Gamma_i  \phi =\sum_j \Gamma_i  \phi _ j$; this is the
reduction to $\phi$ when $ p _ i $ is updated.

\smallskip

By Lemma \ref{lem:1d-rate-of conv}, $\Gamma_i \phi_i$
is at least $I_i= w_i|\Delta_i p_i |
= \lmbd p_i w_i \min\{ \frac{|x_i - w_i|}{w_i}, 1\} =
\lmbd \phi_i \cdot \min\{ 1 , \frac{w _i}{|x_i - w_i|} \}$.

To achieve a similar rate of progress as in the 1-d case,
as we will see,
it will suffice to assume that  Property \ref{prop:progress-1} below holds.

\begin{definition}
\label{def:toward-w}
The update $p'_i$ is said to be \emph{toward} $w_i$ if
$|x_i(p_{-i},p'_i) - w_i| < |x_i(p_{-i},p_i) - w_i|$.
\end{definition}

\begin{property}
\label{prop:progress-1}
When $p'_i$ is a same-side update toward $w_i$,
$\sum_{j\ne i}| \Gamma_i  \phi _ j  | + \alpha  I_i  \le  \Gamma_i   \phi _ i$,
for some $ \alpha $, $0  < \alpha \le 1 $.
\end{property}
The term on the RHS is the reduction in $\phi_i$.
The sum on the LHS is the the sum of the changes to the $\phi_j $,
$j \ne i$, which are assumed, in a worst-case way,
to all be increases (which is possible with certain $w_j$,
namely  $w_j > x_j > x'_j $  or $w_j < x_j < x'_j $).
The second term on the LHS is the desired reduction in $ \phi _ i $.
In the 1-d case the reduction was by at least $I_i$;
here we relax this to provide more flexibility.

\begin{lemma}
\label{lem:phi-reduction}
If Property \ref{prop:progress-1} holds,
then when $p'_i$ is given by a same-side toward $w_i$ update,
$\Gamma_i \phi \ge \alpha I_i$.
\end{lemma}
\begin{proof}
$\Gamma_i  \phi =\sum_j \Gamma_i  \phi _ j \ge
 \Gamma_i  \phi _ i - \sum_{j \ne i} |\Gamma_i  \phi _ j| \ge \alpha I_i$.
\end{proof}

Next, we show that the standard assumption of weak gross substitutes
(WGS) in Fisher markets implies Property \ref{prop:progress-1} with $ \alpha = 1 $.
This claim is shown in Lemma \ref{lem:progress-1}, below.

\Xomit{
\begin{definition}
\label{def:WGS}
A market satisfies the \emph{gross substitutes property} \emph{(}GS\emph{)} if for
any good $i$, increasing $p_i$ leads to increased demand for all
other goods. The market satisfies \emph{weak gross substitutes} if
the demand for every other good increases or stays the same.
\end{definition}
}

\begin{definition}
\label{defn:spndng-neutr}
Let price $p_i$ change to $p_i + \Delta_i p_i$. The
corresponding {\em spending neutral} change to demand for good $i$,
$\Delta_i^N x_i$, is given by:
\begin{equation}
\label{eqn:spndg-neutr}
p_i x_i = (p_i + \Delta_i p_i) (x_i + \Delta_i^N x_i).
\end{equation}
\end{definition}

In the following lemma we show that given a same side update toward $w_i$,
the corresponding spending neutral change to $x_i$ decreases $\phi_i$ by
$w_i|\Del_i p_i| = I_i$.

\begin{lemma}
\label{lem:spdng-neutr}
\begin{enumerate}
\item
If $x_i > x_i + \Delta_i^N x_i \geq w_i$, then $\Delta_i p_i > 0$ and
\[ p_i (x_i - w_i) - (p_i + \Delta_i p_i) (x_i + \Delta_i^N x_i - w_i) = w_i \Delta_i
p_i = I_i. \]
\item
If $x_i < x_i + \Delta_i^N x_i \leq w_i$, then $\Delta_i p_i < 0$ and
\[ p_i(w_i - x_i) - (p_i + \Delta_i p_i) (w_i - x_i - \Delta_i^N x_i) = -w_i \Delta_i p_i
= w_i |\Delta_i p_i| = I_i.\]
\end{enumerate}
\end{lemma}
\begin{proof}
We show the first claim. By (\ref{eqn:spndg-neutr}),
\[ p_i (x_i - w_i) - (p_i + \Delta_i p_i) (x_i + \Delta_i^N x_i - w_i) = -p_i w_i +
(p_i + \Delta_i p_i) w_i = w_i \Delta_i p_i. \]
Note that $\Delta_i p_i > 0$ as $\Delta_i^N x_i < 0$.

\smallskip

The proof of the second claim is similar.
\end{proof}

\begin{lemma}
\label{lem:progress-1}
If all demands obey WGS, then
Property \ref{prop:progress-1} holds with $ \alpha = 1 $.
\end{lemma}
\begin{proof}
We begin by showing that $\sum _{j \ne i} \Gamma_i  \phi_j
  = \sum _{j \ne i} p_j \Delta _i x_j =
p_i x_i - (p_i + \Delta p_i)( x_i + \Delta _i x_i)$.
This follows from
$\sum _j p_j x_j  =  (p_i + \Delta p_i)(x_i + \Delta _i x_i)
                   + \Sigma _{j \ne i} p_j (x_j + \Delta _i x_j)$,
which holds because when $p_i$ is updated,
the combined total spending on all the goods is unchanged.

Next we show that $\Gamma_i  \phi _ i  = I_i + \sum _{j \ne i} p_j |\Delta _i x_j|$.

For the case $x_i > w_i$, we have:
\begin{eqnarray*}
\Gamma_i  \phi _ i  & = & \phi  _ i ( p _{- i }, p _ i ) - \phi  _ i ( p _{- i }, p '_ i )\\
                 & = &  p_i ( x_i - w_i) - ( p_i + \Delta_i  p_i  ) ( x_i + \Delta_i  x_i - w_i)\\
                 & & ~~~\mbox{(recall that a same-side update is being applied to $p_i$)}\\
                 & = &  w_i \Delta_i  p_i +  p_i x_i - (p_i + \Delta p_i)( x_i + \Delta _i x_i)\\
                 & = &  I_i + \sum_{j\ne i} p_j \Delta_i x_j~~~~\mbox{as $\Del_i p_i > 0$.}
\end{eqnarray*}
And $\Delta_i x_j \ge 0$ for $j\ne i$, as there is a price increase applied to $p_i$,
which by WGS only increases the demand for good $j$, $j\ne i$.
So $\Gamma_i  \phi _ i =  I_i + \sum_{j\ne i} p_j |\Delta_i x_j|$
in this case.

For the case $x_i < w_i$, $\phi  _ i ( p _{- i }, p _ i ) = p_i ( w_i - x_i)$
and $\phi  _ i ( p _{- i }, p '_ i ) = ( p_i + \Delta_i  p_i  ) (w_i - x_i - \Delta_i)$.
Running through the same algebra yields
$\Gamma_i  \phi _ i = - w_i \Delta_i  p_i
- \sum_{j\ne i} p_j \Delta_i x_j= I_i + \sum_{j\ne i} p_j |\Delta_i x_j|$,
for in this case there is a price decrease applied to $p_i$, which implies that
$\Delta_i x_j \le 0$ for $j\ne i$.

Finally, for $j \ne i$, we observe that:
\begin{eqnarray*}
 |p _ j \Delta_i x _ j|   & = &  |p _ j (x _ j + \Delta_i  x _ j - w _ j) -  p _ j (x _j - w_ j)|\\
                                & = & \left\{\begin{array}{ll}
                                             |\phi  _ j ( p _{- i }, p '_ i ) - \phi _ j ( p _{-i}, p _ i) | &
                                                   ~~~\mbox{if }\text{sign}(x_j - w_j) = \text{sign}(x_j + \Delta_i x_j- w_j)\\
                                             |\phi  _ j ( p _{- i }, p '_ i ) + \phi _ j ( p _{-i}, p _ i) | &
                                                   ~~~\mbox{otherwise}
                                         \end{array}
                                        \right.\\
                               & \ge & |\phi  _ j ( p _{- i }, p '_ i ) - \phi _ j ( p _{-i}, p _ i) | =  |\Gamma_i  \phi _ j|.
\end{eqnarray*}

Hence $\Gamma_i  \phi _ i \ge \sum_{ j \ne  i}  |\Gamma_i  \phi _ j| + I_i  $.

\end{proof}
This will not suffice when we seek to analyze multiple asynchronous updates
to different prices.
The difficulty can be understood by considering the special case where all
prices are updated synchronously.
One way to analyze this would be to apply the updates one by one in turn.
The difficulty we face is that it seems unreasonable
to assume that each successive update is based on $ x _ j $ values
updated to take account of the previous price updates.

Instead, we will show that if we apply the updates one at a time using the price updates
given by the original $  x _ i  $ values,
then there is an ordering of the goods so that when the update
to $  p _ i  $ occurs, $  \mbox{sign} ( x _ i - w _ i ) $ has not yet changed,
nor does the update change it (though, possibly, it makes  $ x _ i  =  w  _i  $).
Consequently, by Lemmas \ref{lem:phi-reduction} and \ref{lem:progress-1},
the reduction to the potential $\phi$ due to the update to $  p _ i  $
(based on the value of $  x _ i  $  prior to any of the updates)
is at least $  I _ i  $.
Hence the reduction to $\phi$ is at least $  \sum_ i  I _ i  $.

To achieve this we will need $  \lambda  (2E -1) \le  \frac12 $. Note that this condition
implies $\lambda  E \le  \frac12 $ as $E \ge 1$.

The discussion is simplified by introducing the following function $\psi_i$.
\begin{definition}
\label{def:psi}
$\psi_i = p_i(w_i - x_i)$, if $  x _ i  < w_ i  $ prior to the price updates;
while if $  x _ i  > w_ i  $ prior to the price updates, $\psi_i = p_i(x_i - w_i)$.
We also define $\Gamma_i \psi_i$, analogously to $\Gamma_i \phi_i$;
it denotes the decrease to $\psi_i$ due to the update to $p_i$.
More precisely, WLOG, suppose that the price updates are applied in the order
$p_1, p_2, \cdots, p_n$.
$\Gamma_i \psi_i = \psi_i(p'_1,\cdots,p'_{i-1},p_i,\cdots,p_n)
                        - \psi_i(p'_1,\cdots,p'_i,p_{i+1},\cdots,p_n)$.
\end{definition}

\begin{lemma}
\label{lem:seq-updates}
Suppose that the price updates are performed sequentially in some order.
Consider a point when the update to $p_i$ is about to occur.
Let $ \phi _i^I$ be the value of $ \phi _i $ prior to any update,
and let $\psi'_ i $ be the current value of $\psi_i$.
Suppose further that despite price updates to the other goods,
$\psi'_ i \ge \phi_i^I/2$.
Then, assuming that $  \lambda (2E -1) \le  \frac12  $,
$  p _ i  $'s update \emph{(}based on $  x _ i  $'s original value\emph{)} is toward $w_i$ and same-sided,
and it reduces $\psi_i = \phi_ i $ as follows:
if $\psi'_i \le \phi _i^I$, then $\Gamma_i \psi_i \le \frac12 \phi _i^I$,
and if $\psi'_i > \phi _i^I$,
then $\Gamma_i \psi_i \le (\psi'_i - \phi _i^I) +  \frac12 \phi _i^I$.
In either case, $\psi_i \ge 0$ right after the update.
Further, $\Gamma_i \phi \ge I_i = w_i|\Delta_i p_i|$.
\end{lemma}
\begin{proof}
{\bf Case 1}. $  x _ i  < w_ i  $.\\
If $\Gamma_i \psi_i\le \psi'_i$, then $\psi_i(p_{-i},p'_i) \ge 0$ and consequently
$x_i(p_{-i},p'_i) \le w_i$, in which case the update is same-sided;
consequently $\psi_i \ge 0$  right after the update.
As $\psi_i > 0$ right before the update also,
$\psi_i = \phi_i$, both right before and right after the update,
so $\Gamma_i \psi_i = \Gamma_i \phi_i $.
Further, by  Lemmas \ref{lem:progress-1} and \ref{lem:phi-reduction}, $\Gamma_i \phi \ge I_i$.
Thus it suffices to show that $\Gamma_i \psi_i \le \phi_i^I/2$, as we assume that
$\psi'_i \ge \phi_i^I/2$.

\begin{eqnarray*}
\Gamma_i \psi_ i & \le & \psi '_ i  -
   \left[ p _ i \left(1 + \lambda  \frac{  x _ i -  w  _ i }{  w _ i }\right) \right]
    \left[w_i - \left( w _i  - \frac{ \psi' _ i }{  p _ i }\right) \left(1+ \lambda  \frac{  x _ i  -  w _ i }{ w _ i  }\right) ^ {- E }\right]\\
& \le & \psi'_ i - p_i w_i + \lambda p_i(w_i - x_i) + (p_i w_i - \psi'_i)\left(1 -\lambda  \frac{  w_i -  x_i}{ w_i }\right)^{-(E-1)}\\
& \le & \psi'_ i - p_i w_i + \lambda \phi _i^I +
       (p_i w_i - \psi'_i) \left(1+2 \lambda  (E-1)  \frac{  w_i -  x_i}{ w_i } \right)\\
 & & ~~~\mbox{applying Fact \ref{fact:taylor}d, with $a=-(E-1)$,
               $\del = -\lambda  \frac{w_i -  x_i}{w_i}$, $\rho = \frac12$ when $E \ge 2$,
                  and $\rho =-\frac12$}\\
    && ~~~\mbox{when $1\le E < 2$; for when $E \ge 2$,
     $\lambda (E-1) \frac{  w_i -  x_i}{ w_i } \le \frac12 $
    and when $1 \le E < 2$, $ \lambda  \frac{  w_i -  x_i}{ w_i } \le \frac12 $;}\\
    && ~~~\mbox{further, in both cases, $-(E-1) \le 0$.}\\
& \le & \lambda \phi _i^I  + 2 \lambda  (E-1)  p_i ( w_i -  x_i) \\
& \le &   (2 E -1)\lambda  \phi _i^I \le  \frac{  \phi _i^I  }{ 2}
               ~~~ \mbox{as $  \lambda (2 E -1) \le  \frac12 $}.
\end{eqnarray*}
\noindent
{\bf  Case 2}. $  x _ i  > w_ i  $.\\
Again,  if $\Gamma_i \psi_i \le \psi'_i$,
we can conclude that the update is same-sided,
$\psi_i \ge 0$  right after the update,
and $\Gamma_i \phi \ge I_i $.
Again, this holds either if $\Gamma_i \psi_i \le \phi_i^I/2$ (which holds when
$\psi'_i \le \phi _i^I$) or if $\Gamma_i \psi_i \le \phi _i^I/2 + (\psi'_i- \phi _i^I)$
(which holds when $\psi'_i > \phi _i^I$).

\smallskip
\noindent
{\bf  Case 2.1}. $   w_i < x_i  \le 2w_i$.
\begin{eqnarray*}
\Gamma_i \psi_ i & \le & \psi '_ i  -
     \left[p _ i \left(1 + \lambda  \frac{  x _ i -  w  _ i }{  w _ i } \right) \right]
       \left[\left(  w  _ i + \frac{   \psi' _ i }{  p _ i } \right)
                 \left(1+ \lambda  \frac{  x _ i  -  w _ i }{ w _ i  }\right) ^ {- E } - w_i\right] \\
  & \le & \psi'_ i + p_i w_i + \lambda \phi _i^I -
        (p_i w_i + \psi'_i)\left(1 +\lambda  \frac{  x_i -  w_i}{ w_i }\right)^{-(E-1)}\\
   & \le & \lambda \phi _i^I +
            \lambda ( E -1) \frac{  x_i -  w_i}{ w_i } (p_i w_i + \psi'_i)\\
       && ~~~ \mbox{applying Fact \ref{fact:taylor}a, with $a=-(E-1)$
                    and $\del = \lambda \frac{  x_i -  w_i}{ w_i }$,
                       as $  \lambda    \frac{  x_i -  w_i}{ w_i }  \le 1 $ and $-(E-1) \le 0$.}\\
  & \le &\lambda \phi _i^I +
    \lambda (E-1) [ (x_i-w_i)p_i + \psi'_i]  \\
    & \le & \left\{\begin{array}{ll}
       \lambda (2 E - 1) \phi _i^I & ~~~ \mbox{if $\psi'_i \le \phi _i^I$}\\
       \lambda (2 E - 1) \phi _i^I + \lambda ( E -1)(\psi'_i- \phi _i^I)
           & ~~~ \mbox{if $\psi'_i > \phi _i^I$}
           \end{array}
    \right.\\
           & \le & \left\{\begin{array}{ll}
         \phi _i^I/2 & ~~~ \mbox{if $\psi'_i \le \phi _i^I$}\\
        \phi _i^I/2 + (\psi'_i- \phi _i^I)
           & ~~~ \mbox{if $\psi'_i > \phi _i^I$}
           \end{array}
       \right.
       ~~~\mbox{as $\lambda (2 E -1)  \le  \frac12 $}
  \end{eqnarray*}
\noindent
{\bf  Case 2.2}. $   2 w_i  < x_i$.
\begin{eqnarray*}
\Gamma_i \psi_ i & \le & \psi '_ i  -
   \left[  p _ i (1 + \lambda ) \right] \left[ \left(  w  _ i + \frac{   \psi' _ i }{  p _ i } \right)
          (1+ \lambda  )\ ^ {- E } - w _i \right ]\\
& \le & \psi'_ i + p_i w_i(1 +\lambda) - (p_i w_i + \psi'_i)(1+ \lambda  )\ ^ {- (E-1) }\\
& \le & \lambda p_i w_i  + \lambda ( E -1)(p_i w_i + \psi'_i)\\
     && ~~~ \mbox{applying Fact \ref{fact:taylor}a, with $a=-(E-1)$
                    and $\del = \lambda$, as $  \lambda \le 1 $ and $-(E-1) \le 0$.}\\
& \le & \lambda E \phi _i^I +\lambda ( E -1)(2\phi _i^I + \psi'_i - \phi _i^I)
   ~~~ \mbox{as $  p_i w_i\le p_i (x_i - w_i) = \phi _i^I$ }\\
& \le & \lambda (2 E -1) \phi _i^I + \lambda (E-1)(\psi'_i - \phi _i^I)\\
& \le & \phi _i^I/2   +  (\psi'_i - \phi _i^I)
                ~~~ \mbox{as $\lambda ( 2E -1)  \le  \frac12 $}.
\end{eqnarray*}
\end{proof}

\begin{lemma}
\label{lem:ordering-updates}
When all the updates are applied, for at least one good, WLOG  $  g _ a  $,
either $  \mbox{sign$  (x _ a - w _ a ) $}$ is unchanged or $ x _ a  =  w _ a  $  now.
\end{lemma}
\begin{proof}
WLOG let $  g_1,  \cdots,  g_i $ be the goods for which $  x_h  \ge  w_h  $,
$  h  \le  i  $,
$  g _{ i +1}, \cdots,  g_j  $ those for which $  x_h =  w_h $, $ i  <  h  \le  j $,
and $  g _ { j +1}, \cdots,  g _ n  $,
the remaining goods, be those for which $  x_h < w _ h $, $j < h  \le  n $.
We show by induction that there is an ordering of the updates so that,
for all $i$,
just before an update is applied to $  p_i  $,
$  \psi'_i  \ge  \phi_i ^ I /2 $;
then, by Lemma \ref{lem:seq-updates},
following the update to $p_i$,
$\psi_i \ge 0$.

\noindent
{\bf Case 1 } $  \phi_1^I + \cdots + \phi_i^I  \le  \phi _{ j +1}^I + \cdots + \phi_n^I$.

Consider applying all the price updates for goods $g_ 1, \cdots, g_ i  $.
We argue that the update to price $p_h$ leaves  $  \psi_h \ge  \phi_h^ I /2 $.
For as the ultimate change to $\psi_h$ is the same regardless of the order in
which the updates to $p_ 1, \cdots, p_ i  $ are applied,
we will analyze the case in which  $p_h$ is updated first:
then by Lemma \ref{lem:seq-updates},
its update
reduces $\psi_h$ by at most $\phi_h^I/2$;
the updates to the other $  p_{h'} $, $ h ' \le  i  $, price increases,
by WGS only increase $x_ h $ and hence $  \psi_h $.
We conclude that after the updates to $p_ 1, \cdots, p_ i  $, for every $h\le i$,
$\psi_h \ge \phi_h^I/2$.

Because the total spending is fixed,
any reductions to $\psi_l$, $l > j$, result from matching reductions to the $\psi_h$,
$h\le i$. The total available reduction is at most
$  \frac12 \sum_{ h \le i}  \phi _ h^ I  \le \frac 12  \sum_{ l  >  j } \phi _ l ^ I  $.
Hence there is some good $  g_l  $, such that, after the updates to all of $p_ 1, \cdots, p_ i  $,
$  \psi_l  \ge \phi_l ^ { I }/2 $. WLOG let $ l =  n  $.

Applying the updates to $   p  _ { j +1}, \cdots,   p  _{ n -1}  $
(price decreases) only decreases $ x_n  $ and hence only increases $  \psi_n  $.

So the update to $ p_n  $ can be applied last and just before it is applied,
$  \psi_n  \ge  \frac12 \phi_n ^ I  $.
Thus we can choose $  a  =  n $.

\noindent
{\bf Case 2 }  $  \phi_1^I + \cdots + \phi_i^I  \ge  \phi _{ j +1}^I + \cdots + \phi_n^I$.

A symmetric argument applies here.

By induction, the claim applies to the remaining $  n -1 $ goods.
\end{proof}

\noindent
{\bf Theorem \ref{cor:sim-progress}.}~%
\emph{%
Suppose that the market obeys WGS and elasticity $E$.
When all the prices are updated simultaneously, if $\lmbd(2E-1) \le \frac12$,
$\phi=\sum_i \phi_i$ reduces by at least
$\sum_i \lmbd \phi_i^I\min\left\{1, \frac{w_i}{|x_i - w_i|}\right\}$.
If the demands satisfy $x_i \le dw_i$ throughout, then
$\phi$ reduces from $\phi_I$ to at most $\phi_F$
in $O(\frac{1}{\lmbd} (d-1)\log\frac{\phi_I}{\phi_F})$
iterations.
}
\begin{proof}
Simply apply the updates in the order given by Lemma \ref{lem:ordering-updates}.
\end{proof}

\subsection{Asynchronous Price Updates with Bounded Demands}

\Xomit{
But we don't want to assume that all the price updates are simultaneous.
For this to make sense we have to reinterpret the meaning of the parameters for supply and demand.
We now view these as rates.

To make this concrete, let us call the basic time unit a day.
$w_i$  is the daily supply of good $i$, which is assumed to arrive continuously throughout the day.
$x_i$ will be the instantaneous rate of demand given the current prices.
The demand over some time period  $[t_1 , t_2] $ is given by $\int _{t_1 } ^{t_2 } x_i (t) dt $.
 We are interested in scaling this demand by the time period;
we call this the average demand, denoted  ${\overline x}_i [ t_1 , t_2 ] $;
i.e., ${\overline x}_i [ t_1 , t_2 ] = \frac {1} {t_2 -t_1 }\int _{t_1 }^ {t_2 } x_i (t)dt$.
We are particularly interested in the average demand for good $i$ since the time of last update to $p_i$,
at time $\tau_i$ say:
we define the average demand at the current time by
${\overline x}_i = {\overline x}_i [\tau_i ,  t] =\frac {1} {t_2 -t_1 }\int _{t_1 } ^{t_2 } x_i (t) dt $.
The average demand can be measured in our setting and
so this is the value for the demand that will be used in the price update function:
\[
p_i'=p_i \left(1+ \lmbd \min\left\{ 1, \frac{{\overline x}_i - w_i}{w_i} \right\} \right).
\]
To ensure progress we require that each price updates at least once a day.
}

Recall that we use the following potential function: $\phi=\sum_i \phi_i$, where
\begin{eqnarray}
\label{eqn:phi-asynch-def-repeat}
\phi_i(x_i,{\overline x}_i,w_i) =
p_i \left[\span(x_i,{\overline x}_i,w_i) -\alpha_1 \lmbd|w_i - \xbi|(t - \tau_i)\right],
\end{eqnarray}
with $\alpha_1  > 0 $ being a suitable constant and
$\span(x,y,z)$
denoting the length of the interval spanned by its arguments, i.e.\
$\max\{x,y,z\} - \min\{x,y,z\}$.
At the end of this section we will relate $\phi$ to our definition of misspending.

$\phi_i$ will decrease continuously:
in Lemma \ref{lem:cont-progress}, we will show that $\frac{d\phi_i}{dt} \le -\lmbd \phi_i$.
In Corollary \ref{cor:updates-only-help}, we will also show that when the price update occurs,
$\phi =\sum_i \phi_i$ only decreases.
Together, these imply a daily decrease in potential by at least a $ 1 -\theta (\lambda ) $ factor.

In fact, we will need a more elaborate potential to cope with the following scenario:
suppose that ${\overline x}_i < w_i$
and yet owing to last-minute price reductions to other goods
$x_i \gg w_i$.
Then applying the update to $p_i$ may increase $\phi_i$
well beyond the available ``savings'' of $p_i |w_i - x_i|$.
So for the moment we assume that all demands are bounded at all times:
\begin{assumption}
\label{ass:bdd-dem}
$x_i \le d w_i$ for all $i$, where $d \ge 2$ is a suitable constant.
\end{assumption}
Later, we will show how to drop this assumption.

\begin{lemma}
\label{lem:cont-progress}
If $\loa \leq \frac{1}{2}$, then
$\frac{d\phi_i}{dt} \leq - \loa p_i\span(\xbi, x_i, w_i) \leq - \loa
\phi_i$ at any time when no price update is occurring \emph{(}to any
$p_j$\emph{)}; this bound also holds for the one-sided derivatives when a
price update occurs.
\end{lemma}
\begin{proof}
We begin by showing $\frac{d\xbi}{dt} = \frac{1}{t - \tau_i} (x_i -
\xbi)$. For $ \xbi = \frac{1}{t - \tau_i} \int_{\tau_i}^t x_i dt$;
so $\xbi + (t - \tau_i) \frac{d \xbi}{dt} = x_i$.

Now we bound $\frac{d\phi_i}{dt}$.

\medskip\noindent
{\bf Case 1:} $\xbi \geq x_i \geq w_i$ (or symmetrically, $\xbi \leq
x_i \leq w_i$). Then $\phi = p_i [(\xbi - w_i) (1 - \loa(t -
\tau_i))]$.
\[ \frac{d\phi_i}{dt} = p_i \left[ \frac{-1}{t - \tau_i} (\xbi - x_i) \left(1 - \loa (t -
\tau_i) \right) - \loa (\xbi - w_i) \right] \leq -\loa p_i(\xbi -
w_i) \leq -\loa \phi_i ~~\mbox{as $t - \tau_i \le 1$}.  \]

\medskip\noindent
{\bf Case 2:} $x_i > \xbi \geq w_i$ (or symmetrically, $x_i < \xbi
\leq w_i$). Then $\phi_i = p_i [(x_i - w_i) - \loa (t - \tau_i)
(\xbi - w_i)]$.
\[\frac{d\phi_i}{dt} = -\loa p_i\left[(\xbi - w_i) + (x_i - \xbi) \right] = -\loa p_i
(x_i - w_i) \leq -\loa \phi_i .\]

\medskip\noindent
{\bf Case 3:} $\xbi > w_i \geq x_i$ (or symmetrically, $\xbi < w_i
\leq x_i$). Then $\phi_i = p_i [(\xbi - x_i) - \loa (t -
\tau_i)(\xbi - w_i) ]$.
\[ \frac{d\phi_i}{dt} = p_i \left[ \frac{-(\xbi - x_i)}{t - \tau_i} - \loa ((\xbi - w_i) +
(x_i - \xbi)) \right] \leq p_i \left[ -(\xbi - x_i) - \loa (x_i -
w_i) \right] \]
\[ \leq -p_i \left(1 - \loa\right) (\xbi - x_i) \leq - \left(1 - \loa\right) \phi_i \leq
-\loa \phi_i, \qquad \mbox{as~}\loa \leq \frac{1}{2}. \]
\end{proof}

\begin{lemma}
\label{lem:bound-Delta-x} ~
\begin{enumerate}
\item
If $\Delta_i p_i > 0$,
\[-\frac{\Delta_i x_i}{x_i} \leq \frac{E\Delta_i p_i}{p_i} = \lambda E \min \left\{1,
\frac{\xbi - w_i}{w_i} \right\}.\]
\item
If $\Delta_i p_i < 0$,
\[\frac{\Delta_i x_i}{x_i} \leq \frac{E(-\Delta_i p_i)}{(1 - \lE) p_i} = -\frac{ \lE}{1-\lE}
\frac{(\xbi - w_i)}{w_i} , \mbox{~if~~}\frac{ \lE}{1-\lE} \leq 1.\]
\end{enumerate}
\end{lemma}
\begin{proof}
We begin with (1). By Lemma \ref{lem:x-bdd},
\begin{eqnarray*}
 \frac{x_i - (-\Delta_i x_i)}{x_i} \geq \left( \frac{p_i}{p_i + \Delta_i p_i}\right)^E
= \left(1 + \frac{\Delta_i p_i} {p_i} \right)^{-E} \geq 1 - \frac{E\Delta_i
p_i}{p_i}, \\
~~~ \mbox{using Fact \ref{fact:taylor}a, with $a=E$ and $\del = \frac{\Del_i p_i}{p_i}$,
as $\frac{E\Delta_i p_i}{p_i} \ge -1$}.
\end{eqnarray*}
(1) now follows readily.

\smallskip

For (2), by Lemma \ref{lem:x-bdd},
\begin{eqnarray*}
\frac{x_i + \Delta_i x_i}{x_i} \leq \left( \frac{p_i}{p_i - (-\Delta_i p_i)}\right)^E =
\left( \frac{p_i - (- \Delta_i p_i)}{p_i}\right)^{-E} \leq 1 +
\frac{E}{1-\lmbd E}\frac{(-\Delta_i p_i)}{ p_i},\\
~~~ \mbox{using Fact \ref{fact:taylor}d, with $a=-E$, $\del = \frac{\Del_i p_i}{p_i}$,
and $\rho =\lmbd E$,
as $E(\frac{(-\Delta_i p_i}{p_i}) \le E\lmbd$}.
\end{eqnarray*}
(2) now follows readily.
\end{proof}

Note that any non-same-side update can be split into two same-side updates:
the first causes ${x_i} = {w_i}$ and the second changes
${x_i}$ to its final value.
Consequently, we will analyze only same-side updates henceforth.
We say that the update of ${p_i}$ to
$p'_i$ is \emph{toward} ${w_i}$ if either ${x_i} > x'_i \ge {w_i}$
or
${x_i} < x'_i \le {w_i}$, and it is \emph{away from}
${w_i}$ if either ${x'_i} > x_i  \ge {w_i}$
or
${x'_i} < x_i  \le {w_i}$,
where $x'_i  = {x_i}( p'_i)$.

Let ${\psi _i} = {\phi _i} - \span(x_i,x_i^u,w_i) - {\alpha _1}\lambda | w _i - x_i^u| (t - {\tau _i} )$,
and let $\Delta_i \psi_i$ denote  the increase in $\psi_i$ when $p_i$ is updated
(this definition is useful in later sections where
the definition of $\phi_i$ is changed with the effect that ${\psi _i} \ne 0$).
$x_i^u$ is the values of $x_i$ used in computing the price update; in this section,
$x_i^u =  \xbi$.

\begin{lemma}
\label{lem:drop-span}
If  Assumption \ref{ass:bdd-dem} holds \emph{(}$x_i \le d w_i$\emph{)},
when $p_i$ is updated, $\phi $ increases by at most the following amount:
\\
Case 1. the update is toward ${w_i}$:
\begin{equation}
\label{eqn:0cstrt-1-lem-updates-only-help}
\loa | x_i^u - w_i| p_i  + \Delta_i \psi_i - w_i|\Delta_i p_i|.
\end{equation}
\\
Case 2.  The update is away from ${w_i}$:
\begin{equation}
\label{eqn:0cstrt-2-lem-updates-only-help}
 \left(1 + \frac{2Ed}{1 - \lE} \right) w_i |\Delta_i p_i|
+ \Delta_i \psi_i
- (1- \loa )  p_i|x_i^u  - {w_i}|.
\end{equation}
\end{lemma}
\begin{proof}
Case 1:
First, we increase $\psi_i$ by $\Delta_i \psi_i$.

Next, we reduce the term $p_i \span( x_i, x_i^u, w_i)$ in $\phi_i$
to $p_i |x_i - w_i|$.
Following this, we update $x_i$ to $x_i + \Delta_i^N x_i$ and reduce
$p_i|x_i - w_i|$ to $(p_i + \Delta_i p_i)|x_i + \Delta_i^N x_i - w_i|$.
By Lemma \ref{lem:spdng-neutr}, this reduces $\phi_i$ by
$w_i|\Delta_i p_i|$. We also remove the term $\loa(t - \tau_i)|x_i^u -
w_i| p_i$ from $\phi_i$.
The increase to $\phi_i$  following these changes is at most
\[
\loa (t - \tau_i) |x_i^u - w_i| p_i + \Delta_i \psi_i - w_i|\Delta_i p_i|.
\]
Note that $t - \tau_i \le 1$, as there is a price update at least once a day;
this yields the bound in (\ref{eqn:0cstrt-1-lem-updates-only-help})
in the statement of the Lemma.

Finally, we change $x_i$ by a further $\Delta_i x_i - \Delta_i^N x_i$.
This reduces $\phi_i$ by $(p_i + \Delta_i p_i) (\Delta_i x_i - \Delta_i^N
x_i)$ and may increase other $\phi_j$ by up to this amount (due to a
transfer of this amount of spending from the span term in $\phi_i$
to span terms in $\phi_j$, $j\ne i$).
At worst, this leaves the potential $\phi$ unchanged.
\\
Case 2:
First, we increase $\psi_i$ by $\Delta_i \psi_i$.

Next, we reduce $p_i \span(x_i, x_i^u, w_i)$ to $p_i |x_i - w_i|$.
This yields a saving of $p_i |x_i^u - w_i|$.

Again, we remove the term $\loa p_i  (t - \tau_i)|x_i^u - w_i|$, with this cost.
Next, we update $p_i$ to $p_i + \Delta_i p_i$ and change $x_i$
by $ \Delta_i^N x_i$. By Lemma~\ref{lem:spdng-neutr}, this
increases $\phi_i$ by $w_i |\Delta_i p_i|$. Then, we update $x_i$ by a
further $\Delta_i x_i - \Delta_i^N x_i $. This increases $\phi_i$ by
$(p_i + \Delta_i p_i) |\Delta_i x_i - \Delta_i^N x_i|$. It may also cause
up to an equal increase in other $\phi_j$ due to matching spending
transfers between good $i$ and goods $j \neq i$. The net increase in
potential is bounded by $w_i |\Delta_i p_i| + 2(p_i + \Delta_i p_i) |\Delta
x_i| + \loa p_i |x_i^u - w_i| + \Delta_i \psi_i - p_i | x_i^u - w_i|$.
\\
Case 2.1: $x_i \leq w_i < x_i^u$.\\
Here $0 \le \Delta_ip_i \le \lmbd p_i$.
By Lemma~\ref{lem:bound-Delta-x}, as $\Delta_i p_i \geq 0$, $|\Delta_i x_i| \leq
x_i E \Delta_i p_i / p_i \le w_i E (\Delta_ip_i)/p_i$. Also $\Delta_i p_i \le \lmbd p_i$.
So the increase to $\phi$ is at most $w_i |\Delta_i p_i| + 2p_i (1 + \lmbd) w_i E (\Delta_ip_i)/p_i
+ \loa p_i |x_i^u - w_i| + \Delta_i \psi_i - p_i |x_i^u - w_i|$;
on rearranging, this amounts to
$w_i |\Delta_i p_i| (1 + 2E (1 + \lmbd) )
 + \Delta_i \psi_i  - (1 - \loa)  p_i | x_i^u - w_i|$.
This is bounded by (\ref {eqn:0cstrt-2-lem-updates-only-help}).
\\
Case 2.2: $x_i \ge w_i >x_i^u$.\\
Here $\Delta_i p_i \leq 0$; so
$p_i + \Delta_i p_i \leq p_i$. By Lemma~\ref{lem:bound-Delta-x}, $|\Delta_i
x_i| \leq \frac{E}{1-\lmbd E} x_i \frac{|\Delta_i p_i|}{p_i} \le \frac{E}{1-\lmbd E} d w_i \frac{|\Delta_i p_i|}{p_i}$.
So the increase to $\phi$ is at most
$w_i |\Delta_i p_i| + 2p_i  \cdot \frac{E}{1-\lmbd E} d w_i \frac{|\Delta_i p_i|}{p_i}
 + \loa p_i |x_i^u - w_i| + \Delta_i \psi_i - p_i |x_i^u - w_i|$;
on rearranging, this amounts to
(\ref {eqn:0cstrt-2-lem-updates-only-help}).
\end{proof}

\begin{corollary}
\label{cor:updates-only-help}
If Assumption \ref{ass:bdd-dem} holds \emph{(}$x_i \le d w_i$\emph{)}, $\alpha_1(d-1) \leq 1$,
and $\loa + \lambda (\frac{2Ed}{1-\lmbd E}  + 1)\leq 1$
then when $p_i$ is updated, $\phi$ only decreases.
\end{corollary}
\begin{proof}
Recall that $\Delta_i p_i = \lambda \min \{ 1, \frac{\xbi -
w_i}{w_i} \}$.
Also note that $\psi_i =0$ and hence  $\Delta_i \psi_i =0$ too.
Finally, recall that $x_i^u  = \xbi$.

\medskip\noindent
{\bf Case 1:} A toward ${w_i}$ update.

The increase in $\phi$
is bounded by (\ref{eqn:0cstrt-1-lem-updates-only-help}).

If $|\xbi - w_i|  \leq w_i $, $\Delta_i p_i = \lambda p_i (\xbi -
w_i) / w_i$;
for the increase to be non-positive,
it suffices that
$\alpha_1 \leq 1$.

If $|\xbi - w_i|   > w_i$, then $\Delta_i p_i = \lambda p_i$.
Also $|\xbi - w_i| \le (d-1) w_i$.
Again,
for the increase to be non-positive,
it suffices that
$\alpha_1(d-1) \leq 1$.

\smallskip\noindent
{\bf Case 2:} An away from ${w_i}$ update.
\\
The increase in $\phi$ is bounded by
(\ref{eqn:0cstrt-2-lem-updates-only-help}).
Again, $\Delta_i \psi_i =0$.
Also, $|\Delta_i p_i| \leq p_i \lambda |\xbi -w_i| / w_i$.
For the increase to be non-positive,
it suffices that $\loa + \lmbd \left(1 + \frac{2Ed}{1-\lmbd E } \right) \leq 1$.
\end{proof}

\noindent
{\bf Theorem \ref{lem:progr-bdd-dem}.}~%
\emph{%
Suppose that the market obeys WGS and has elasticity $E$.
If $x_i \le d w_i$, for all $i$, where $d \ge 2$,
$\alpha_1 (d-1) \le 1$,
$\loa + \lambda \left(1 + \frac{2Ed}{1-\lmbd E } \right)\le 1$,
and each price is updated at least once every day,
then $\phi$ decreases by at least a $1 - \frac{\loa}{2}$ factor daily.
Hence $\phi$ reduces from $\phi_I$ to at most $\phi_F$
in $O(\frac{1}{\loa} (d-1)\log\frac{\phi_I}{\phi_F})$
days.
}
\begin{proof}
By Corollary \ref{lem:progr-bdd-dem}, $\phi$ only decreases whenever a price update occurs.
By Lemma \ref{lem:cont-progress}, $d\phi/dt \le  \loa \phi$,
which implies $\phi(t+1)\le e^{-\loa}\phi(t) \le (1 -\frac{\loa}{2}) \phi(t)$ as $\loa \le 1$.
\end{proof}

\noindent
{\bf Remark}.
One could imagine having a distinct elasticity bound for each good,
$E_i$ for good $i$, $1 \le i \le n $, say.
Then we would use a distinct
parameter $\lmbd_i$ for good $i$.
To ensure uniform progress in reducing the
various $\phi_i $,
a price update for good $i$ would need to
occur every $\Theta(\lmbd_i f)$ time units,
where $ f $ is an upper bound on the
time between updates for goods with elasticity 1.
As this does not alter the character of our results, we leave the details to
the interested reader. 

\subsection{The Ongoing Market, or Incorporating Warehouses}

We begin by showing that Constraint\ref{cst:wti-bdd} allows us to relate $\wti$ and $w_i$.

\begin{lemma}
\label{lem:wti-constr}
If Constraint \ref{cst:wti-bdd} holds, then
$\frac34 \wti \le w_i \le \frac32 \wti$ and
$|\wti - w_i| \le \frac12 \wti$.
In addition,
Constraint \ref{cst:wti-bdd} holds for all possible warehouse contents if
$\kpi\max\{|c_i - s_i^*|,s_i^*\}\le \frac13 w_i$, and if
in addition $s_i^* = c_i/2$, the condition becomes $\kpi \le \frac23\frac{w_i}{c_i}$.
\end{lemma}
\begin{proof}
Constraint \ref{cst:wti-bdd} implies $\frac23 w_i \le \wti \le \frac43 w_i$.
Thus $\frac34 \wti \le w_i \le \frac32 \wti$, which gives the first claim.
The second pair of claims is immediate from the definition of $\wti$, as
$|s_i - s_i^*|$ is maximized either when $s_i=c_i$ or when $s_i=0$.
\end{proof}


\Xomit{
The price of good $i$ is updated according to the following rule:
\begin{equation}
\label{eqn:pr_change}
p'_i \leftarrow  p_i \left(1 +
     \lambda~ \mbox{median}\left\{-1, \frac{ \overline z_i
         (p)}{w_i}, 1 \right\}
   \right)
\end{equation}
This rule ensures that the change to $p_i$ is bounded by
$\pm \lambda_i p_i$.
}

\Xomit{
We will also need the parameters
$\kappa_{\min}=\min_i\kappa_i$ and $\kappa_{\max}=\max_i\kappa_i$.
}

\Xomit{
We redefine the potential $\phi_i$ to also take account of the imbalance in the warehouse
stock as follows:
\[
\phi_i = p_i[\span(x_i, \xbi, \wti) -
\loa (t - \tau_i) | \xbi- \wti |
+ \alpha_2 | \wti - w_i |]
\]
where $1< \alpha_2 < 2$ is a suitable constant;
this is simply Equation \ref{eqn:phi-asynch-def}, with
$\wti$ replacing $w_i$ and with the additional term
$ \alpha_2 p_i | \wti - w_i |$.
}

We prove results analogous to Lemma \ref{lem:cont-progress}
and Corollary \ref{cor:updates-only-help} to demonstrate progress as before.
To enable us to apply the following lemma in later sections, we
define $\chi_i =\frac{d\wti}{dt}+\kpi(x_i-w_i)$ (in the current section, $\chi_i=0$).

\begin{lemma}
\label{lem:war-cont-progress}
Suppose that $ 4\kpi(1 + \alt) \le \loa \leq \frac{1}{2}$.
If $|\wti - w_i| \le 2\span(\xbi, x_i, \wti)$, then
$\frac{d\phi_i}{dt}  \leq -\frac{\loa}{4(1 +\alt)}\phi_i - p_i\chi_i$,
and otherwise
$\frac{d\phi_i}{dt}
 \leq - \frac{\kpi(\alt - 1)}{2} \phi_i - p_i\chi_i$,
at any time when no price update is occurring
\emph{(}to any $p_j$\emph{)};
this bound also holds for the one-sided derivatives when a
price update occurs.
\end{lemma}
\begin{proof}


The analysis builds on the proof for Lemma \ref{lem:cont-progress}.
Taking account of the two changes to the previous form of the potential,
we can conclude
\begin{eqnarray*}
\frac{d\phi_i}{dt} & \le & -\loa p_i\span(\xbi, x_i, \wti)
+ [\kpi(x_i - w_i) - \chi_i] p_i \max\{(1 -\loa(t - \tau_i)),\loa(t-\tau_i)\} \\
&& ~~~~~~ - \kpi(x_i - w_i) \alpha_2 p_i \sgn(\wti-w_i).
\end{eqnarray*}

\smallskip\noindent
{\bf Case 1:} $|\wti - w_i| \le 2\span(\xbi, x_i, \wti)$.\\
Then $|x_i - w_i| \leq 3 \span(\xbi, x_i, \wti)$. And
\begin{eqnarray*}
\frac{d\phi_i}{dt} & \le & -\loa p_i\span(\xbi, x_i, \wti)
+\kpi p_i(1 + \alpha_2)|x_i - w_i|  - p_i\chi_i\\
& \le & -p_i (\loa - 3\kpi(1 + \alpha_2))\span(\xbi, x_i,\wti) - p_i\chi_i\\
& \le & -p_i \frac{(\loa - 3\kpi(1 + \alpha_2))}{1+\alt}
\left[\span(\xbi, x_i, \wti) + \alt |\wti - w_i|\right] - p_i\chi_i \\
& \le & -\frac{(\loa - 3\kpi(1 + \alpha_2))}{1+\alt} \phi_i - p_i\chi_i
\le -\frac{\loa}{4(1 +\alt)}\phi_i - p_i\chi_i.
\end{eqnarray*}

\smallskip\noindent
{\bf Case 2:} $|\wti - w_i| > 2\span(\xbi, x_i,\wti)$.\\
Then $|\wti - w_i| \le 2 |x_i - w_i|$. And
\begin{eqnarray*}
\frac{d\phi_i}{dt} & \le & -\loa p_i \span(\xbi, x_i, \wti)
-\kpi p_i(\alpha_2-1)|x_i - w_i| - p_i\chi_i \\
& \le & -p_i \loa \span(\xbi, x_i, \wti)
-\kpi p_i \frac{(\alpha_2-1)}{2}|\wti - w_i| - p_i\chi_i \\
& \le &  - \min\left\{ \loa, \kpi\frac{(\alpha_2-1)}{2} \right\}\phi_i - p_i\chi_i
 \le  -  \frac{\kpi(\alt - 1)}{2} \phi_i - p_i\chi_i.
\end{eqnarray*}
\end{proof}

\begin{corollary}
\label{cor:phi-rate-of-change}
Suppose that $ 4\kpi(1 + \alt) \le \loa \leq \frac{1}{2}$.
If $\phi \ge 2(1 + 2\alt) \sum_i |\wti - w_i| p_i$, then
$\frac{d\phi}{dt}  \leq - \frac{\loa}{8(1 +\alt)}\phi - \sum_i p_i\chi_i$,
and otherwise
$\frac{d\phi}{dt}
 \leq - \frac{\kpi(\alt - 1)}{2} \phi - \sum_i p_i\chi_i$,
at any time when no price update is occurring
\emph{(}to any $p_j$\emph{)};
this bound also holds for the one-sided derivatives when a
price update occurs.
\end{corollary}
\begin{proof}
Let $I = \{i~|~|\wti - w_i| \le 2\span(\xbi, x_i, \wti)\}$.
If $\phi \ge 2(1 + 2\alt) \sum_i |\wti - w_i| p_i$, then
$\sum_{i\in I} \phi_i \ge \frac12 \phi$.
Thus
\begin{eqnarray*}
\frac{d\phi}{dt} & \le & \sum_{i\in I} \frac{d\phi_i}{dt} - \sum_{i\notin I} p_i\chi_i \\
& \le & - \frac{\loa}{4(1 +\alt)} \sum_{i\in I} \phi_i  - \sum_i p_i\chi_i \\
& \le & - \frac{\loa}{8(1 +\alt)} \phi - \sum_i p_i\chi_i.
\end{eqnarray*}
\end{proof}

In Lemma \ref{lem:war-updates-only-help-templ} below, $\wti$ takes on the role
of $w_i$ in Corollary \ref{cor:updates-only-help}.
Accordingly,
the bound of $x_i \le d w_i$ from Assumption \ref{ass:bdd-dem}
is replaced by a bound of $x_i \le d \wti$,
which is ensured by the following assumption.
\begin{assumption}
\label{ass:revised-dem-bdd}
$x_i \le \frac23 d w_i$ for all $i$, where $d\ge 2$.
\end{assumption}
\begin{lemma}
\label{lem:xbi-bound}
If Assumption \ref{ass:revised-dem-bdd} and Constraint \ref{cst:wti-bdd} hold,
then $x_i \le d \wti$.
\end{lemma}
\begin{proof}
This is immediate from Lemma \ref{lem:wti-constr}.
\end{proof}
\begin{lemma}
\label{lem:war-updates-only-help-templ}
If Constraint \ref{cst:wti-bdd} holds and $x_i \le d \wti$,
then when $p_i$ is updated, $\phi$ increases by at most the following:
\\
(i) With a toward $\wti$ update:
\begin{equation}
\label{eqn:cstrt-1''-lem-updates-only-help}
\loa |\xbi - \wti| p_i + \frac12\alt \wti|\Del_i p_i|  - \wti|\Delta_i p_i|.
\end{equation}
(ii) With an away from $\wti$ update:
\begin{equation}
\label{eqn:cstrt-2''-lem-updates-only-help}
\left(1 + \frac{2E d}{1-\lE} \right) \wti |\Delta_i p_i|
 + \frac12\alt \wti|\Del_i p_i|
- (1- \loa )  p_i| \overline{x}_i  - {\wti}|.
\end{equation}
\end{lemma}
\begin{proof}
We apply Lemma \ref{lem:drop-span},
with $x_i^u=\xbi$ and $\wti$ replacing $w_i$, as here
$\phi_i$ is defined in terms of $\wti$ rather than $w_i$.
Also, now  $\psi_i = \alt p_i |\wti - w_i|$.
So $\Del_i \psi_i = \alt \Del_i p_i |\wti - w_i| $.
By Lemma \ref{lem:wti-constr}, $|\wti - w_i| \le \wti/2$,
so
$\Del_i \psi_i \le \alt |\Del_i p_i| \wti/2$.
\end{proof}
\begin{corollary}
\label{cor:war-updates-only-help}
If Constraint \ref{cst:wti-bdd} holds and $x_i \le d \wti$,
$\frac{\alt}{2} + \max\{\frac32,(d-1)\}\alpha_1 \leq 1$,
and $\loa + \frac43 \lmbd \left(1 + \frac{2E d}{1-\lE}  + \frac12 \alt \right) \le 1$,
then when $p_i$ is updated, $\phi$ only decreases.
\end{corollary}
\begin{proof}
Recall  that by Lemma \ref{lem:xbi-bound}, $\xbi \le (d-1)\wti$.

\smallskip
\noindent
{\bf Case 1:} An update  toward ${w_i}$:
\\
The increase in $\phi$ is is bounded by
(\ref{eqn:cstrt-1''-lem-updates-only-help}).
\\
(i) $|\xbi - \wti|  \leq w_i $.

$\Delta_i p_i = \lambda p_i (\xbi - \wti) / w_i$, and
by Lemma \ref{lem:wti-constr} $ w_i \le \frac32 \wti$.
Thus
for the increase to be non-positive,
$\frac32 \alpha_1 + \frac12 \alt \leq 1$ suffices.
\\
(ii) $|\xbi - \wti|  \geq w_i $.

As $|\xbi - \wti| \le (d-1) \wti$ and $|\Delta_i p_i| = \lambda p_i $,
for the increase to be non-positive,
$\alpha_1(d-1) + \frac12 \alt \leq 1$ suffices.

\smallskip
\noindent
{\bf Case 2:} An update  away from ${w_i}$:

The increase in $\phi$ is is bounded by
(\ref{eqn:cstrt-2''-lem-updates-only-help}).
As $\Del_i p_i \le \lmbd p_i |\xbi - \wti|/w_i$,
for the increase to be non-positive, the condition
$\frac{\wti}{w_i}\lmbd  \left(1 + \frac{2E d}{1-\lE}  + \frac12 \alt \right) \le 1 - \loa$ suffices.
By Lemma \ref{lem:wti-constr}, $\wti \le \frac43 w_i$, so the condition
of the previous sentence
is subsumed by
$\loa + \frac43 \lmbd \left(1 + \frac{2E d}{1-\lE}  + \frac12 \alt \right) \le 1$.
\end{proof}

\Xomit{
{\bf Notation}. Let $\kmin$ and $\kmax$ denote $\min_i \kpi$ and $\max_i \kpi$, respectively.
}

\noindent
{\bf Theorem \ref{lem:war-progr-bdd-dem}.}~%
\emph{%
If Constraint \ref{cst:wti-bdd} holds and $x_i \le d \wti$,
$\frac{\alt}{2} + \alpha_1  \max\{\frac32,(d-1)\} \le 1$,
$\loa + \frac43 \lmbd  \left(1 + \frac{2E d}{1-\lE}  + \frac12 \alt \right) \le 1$,
$ 4\kpi(1 + \alt) \le \loa \leq \frac{1}{2}$,
$\frac{\kmin(\alt -1)}{2} \le 1$,
and each price is updated at least once every day,
then $\phi$ decreases by at least a $1 - \frac{\kmin(\alt -1)}{4}$ factor daily.
}

\emph{%
In fact, if
$\phi \ge 2(1 + 2\alt) \sum_i |\wti - w_i| p_i$,
then $\phi$ decreases by at least a $1 - \frac{\loa}{8(1 +\alt)}$ factor daily.
}
\begin{proof}
By Corollary \ref{cor:war-updates-only-help}, $\phi$ only decreases whenever a price update occurs.
By Lemma \ref{lem:war-cont-progress}, $d\phi/dt \le  -\frac{\kmin(\alt -1)}{2} \phi$,
as $\chi_i =0$ for all $i$,
which implies $\phi(t+1)\le e^{-\frac{\kmin(\alt -1)}{2}}\phi(t) \le (1 -\frac{\kmin(\alt -1)}{4}) \phi(t)$ as $\frac{\kmin(\alt -1)}{2} \le 1$,
arguing as in the proof of Lemma \ref{lem:progr-bdd-dem}.

The second claim is shown in exactly the same way.
\end{proof}

\Xomit{
We finish by noting the relation between $\phi$ and the misspending $S$.
\begin{lemma}
\label{lem:missp-pot-war}
$\frac12 \phi(t) \le S(t) \le \frac{a} { 1 - \loa} \phi(t-1)$.
\end{lemma}
\begin{proof}
This is proved in eseentially the same way as Lemma \ref{lem:missp-pot}.
The only change is the presence of the term $\alt |\wti - w_i|$ in both $\phi$ and $S$.
\end{proof}
} 
\subsection{Bounds on Demands and Prices}

Next, we
determine an $f$-bound on prices given that they are $c$-demand bounded.
To obtain this we need to assume a bounded elasticity of wealth $E'=0$.
Recall that when this holds demands are called {\em normal}.

\begin{definition}
For $c \ge 1$, define $ r^{(c)} = \max_i  p_i^* / p_i^{(c)} $
and $ r^{(1/c)} = \max_i  p_i^{(1/c)}/  p_i^*  $.
\end{definition}

\noindent
{\bf Notation}. Let $ \rho = \max_{i,j} \frac{w_ip_i^*}{w_jp_j^*} $.

\begin{lemma}
\label{lemma:price-upper-bound}
$ r^{(1/c)} \le c n \rho $.
\end{lemma}
\begin{proof}
Let $ h = \arg\max_i w_ip_i^* $. Then the money $ M_a $ spent at
equilibrium is at most $ n w_h p_h^* $. Let
$ k = \arg\max_i p_i^{(1/c)} /p_i^* $.
The money $ M_b $ spent on good $k$
at prices $ {\boldp }^{(1/c)} $ is $ \frac1c w_k p_k^{(1/c)} $.
By WGS, $ M_b \le M_a $. But
\[ \frac{1}{c} r^{(1/c)} w_kp_k^* = \frac{1}{c} w_k p_k^{(1/c)} =
  M_b \le M_a \le n w_h p_h^*. \]
\[ \text {so ~} r^{(1/c)} \le cn \frac{w_hp_h^*}{w_k p_k^*} \le c n \rho. \]
\end{proof}

Obtaining a bound on $ r^{(c)} $ for $c \ge 1$
entails a more elaborate argument.

Consider the following process for decreasing prices from
from $ \boldp^* $ to $ \boldp^{(c)} $, where
WLOG $ p_1^* / p_1^{(c)} \ge p_2^*/ p_2^{(c)}
\ge \cdots \ge p_n^* / p_n^{(c)} $.

Begin by decreasing all prices uniformly, until
$ p_n $ reaches $ p_n^{(c)} $.
Continue reducing all prices except $ p_n $ uniformly until
$ p_{n-1} $ reaches $ p_{n-1}^{(c)} $, and so forth.
Call this the \emph{uniform price reduction process} (UPR for short).

\begin{lemma}
\label{lemma:price-order-reduction}
$ r^{(c)} $ is maximized when demand increases as follows, as the UPR
is applied.
During the reduction of $p_n$ from  $p_n^* $ to $ p_n^{(c)} $
only $ x_n $ changes \emph{(}by increasing to $ c w_n $\emph{)}.
In general, during the reduction of $p_i$ from
$ p_i^* p_{ i+1}^{(c)} / p_{i+1}^* $ to $ p_i^{(c)} $,
only $ x_i $ changes.
\end{lemma}
Loosely speaking, the demands increase one time.
\begin{proof}
By normality, for those goods whose prices are decreasing during the UPR,
demand either stays the same or is increased.

To maximize the price reduction goods $1, \cdots, n-1$ can
achieve one needs to
minimize the amount of money spent on them.
This implies that $ p_n^{(c)} $ needs to be as large as possible,
which occurs if all the increase in demand goes to good $n$
as the UPR proceeds, in going from $ p_n $ to $ p_n^{(c)} $.
\Xomit{
(For if in the candidate best solution, the demand for $ x_n $
is above $fw_n$ when all prices are reduced by a factor
$ p_n^* / p_n^{(c)} $, then one could imagine transferring the ``excess''
spending on good $n$ at this price  to the other goods to which it will
eventually be transferred as the UPR continues. This does not affect the
price reduction achievable on the other goods. Now one
can see it is possible
to have a larger $ p_n^{(c)} $, enabling some other $ p_i^{(c)} $'s to
be smaller.)
}

By induction, all other demands are increased as specified in the lemma.
\end{proof}

\begin{lemma}
\label{lemma:price-lower-bound}
$ r^{(c)} \le c (\rho n)^{(c-1)}$.
\end{lemma}
\begin{proof}
By Lemma \ref{lemma:price-order-reduction},
the price drop is maximized if demands increase one at
a time. So consider the initial price drop in which the price of
the first
good increases. Now imagine dividing this good into two goods,
with their demands increasing in sequence. This only increases the maximum
price drop. We iterate this process ad infinitum, leading to a continuous
process which we express as follows.

By rescaling if needed, we imagine that all goods have the same
equilibrium
price. Note that this is now a continuum of goods. As the prices are
multiplied by a factor $s<1$ (a decrease), we track a measure $W(s)$,
the quantity of goods whose
demand has not yet changed. $W(1)-W(s)$ is the quantity of
goods whose demand has grown by a factor of $c$.
We can express $W$ as a differential process.

Consider decreasing
$s$ to $s - d s$.
Suppose this results in $d W$ of the goods having demand
grow by a factor of $c$. Assuming the total spending on goods does not
increase (this maximizes the price drop),
$ W s = (W- d W) (s - d s) + c d W (s - d s) $.
To first order, $ W d s = (c -1) s d W $,
which yields $ \frac {W} {(c-1) s} = \frac {dW} {ds} $ or
$ W = W (1) s^{1/(c-1)} $.

To get back to the discrete process, the price reduction is stopped
when $W$ is reduced to $ W (1) / \rho n $, for this is a lower bound
on the amount of good $k$ with $k = \arg\min_jw_j$.

Good $k$ has a further price drop of no more than a factor of $c$
to bring its demand to $c w_k$.

The maximum possible price reduction in reducing
$ W $ to $ W (1) / \rho n $ occurs when
$ s^{1/(c-1)} = 1/(\rho n) $.
Thus $ r^{(c)} \le c/s \le c (\rho n)^{(c-1)} $.
\end{proof}

Lemma \ref{lemma:price-lower-upper-bound}
follows immediately from Lemmas \ref{lemma:price-order-reduction} and \ref{lemma:price-lower-bound}.

\subsection{Faster Updates with Large Demands}

Our previous analysis in Lemma \ref{lem:war-updates-only-help-templ}
(see Lemma \ref{lem:drop-span} also)
uses the assumption $x_i\le d \wti$
in Case 2.
This in used in the following scenario:
$x_i$ increases from less than $\wti$ to much more than $\wti$
between two successive updates to $p_i$. The price update (a
decrease) could then increase the overall potential.
However, with our new price update rule,
as $x_i$ is then large, this price decrease will be followed quickly by
several price increases. These more than undo the just mentioned
potential increase.

To analyze this, the potential `pretends' that the troublesome price
decrease did not occur yet, by delaying its instantiation in one of two ways:
either this price decrease is combined with the next
one or two price increases so that the net effect is no change or a
price increase,
or the potential instantiates the price update when next $x'_i \le (d-1)\wti$,
whichever occurs sooner, where $x'_i$  is the value of $x_i$
resulting from the delayed price updates.
All other price changes continue as before by their
previous amounts.

We need some significant changes to the potentials. For goods $i$
that do not have a delayed price increase currently, we define
\[
\phi_i = \psi_i^r = p_i \left[ \span(x'_i, \xbi', \wti)
- \loa(t - \tau_i)|\xbi' - \wti|
+ \left[ 1 - \loa (t -
\tau_i)\right]\int_{\tau_i}^t (x'_i - x_i) dt
+ \alt|\wti - w_i| \right],
\]
where
$\xbi'$ is the average value of $x'_i$ since time $\tau_i$. Note
that by WGS, $x'_i \geq x_i$ since there is no delayed update to $p_i$ in
this case. (The superscript $r$ on $\psi$ denotes
{\em regular} or non-delayed updates.)

For goods $i$ with a currently delayed update, we define
\begin{eqnarray*}
\phi'_i = \psi_i^d &=& p_i \left[
\span(x'_i, d\wti, \wti) + (\wti(\tau_i^s) - \xbi(\tau_i^s))\left[
   1 - \loa (t - \tau_i)
 \right]
 - \loa \int_{\tau_i^s}^t (x'_i - \wti) dt  + \alt|\wti - w_i| \right]\\
 && ~~~- p_i\left[  \frac{\lE }{1-\lE}
[\wti(\tau_i^s) -
\xbi(\tau_i^s)]\int_{\tau_i^s}^t  \frac{x'_i}{w_i} dt  \right]
\end{eqnarray*}
where $\tau_i^s$ is
the time at which the delayed price decrease occurred in reality.
(The superscript $d$ on $\psi$ denotes {\em delayed}.)

Next, we bound $x'_i -x_i$.

\begin{lemma}
\label{lem:wealth-eff}
Suppose that the market has bounded wealth elasticity $E'$
and bounded demand elasticity $E$.
 Let $q$ and $p$ be price vectors with $q_j
= (1 - \lmbd) p_j$ for all $j \ne i$, and
and $q_i = p_i$. Then,
\[
x_i(q) \geq (1 - \lmbd)^{E + E'}  x_i(p).  
\]

\end{lemma}
\begin{proof}
Consider prices $q'(j)= p(j) (1 - \lmbd)$ for all $j$.
Then, by the bound on wealth elasticity,
$x_i(q') \geq (1 - \lmbd)^{E'} x_i(p)$.
Increasing the price $q'(i)$ to $p(i)$ decreases
$x_i$ by at most $(1 - \lmbd) ^{E}$.
So $x_i(q) \geq (1 - \lmbd)^{ E'}  (1 + \lmbd)^{ -E}  1x_i(p) \geq (1 - \lmbd)^{E + E'}  x_i(p)$.
\end{proof}

\begin{lemma}
\label{lem:bdd-xi'diff}
$x'_i \leq \left(1 + \frac{\lmbd (E + E') } {1-\lmbd (E + E')} \right) x_i$, if $\lmbd (E + E') < 1$.
\end{lemma}
\begin{proof}
We apply the bound from Lemma \ref{lem:wealth-eff}.
Let $p$ denote the actual prices, $q$ the
prices as defined in Lemma \ref{lem:wealth-eff},
and $q'$ the prices
where the delayed updates have not been performed.
Now, for each $j$, the delayed update to $p_j$ is a decrease by at most
a $1 - \lmbd$ factor, so $q_j \le q'_j \le p_j$, ofr $j \ne i$.
At prices $q$ the demand for good $i$ is at least $x'_i (1 - \lmbd)^{E + E'}$,
by Lemma \ref{lem:wealth-eff}.
Increasing prices $q_j$ to $q'_j$, for $j \ne i$, only increases the demand
for good $i$. Reducing $q'_i = p_i$ to $q_i$ also only increases the demand for good $i$.
Thus $x_i = x_i(p) \ge (1-\lmbd)^{E + E'} x_i(q) = (1-\lmbd)^{E + E'} x_i(q') =  (1-\lmbd)^{E + E'} x'_i$;
or, $x'_i \leq (1 -
\lmbd)^{-(E + E')} x_i \leq \left(1 + \frac{\lmbd (E + E') } {1- \lmbd (E + E')} \right) x_i$, by Fact \ref{fact:taylor}(d)
with $\rho= \lmbd (E + E')$.
\end{proof}

\noindent
{\bf Notation}.
Let $E''$ denote $E'+E$.

The following lemma bounds how long a price update can be delayed until
it is instantiated in the potential function.
\begin{lemma}
\label{lem:delay-bd} If $d \geq 5$, and $\lE''  \leq \frac14$ then the instantiation of a
price update to $p_i$ is delayed by at most one day,
and it is instantiated by the time of the second
subsequent price increase to $p_i$.
\emph{(}If the net
effect of the price increase following the delayed decrease is to
leave the price unchanged, this is considered the performing of the
price update.\emph{)}
\end{lemma}
\begin{proof}
While the instantiation of update is being delayed, by Lemma \ref{lem:bdd-xi'diff},
$x_i \ge (1 -\lE'') x'_i > (1 -\lE'')(d-1)\wti$.
Hence, there has been
an excess demand of at least $2w_i$ after a time period of at most
$\frac{2w_i}{ (1 -\lE'')(d-1) \wti} \le \frac{3} {(1 -\lE'')(d-1)} \le 1$ day, as $\wti \ge \frac23 w_i$ by
Lemma \ref{lem:wti-constr};
thus two price
increases to $p_i$  will occur in this time. The net change to $p_i$ is
a multiplicative increase of at least $(1 - \lmbd) (1 + \lmbd)^2 \geq 1$ as
$\lmbd \leq \lE \le \frac14$. Consequently, no later than the second price
increase, the net change to the prices is instantiated.
\end{proof}

\begin{lemma}
\label{lem:p-only-decr}
When $p_i$ receives its delayed update,
$\phi_i$ only decreases if $d \geq 5$, $1 \le \alt \le 2$,
\begin{equation}
\label{eqn:p-only-decr-cnstrt1}
\left[1 - 2\loa - \frac{3\lE }{1-\lE} \left( 1 + \frac{\lE'' }{1-\lE''} \right) \right]
  \ge \frac43\lmbd\left[ 2(d-1)\frac{E }{1-\lE} + 1\right],
\end{equation}
\begin{equation}
\label{eqn:p-only-decr-cnstrt2}
\lmbd (\frac23 - \eta -\eta \lmbd)\left(1 - \frac13 \alt\right) \geq
      \loa \left(3 \left( \frac{\lE'' }{1-\lE''} + 1 \right) -\frac23\right),
\end{equation}
with
\begin{equation}
\label{eqn:eta-defn}
\eta = \frac
{ \loa \left(3 \left( \frac{\lE'' }{1-\lE''} + 1 \right) -\frac23\right)}
{\frac43\lmbd\left[ 2(d-1)\frac{E }{1-\lE} + 1\right]}.
\end{equation}
This holds if $d=5$, $\lE, \lE''\le \frac{1}{17}$, $\alo \le \frac15$, and $\alt=\frac32$.
\end{lemma}
\begin{proof}
By Lemma~\ref{lem:delay-bd}, this update occurs within one day of
its actual time, and no later than the second subsequent price
increase. So when the update occurs, $\int_{\tau_i^s}^t (x_i - w_i)
dt \leq 2w_i$. Hence $ \int_{\tis}^t x_i \leq 3w_i$.
Now,
\begin{eqnarray*}
\int_{\tis}^t x'_i dt
& \leq & \int_{\tis}^t x_i \left( 1 + \frac{\lE'' }{1-\lE''} \right) dt \qquad
(\text{using Lemma~\ref{lem:bdd-xi'diff}}) \\
& \leq & 3\left( 1 + \frac{\lE'' }{1-\lE''} \right)  w_i.
\end{eqnarray*}
Thus,
\begin{eqnarray*}
\label{eqn:lemma14-1} \psi_d^i & \geq p_i \Bigg[ &
\hspace{-2mm}\span(x'_i, d\wti, \wti) + [\wti(\tis) - \xbi(\tis)] \left[1 -
2\loa - \frac{3\lE }{1-\lE} \left( 1 + \frac{\lE'' }{1-\lE''} \right) \right]
\nonumber \\
& & - \loa \left(3 \left( \frac{\lE'' }{1-\lE''} + 1 \right) -\frac23\right) w_i + \alt | \wti - w_i |
~\Bigg]
~~~~\text{as $(t - \tau_i) \le 2$}.
\Xomit{
& \ge & p_i \Bigg[
\span(x'_i, d\wti, \wti) +\frac13 [\wti(\tis) - \xbi(\tis)] -3\loa w_i + \alt | \wti - w_i | \Bigg]
\\
&& ~~~~\mbox{as } \loa \le \frac13
 \mbox{~and~} 2 \lmbd E \left(\frac{20}{3}\lmbd E + \frac{10}{3}\right) \le \frac13,
 }
\end{eqnarray*}

We need to show that this expression is at least as large as the potential following the update.
Note that after the update $\phi_i =(p_i + \Del_i p_i)[(x'_i+\Del_ix'_i -\wti) +\alt|\wti - w_i|]$.
Of course, some of the $\phi_j$, $j\ne i$, may also change as a result of the update.

\noindent
{\bf Case 1:} $x'_i = (d-1) \wti$.
\\
Then the original price update, a decrease, is applied at this moment.
First, the $\span$ term is reduced by $\wti$ to $x_i' - \wti$, reducing the potential
by $p_i\wti \ge \frac23 p_i w_i$ (by Lemma \ref{lem:wti-constr}).
The term $-p_i\loa \left(3 \left( \frac{\lE'' }{1-\lE''} + 1 \right) -\frac23\right)w_i$ is also removed.
As $\loa \le \frac19$ and $\lE'' \le \frac12$, this is a net reduction.

Now, an analysis similar to that of Case 2.1 in Lemma \ref{lem:drop-span}
can be applied, as follows.
First, we update $p_i$ to $p_i+\Del_i p_i$, and increase $x'_i$ by a spending neutral
change to $x'_i + \Del_n x'_i$.
By Lemma \ref{lem:bound-Delta-x} (with $\wti$ replacing $w_i$),
this adds $\wti\Del_i p_i$ to the term $p_i(x'_i - \wti)$ in $\phi_i$.
Next, we further increase $x'_i$ to its final value $x'_i + \Del_i x_i$, which increases the potentials $\phi_i$
and $\phi_j$ for $j \ne i$ by at most
$2(p_i + \Del_i p_i) x'_i \left[(1-\frac{\Del_i p_i}{p_i} )^{-E} - 1\right]
\le 2(d-1)\wti \cdot \frac{E }{1-\lE}\Del_i p_i $,
using Fact \ref{fact:taylor}(d) with
$\rho = \frac{E }{1-\lE}$
(recall that $\rho < 1$ suffices).
The update leaves $\phi$ no larger if
\[
 p_i [\wti(\tis) - \xbi(\tis)]\left[1 -
2\loa - \frac{3\lE }{1-\lE} \left( 1 + \frac{\lE'' }{1-\lE''} \right) \right]
  \ge \lmbd\left[ 2(d-1)\frac{E }{1-\lE}\wti \Del_i p_i + \wti \Del_i p_i \right].
\]
(Note that the term involving $\alt$ only decreases as the update is a price decrease.)
But $|\Del_i p_i| \le \lmbd [\wti(\tis) - \xbi(\tis)] p_i/w_i$.
So it suffices that
(\ref{eqn:p-only-decr-cnstrt1}) holds
(using Lemma \ref{lem:wti-constr} to bound $\wti/w_i$ by $\frac43$).

\smallskip\noindent
{\bf Case 2:} $x'_i > (d-1) \wti$.
\\
The price update in this case is an increase.
\\
{\bf Case 2.1:} $ \wti - \xbi (\tis) \geq \eta w_i $.
\\
If
\begin{equation}
\label{eqn:eta-cnstrt-def}
\eta\left[1 -
2\loa - \frac{3\lE }{1-\lE} \left( 1 + \frac{\lE'' }{1-\lE''} \right) \right]
\ge \loa \left(3 \left( \frac{\lE'' }{1-\lE''} + 1 \right) -\frac23\right),
\end{equation}
then
\begin{equation}
\label{eqn:psi-d-case2-1}
\psi_i^d \ge p_i[ \span(x'_i,d\wti,\wti) + \alt |\wti - w_i | ].
\end{equation}
Note that the definition of $\eta$ in (\ref{eqn:eta-defn})
makes (\ref{eqn:eta-cnstrt-def}) equivalent to (\ref{eqn:p-only-decr-cnstrt1}).

Now, we proceed as in the analysis of Case 1 of Lemma~\ref{lem:drop-span}.
We start by reducing the $\span$ term to $x'_i - \wti$.
Next, we update $p_i$ to $p_i+\Del_i p_i$, and decrease $x'_i$ by a spending neutral
change to $x'_i + \Del_n x'_i$.
By Lemma \ref{lem:bound-Delta-x}, this decreases $\phi_i$ by $\Del_i p_i \wti$.
Finally, we further decrease $x'_i$ to its final value $x'_i + \Del_i x_i$; this only
reduces $\phi_i$ and possibly increases the other $\phi_j$ by up to an equal amount.
Thus the overall change to $\phi$ is only a decrease if
$\Del_i p_i \wti \ge \Del_i p_i \alt |\wti - w_i|$.
On applying Lemma \ref{lem:wti-constr},
we see that $\frac23 \ge \frac13 \alt$ suffices, i.e.\ $\alt \le 2$.

\smallskip\noindent
{\bf Case 2.2:} $ \wti - \xbi (\tis) < \eta w_i $.
\begin{equation}
\label{eqn:psi-d-case2-2}
\text{~Then,~~~~~}\psi^d_i \geq p_i \left[ \span(x'_i, d\wti,\wti)
- 3\loa \left(3 \left( \frac{\lE'' }{1-\lE''} + 1 \right) -\frac23\right)w_i
+ \alt | \wti - w_i | \right].
\end{equation}

We proceed as in Case 2.1.
As there, the drop in the first term is at least $\wti
\Delta p_i$, and we need
\[
\wti \Delta p_i \geq \loa p_i w_i \left(3 \left( \frac{\lE'' }{1-\lE''} + 1 \right) -\frac23\right)
+ \frac13 \alt w_i \Delta p_i .
\]
Since $\wti(\tis) -
\xbi(\tis) \leq \eta w_i$, the delayed
price reduction was by at most a factor of $(1 - \eta \lmbd)$. The
next (current) price increase is by a factor of $(1 + \lmbd)$,
yielding a net increase of at least $(1 + \lmbd) (1 - \eta \lmbd)
\geq 1 + (1 - \eta) \lmbd  - \eta\lmbd^2  \geq 1 + \lmbd (1 - \eta -\eta \lmbd)$.

So it suffices that (\ref{eqn:p-only-decr-cnstrt2}) holds.
\end{proof}

\begin{lemma}
\label{lem:reg-update-redo-pot}
If $2 \loa \leq 1$, at the original time of a
price decrease to good $i$,
\[ \phi_i \geq p_i \left[ \span(x'_i, \xbi, \wti) - \loa (t - \tau_i) | \xbi - \wti |
  + \alt |\wti -w_i| \right]. \]
\end{lemma}
\begin{proof}
At the original time of the update, $\phi_i = \psi_i^r$. And
\[ \int_{\tau_i}^t (x'_i - x_i) dt = (\xbi' - \xbi) (t - \tau_i). \]
Thus,
\[ \psi_i^r = p_i \left[ \span(x'_i, \xbi', \wti) - \loa (t - \tau_i) | \xbi' - \wti | + (\xbi' - \xbi)
\left[1 - \loa (t - \tau_i) \right] + \alt |\wti -w_i| \right]. \]

\medskip\noindent
{\bf Case 1:} $\xbi' \leq \wti$.

Recall that $\xbi' \ge \xbi $. Then $\span(x'_i, \xbi', \wti) +
(\xbi' - \xbi) \geq \span(x'_i, \xbi, w\wti)$ and $\loa | \xbi' - \wti
| + \loa ( \xbi' - \xbi ) = \loa ( \wti - \xbi )$. So $ \psi_i^r \ge
p_i[\span(x_i', \xbar_i, \wti)- \loa (t - \tau_i) (\wti - \xbar_i)]$. The claim holds
in this case.

\medskip\noindent
{\bf Case 2:} $\xbi' > \wti$.

Note that $\xbi < \wti$ as there is supposed to be a price decrease
at this time. So,

\[ \span(x'_i, \xbi', \wti) + (\wti - \xbi)[1 - \loa (t - \tau_i)] \ge
\span(x'_i, \xbi, \wti)) - \loa (t - \tau_i) (\wti - \xbar_i). \]

And \[-[1 - \loa (t - \tau_i)] ( \wti - \xbar_i) - \loa (t - \tau_i) (\xbar_i' - \wti) +
(\xbar_i' - \xbar_i) [1 - \loa (t - \tau_i)] = [1 - 2 \loa (t - \tau_i)] (\xbar_i' - \wti)
\ge 0.\]
Thus \[ \psi_i^r
\ge p_i [\span (x_i', \xbar_i, \wti) - \loa (t - \tau_i) ( \wti - \xbar_i) + \alt |\wti -w_i|]
.\]

The claim holds in this case too.
\end{proof}
\begin{corollary}
\label{cor:reg-more-del} If $2 \loa \le 1$, at the original time of a delayed price increase,
$\psi_i^r \geq \psi_i^d$.
\end{corollary}
\begin{proof}
Note that as the price increase is delayed, $x'_i \ge d \wti > \wti$.
By Lemma \ref{lem:reg-update-redo-pot},
\begin{eqnarray*}
\psi_i^r & \ge & p_i \left[ \span(x'_i, \xbi, \wti) - \loa (t - \tau_i) | \xbi - \wti| + \alt |\wti -w_i| \right] \\
& = & p_i \left[ | x'_i - \wti | + | \wti - \xbi | - \loa (t - \tau_i) | \xbi - \wti | (t - \tau_i) + \alt |\wti -w_i| \right].
\end{eqnarray*}
Note that at this time, $t = \tau_i^s$, and this expression in
$\psi_i^d$.
\end{proof}

We restate Lemma \ref{lem:drop-span}
with $x_i$ replaced by $x'_i$, $w_i$ replaced by $\wti$, and $\xiu = \xbi$.
Note that as in Lemma \ref{lem:war-updates-only-help-templ},
$\psi_i = \alt p_i |\wti - w_i|$.

\begin{lemma}
\label{lem:fast-updates-only-help-templ}
If Constraint \ref{cst:wti-bdd} holds and $x'_i \le d \wti$,
then when $p_i$ is updated, $\phi$ increases by at most the following:
\\
(i) With a toward $\wti$ update:
\begin{equation}
\label{eqn:cstrt-1''-lem-updates-only-help-fast-updts}
\loa |\xbi - \wti| p_i + \frac12\alt \wti|\Del_i p_i|  - \wti|\Delta_i p_i|.
\end{equation}
(ii) With an away from $\wti$ update:
\begin{equation}
\label{eqn:cstrt-2''-lem-updates-only-help-fast-updts}
\left(1 + \frac{2E d}{1-\lE} \right) \wti |\Delta_i p_i|
 + \frac12\alt \wti|\Del_i p_i|
- (1- \loa )  p_i| \overline{x}_i  - {\wti}|.
\end{equation}
\end{lemma}

\begin{lemma}
\label{lem:reg-update-helps}
If Constraint \ref{cst:wti-bdd} holds and $x'_i \le d \wti$,
$\frac{\alt}{2} + \max\{\frac32,(d-1)\}\alpha_1 \leq 1$,
and $\loa + \frac43 \lmbd \left(1 + \frac{2E d}{1-\lE}  + \frac12 \alt \right) \le 1$,
and if $p_i$ is updated at the regular time,
then $\psi^r_i$ only decreases.
\end{lemma}
\begin{proof}
The proof is identical to that of Corollary \ref{cor:war-updates-only-help}.
\end{proof}

Our next goal is to show $\frac{d\phi}{dt} \leq - \Theta(\kpi) \phi$.
We begin with two technical claims.

\begin{lemma}
\label{lem: xi to xi-prime-bound}
If $\frac{\lE}{1-\lE} < 1$,
$x_i - x'_i \leq
\frac{\lE''}{1-\lE''} [\wti(\tis) - \xbi(\tis)] \frac{x'_i}{w_i}$ and
$x_i \le \left( 1 + \frac{\lE}{1-\lE}\right) x'_i$.
\end{lemma}
\begin{proof}
We seek to upper bound $x_i/x'_i$.
Starting from the prices with delayed updates yielding demand $x'_i$,
if one updated $p_j$, $j\ne i$, this would only decrease the demand for good $i$.
So the only price change that increases the demand for good $i$ is the update to $p_i$,
which is by a $1-\lmbd \min\{1,[\wti(\tis) -\xbi(\tis)]/w_i \}$ factor.
Thus,
by bounded elasticity,
for a good $i$ with a delayed price update:
\begin{eqnarray*}
\label{eqn:xi-xi'bound}
\frac{x_i}{x'_i} & \leq & \left[
\frac{1}{1 - \lmbd \min\{1,[\wti (\tau_i^s) - \xbi(\tau_i^s)] / w_i)\} }
\right]^{-E} \\
& \leq & 1 + \frac{\lE}{1-\lE} \min\{1,[ \wti(\tis) - \xbi(\tis)] /
w_i\}\\
&& ~~~~\text{applying Fact \ref{fact:taylor}(d) with~}\rho=\frac{\lE}{1-\lE}.
\end{eqnarray*}
So, $x_i - x'_i \leq
\frac{\lE}{1-\lE}\min\{x'_i, [\wti(\tis) - \xbi(\tis)] \frac{x'_i}{w_i}\} $.
The claims follow.
\end{proof}
\begin{lemma}
\label{lem:spending-transfr}
\[ \sum_{i: \text{{\tiny delayed update}}} p_i (x_i -  x'_i)
\ge \sum_{i: \text{{\tiny undelayed update}}} p_i (x'_i -
x_i). \]
\end{lemma}
\begin{proof}
Consider the change in spending that would occur were all the
delayed price updates (decreases) to be made. This can only reduce the
amount of money being held. Thus:
\[ \sum_{i: \text{{\tiny delayed update}}} (p_i + \Del_i p_i)x_i - p_i x'_i~~
+ \sum_{i: \text{{\tiny undelayed update}}} p_i (x_i -
x'_i) \quad \ge 0, \]
where $\Del_i p_i < 0$ for each delayed update.
The claim follows.
\end{proof}

\smallskip
\noindent
{\bf Notation:} Let $\span_i = \span(x'_i, \xbi', \wti)$.

\begin{lemma}
\label{lem:cont-prog-fast-update}
Let $ \nu = \min \left\{\frac{\loa (d-2)} {2(d-1) + \alt}  , \frac{\kpi(\alt -1)}{2} \right\} $.
If $d \ge 5$, $ 4\kpi(1 + \alt) \le \loa \leq \frac{1}{2}$, \\
and $\kpi \left(d-1+\alt \left( 1 + \frac{\lE}{1-\lE} \right)\frac{d-1}{d-2} \right)\le \frac12 \loa$,
then $ \frac{d}{dt} \phi \le -\nu \phi $.
\end{lemma}
\begin{proof}
We will show that
\begin{eqnarray*}
\text{(i)} &~~ & \frac{d \psi_i^r}{dt} \le -\frac{\kpi(\alt -1)}{2} \psi_i^r + (x'_i - x_i) p_i,
\\
\text{and }~~~~
\text{(ii)}& ~~ & \frac{d \psi_i^d}{dt} \le -\frac{\loa (d-2)} {2(d-1) + \alt} \psi_i^d - (x_i - x'_i) p_i.
\end{eqnarray*}

Summing over all $i$, and then applying Lemma \ref{lem:spending-transfr}, we conclude that
\begin{eqnarray*}
\frac{d \psi}{dt} & \le & -\nu \psi
~~- \sum_{i: \text{{\tiny delayed update}}} p_i (x_i -  x'_i)
 + \sum_{i: \text{{\tiny undelayed update}}} p_i (x'_i -
x_i)\\
& \le & -\nu \psi.
\end{eqnarray*}
Next, we show (i).
Let $ \widetilde{\psi}^r_i =p_i[\span_i - \loa(t-\tau_i)|\xbi' - \wti| + \alt|\wti - w_i|]$.

Lemma \ref{lem:war-cont-progress} shows that
\[
\frac{d \widetilde{\psi}^r_i}{dt} \le -\frac{\kpi(\alt -1)}{2}\widetilde{\psi}^r_i
\]
for $\widetilde{\psi}^r_i$ is obtained from $\psi^r_i$ by replacing $x_i$ with $x'_i$
in the function $\phi_i$ used in Lemma \ref{lem:war-cont-progress}
and having $\chi_i=0$ there.
Now,
\begin{eqnarray*}
\frac {d \psi^r_i} {dt} & = &
\frac{d \widetilde{\psi}^r_i}{dt}
+ p_i \left[- \loa \int_{\tau_i^s}^t (x'_i -x_i) dt +[1 - \loa(t -\tau_i)](x'_i - x_i) \right].
\end{eqnarray*}
As $x'_i \ge x_i$, (i) follows.

We show (ii) next.
\begin{eqnarray*}
\frac{d \psi_i^d}{dt} & \le & p_i \left[(d-1)\kpi|x_i -w_i| - \loa  [ \wti(\tis) - \xbi(\tis)]  - \loa ( x'_i   - \wti)
  + \alt \kpi|x_i - w_i| \right] \\
&&~~~~ -p_i\left( 1 + \frac{\lE}{1-\lE} \right) [ \wti(\tis) - \xbi(\tis)] \frac{x'_i}{w_i}.
\end{eqnarray*}

Next, we bound $(d-1 +\alt) \kpi|x_i -w_i|$ by $\frac12 \loa  (x'_i - \wti)$.
By  Lemma \ref{lem: xi to xi-prime-bound}, $x_i \le \left(1 + \frac{\lE }{1-\lE} \right) x'_i$.
Recall that $x'_i \ge (d-1) \wti$,
and so $x_i \le \left( 1 + \frac{\lE}{1-\lE} \right) [(x'_i - \wti) + \wti]
\le \left( 1 + \frac{\lE}{1-\lE} \right) (1+\frac{1}{d-2}) (x'_i - \wti)
= \left( 1 + \frac{\lE}{1-\lE} \right) \frac{d-1}{d-2}(x'_i - \wti) $.
As $\kpi(d-1+\alt) \left( 1 + \frac{\lE}{1-\lE} \right) \frac{d-1}{d-2} \le \frac12 \loa $,
if $x_i \ge w_i$,
$(d-1 +\alt) \kpi|x_i -w_i| \le \frac12 \loa (x'_i - \wti)$;
otherwise, if $x_i < w_i$,
$(d-1 +\alt) \kpi|x_i -w_i| \le (d-1 +\alt) \kpi w_i \le \frac32 (d-1 +\alt) \kpi \wti
\le \frac{3}{2(d-2)} (d-1 +\alt) \kpi (x'_i - \wti) \le \frac12 \loa (x'_i - \wti)$.
Thus
\begin{eqnarray*}
\frac{d \psi_i^d}{dt} \le  - \loa  [ \wti(\tis) - \xbi(\tis)]
   - \frac12  \loa ( x'_i   - \wti) - (x_i - x'_i).
\end{eqnarray*}
As $x'_i \ge (d-1) \wti$, $\span(x'_i,d\wti,\wti) \le \frac{d-1}{d-2} (x'_i - \wti)$.
Similarly, $\alt |\wti - w_i| \le \frac12\alt \wti \le \frac{1}{2(d-2)} \alt (x'_i - \wti)$.
It follows that
$\span(x'_i,d\wti,\wti) +\alt |\wti -w_i|
\le \left[ \frac{d-1}{d-2} + \frac{\alt} {2(d-2)} \right] (x'_i - \wti) $.
Thus
\begin{eqnarray*}
\frac{d \psi_i^d}{dt} \le \frac{(d-2)} {2(d-1)+\alt)}\loa \psi_i^d -(x_i - x'_i).
\end{eqnarray*}
\end{proof}

\noindent
{\bf Theorem \ref{lem:fst-updt-progr-bdd-dem}.}~%
\emph{%
If Constraint \ref{cst:wti-bdd} holds,
$d=5$, $\alt = \frac32$, $\lE'' \le \frac{1}{17}$, $\alo \le \frac{1}{16}$,
$\loa + \frac43 \lmbd \left(\frac74 +\frac{10 E }{1-\lE}  \right) \le 1$,
$\kpi \le \frac{loa}{13}$,
and each price is updated at least once every day,
then $\phi$ decreases by at least a $1 - \nu$ factor daily,
where $ \nu = \min \left\{\frac{\loa (d-2)}{2(d-1) + \alt}  , \frac{\kpi(\alt -1)}{2} \right\} $.
}
\begin{proof}
By Lemmas \ref{lem:p-only-decr} and \ref{lem:reg-update-helps},
$\phi$ only decreases whenever a price update occurs.
By Lemma \ref{lem:cont-prog-fast-update}, $d\phi/dt \le  -\nu \phi$,
which implies $\phi(t+1)\le e^{-\nu}\phi(t) \le (1 -\frac{\nu}{2}) \phi(t)$ as $\nu \le \frac12$.
On substitution of the values and bounds for $d, \alt, \lE, \lE'' \alo$, the constraints
$\frac{\alt}{2} + \alpha_1  \max\{\frac32,(d-1)\} \le 1$,
$\loa + \frac43 \lmbd \left(1 +\frac{2 E d}{1-\lE} + \frac12 \alt \right) \le 1$,
$ 4\kpi(1 + \alt) \le \loa \leq \frac{1}{2}$,
$\kpi \left(d-1+\alt \left( 1 + \frac{\lE}{1-\lE} \right)\frac{d-1}{d-2} \right)\le \frac12 \loa$,
from Lemmas \ref{lem:p-only-decr} , \ref{lem:reg-update-helps} and
\ref{lem:cont-prog-fast-update}, reduce to the constraints stated in the current lemma.
\end{proof}

Finally, we relate the potential to the misspending.

\noindent
{\bf Notation}. Let $ S_i = p_i (|x_i -\wti | + | \xbi - \wti|) + p_i | \wti-
w_i | $; this is called be the misspending on the $ i $th good. SO
$ S = \sum_i S_i $; this is the total misspending.

\begin{lemma}
\label{lemma:potential-missp}
$ S  =O (\phi) = O(S +M)$, where $ M$ is the
daily supply of money.
\end{lemma}
\begin{proof}
We observe that
$\psi_i^r = \theta \left( p_i \left[ \span(x'_i, \xbi', \wti)
+\int_{\tau_i}^t (x'_i - x_i) dt
+ \alt|\wti - w_i| \right] \right)$,
and \\
$\psi_i^d = \theta(\span(x'_i, d\wti, \wti)   + \alt|\wti - w_i| )$.
The first bound follows from the fact that $\loa \le \frac12$.
For the second bound, we argue as in the proof of Lemma \ref{lem:p-only-decr}.
As no update has been applied, $x'_i > (d-1)\wti$, so Case 2 applies.
In Case 2.1 the claim is immediate by (\ref{eqn:psi-d-case2-1}).
In Case 2.2 the claim follows from (\ref{eqn:psi-d-case2-2}),
as $3\loa \left(3 \left( \frac{\lE }{1-\lE} + 1 \right) -\frac23\right)w_i \le \wti \le \frac12 \span(x'_i, d\wti, \wti)$.

Now, we bound $S$.

The following observation will be used twice: for a good $i$ with a delayed update,
\begin{equation}
\label{eqn:obs-xi-diff}
x_i - x'_i   \le \frac{\lE }{1-\lE} x'_i    \le \frac{\lE }{1-\lE} \frac{d-1}{d-2} (x'_i - \wti).
\end{equation}
This follows by  using Lemma \ref{lem: xi to xi-prime-bound}
for the second inequality, and the bound $x'_i > (d-1)\wti$ for the third.

Note that for a good
$ i $ with a delayed price update, $ S_i = \theta [(x_i'-\wti)p_i + p_i | \wti-
w_i |]  = \theta(\psi^d_i)$, for
if $x_i \ge \wti$,
$x_i - \wti \le (x_i - x'_i) + (x'_i - \wti)  \le \frac{\lE }{1-\lE} \frac{d-1}{d-2} (x'_i - \wti)  + (x'_i - \wti) $,
using(\ref{eqn:obs-xi-diff});
while if $x_i < \wti$, then $\wti - x_i \le \wti \le x'_i - \wti$.

For a good $ i $ with an up to date price, define
\begin{eqnarray*}
S_i &=&O (\span(x_i, \xbi, \wti) p_i + \alt p_i | \wti-w_i|) \\
& = & O (\span(x_i', \xbi', \wti) p_i + p_i(x'_i - x_i) +p_i(\xbi' - \xbi) + \alt p_i | \wti-w_i| )\\
&  = & O(\span(x_i', \xbi', \wti) p_i + \int_{\tau_i}^t p_i(x_i' - x_i)dt + \alt p_i | \wti-w_i| + p_i(x'_i - x_i))\\
& = & O(\psi^r_i + p_i(x'_i - x_i)).
\end{eqnarray*}

Now, by Lemma \ref{lem:spending-transfr} for the first inequality, and (\ref{eqn:obs-xi-diff})
for the second,
\begin{eqnarray*}
\sum_{i~\text{update not delayed}} p_i(x'_i - x_i) & \le & \sum_{i~ \text{update delayed}} p_i(x_i - x'_i)\\
& \le & \sum_{i~ \text{ delayed}} p_i   \frac{\lE }{1-\lE} \frac{d-1}{d-2} (x'_i - \wti)\\
& = & O(\sum_i \psi^d_i).
\end{eqnarray*}

Thus $S = \sum_i S_i = O(\sum_i \phi_i) = O(\phi)$.

\smallskip

Finally, we bound $\phi$.
\begin{eqnarray*}
\psi^d_i & = & O(x_i' -\wti) +  \alt p_i | \wti-w_i | )\\
& \le & O (p_i (x_i - \wti)  +  \alt p_i | \wti-w_i | + p_i (x'_i - x_i) ) \\
& = & O(S_i + p_i x_i)  ~~~~\text{using Lemma \ref{lem:bdd-xi'diff}}.
\end{eqnarray*}
So $\sum_i \psi^d_i  = O(S +\sum_i p_i  x_i) = O(S +M)$.

\begin{eqnarray*}
\psi^r_i
&=& O \left(p_i \span (x'_i, \xbi', \wti) +  \alt p_i | \wti-w_i |)  + p_i \int_{\tau_i}^t (x_i' - x_i) dt \right)  \\
&=& O (p_i \span (x_i, \xbi, \wti) +  \alt p_i | \wti-w_i |)  + p_i  (x_i' - x_i) + p_i(\xbi' - \xbi) )   \\
 & = & O (S_i  + p_i  x_i'  + p_i \xbi' )   ~~~\text{using Lemma \ref{lem:bdd-xi'diff}}.
\end{eqnarray*}
So $\sum_i \psi^r_i  = O(S +\sum_i p_i  x'_i  + p_i \xbi') = O(S +M)$.
\end{proof}

\subsection{Bounds on Warehouse Sizes}

\Xomit{
For simplicity, we will assume that $c_i/w_i$
is the same for all $i$.
Again, for simplicity, we suppose that $s_i^* = \frac12 c_i$,
that is the target fullness for each warehouse is half full.
}

We begin with two technical lemmas which show that (i) if a warehouse is rather full
 and the price does not decrease too much
henceforth then the warehouse eventually becomes significantly less full,
and (ii) an analogous result for the event that a warehouse is rather empty.

Let $\alf = \min_i \frac{\kpi c_i} {8 w_i}$.
Recall that $\wti -w_i = \kpi ( s_i - s^*_i)$.
Thus $|\wti - w_i| \le 4\alf w_i$.

\begin{lemma}
\label{lem:war-too-full-impr}
Suppose that $s_i \geq \sis + \frac {\alf} {\kpi} w_i  $.
Consider the next $k$ updates to $p_i$. Let $p_{i_1}$ be $p_i$'s
current value, and $p_{i_2}, \ldots, p_{i_{k+1}}$ be its $k$ successive
values. Let $\overline{x}_{i_j}$ be the average value of $x_i$ while
$p = p_{i_j}$. Suppose that $ p_{i_{k+1}} \geq e^{-\lmbd
f} p_{i_1} $ for some $f \geq 0$.
If $k \geq \left(1 + \frac{2}{\alf} \right) f + \frac{2}{\alf} c$, then the warehouse
stock will have decreased to less than $\sis + \frac {\alf} {\kpi} w_i $ at some
point, or by at least $cw_i$, whichever is the lesser decrease.
\end{lemma}
\begin{proof}
Suppose that $s_i \geq \sis + \frac {\alf} {\kpi} w_i $ at the time of each of the
$k$ price updates (or the result holds trivially).

Then a multiplicative price decrease by $(1 - \lmbd)$ is preceded
by an
additive increase to the warehouse stock of at most $w_i$.
A smaller price decrease, by
$1 + \lmbd (\min\{1, \overline{x}_{i_j} - \widetilde{w}_{i_j})/w_i\}$,
for $\widetilde{w}_{i_j} - \overline{x}_{i_j} < w_i$,
follows a warehouse increase of exactly
$w_i - \overline{x}_{i_j} \leq \widetilde{w}_{i_j} - \overline{x}_{i_j} - \alf w_i $.

Note that $1 + x \leq e^x$ for $|x| \leq 1$.

Suppose that there are $f + a$ price decreases by $1 - \lmbd$.
Then there are at least $a$ price increases (as the total price
decreases is by at most $e^{-\lmbd f}$). Suppose that there are
$a+a'$ price changes other than the price decreases by $1 - \lmbd$.
Now,
$\sum_{\widetilde{w}_{i_j} - \overline{x}_{i_j} < w_i} \lmbd
\min\{1,(\overline{x}_{i_j} - \widetilde{w}_{i_j})/w_i\} - (f + a)\lmbd \geq
-\lmbd f$, as this is the exponent in an upper bound on the price
decrease.
Hence $\sum_{\widetilde{w}_{i_j} - \overline{x}_{i_j} < w_i} \widetilde{w}_{i_j} - \overline{x}_{i_j} \le
  \sum_{\widetilde{w}_{i_j} - \overline{x}_{i_j} < w_i} \max\{-w_i, (\widetilde{w}_{i_j} - \overline{x}_{i_j} )\}
  \le [f -(f + a)] w_i$.
It follows that the warehouse stock increases are bounded by
$(f + a)w_i + \sum_{\widetilde{w}_{i_j} - \overline{x}_{i_j} < w_i}
(\widetilde{w}_{i_j} - \overline{x}_{i_j} - \alf w_i ) \leq  [f - \alf (a +a')] w_i$.
As $f+2a+a'=k \ge \left(1 + \frac{2}{\alf} \right) f +  \frac{2}{\alf} c$, $a + a' \geq  \frac{1}{\alf}(f + c)$,
and so the decrease in warehouse stock is at
least $cw_i$.
\end{proof}

\begin{lemma}
\label{lem:war-too-empt-impr}
Suppose that $\lmbd \left(1 + \frac{1}{\alf} \right) \le \frac{1}{2}$
and $s_i \leq \sis + \frac {\alf} {\kpi} w_i $. Let
$p_{i_1}, \ldots, p_{i_{k+1}}$, $\overline{x}_{i_1}, \ldots,
\overline{x}_{i_{k}}$ be as in Lemma~\ref{lem:war-too-full-impr}. Suppose
that $ p_{i_{k+1}} \leq e^{\lmbd f} p_{i_1} $ for some $f \geq
0$.
Further suppose that the excess demand between successive price increases is at most $w_i$.
If
\[ k \geq (1 + \lmbd) \left(1 + \frac{4}{\alf} \right)  f  + \frac{8}{\alf}\lmbd + \frac{4}{\alf} c , \]
then the warehouse stock will have increased to more than
$\sis - \frac {\alf} {\kpi} w_i $ at some point, or by at least $cw_i$, whichever is the lesser increase.
\end{lemma}
\begin{proof}
Similarly to the proof of Lemma \ref{lem:war-too-full-impr},
suppose that $s_i \leq \sis - \frac {\alf} {\kpi} w_i $ throughout this time.

Then a multiplicative price increase by $(1 + \lmbd)$ is
preceded by an
additive decrease of at most $w_i$ in the warehouse stock.
A smaller price increase, by
$1 + \lmbd (\max\{-1, \overline{x}_{i_j} - \widetilde{w}_{i_j})/w_i \}$ for
$\overline{x}_{i_j} - \widetilde{w}_{i_j} < w_i$,
follows a warehouse increase of
$w_i - \overline{x}_{i_j} \geq \widetilde{w}_{i_j} - \overline{x}_{i_j} + \alf w_i $.

Note that $1 + \lmbd x \geq e^{\lmbd x / (1 + \lmbd)} \geq
e^{\lmbd x - 2 \lmbd^2}$ for $0\le \lmbd, x \leq 1$, and
$1+\lmbd x\ge e^{\lmbd x/(1-\lmbd)}\ge e^{\lmbd x-2\lmbd^2}$,
for $-1\le  x \le 1$ and $0 < \lmbd \le \frac12$.

Suppose that there are $(1 + \lmbd) (f + a)$ price increases by
$1 + \lmbd$.
Then there are at least $a(1-\lmbd)$ price decreases
(for each such change is a drop of at most
$1 - \lmbd \geq e^{-\lmbd/(1-\lmbd)}$,
and the price increases yielded an increase of
least $e^{(a + f) (1 + \lmbd) \lmbd / (1 + \lmbd)} = e^{\lmbd(a + f)}$).

Now,
\[ \sum_{\xbij - \wtij < w_i} \left[ \lmbd (\xbij - \wtij )/w_i
  - 2\lmbd^2 \right] + (1+\lmbd)(f + a) \frac{\lmbd}{1 + \lmbd} \leq \lmbd f, \]
as this is the exponent in a lower bound on the price increase.
That is
\[ \sum_{\xbij - \wtij < w_i}  (\wtij - \xbij + 2\lmbd)  \ge  a w_i \]

Suppose that there are $a(1 - \lmbd) + a'$ price changes other than the increases by $(1 + \lmbd)$.

The decrease in the warehouse stock is at most:
\begin{eqnarray*}
& &  (1 + \lmbd) (f + a) w_i  - \sum_{\xbij - \wtij < w_i} ( \wtij - \xbij + \alf  w_i ) \\
& \leq &  (1 + \lmbd) (f + a) w_i - a w_i + 2\lmbd w_i - \alf (a' + (1 - \lmbd) a) w_i  \\
& \leq & w_i \left[  (1 + \lmbd) f + 2 \lmbd - \alf \left[a' + \left(1 - \left(1 + \frac{1}{\alf} \right) \lmbd \right) a \right] \right]
\end{eqnarray*}

Now $f (1 + \lmbd) + 2a + a' = k$, so $\left(1 - \left(1 + \frac{1}{\alf} \right) \lmbd \right) a + a' \geq \frac{\left(1 - \left(1 + \frac{1}{\alf} \right)\lmbd \right)}{2}[k - f(1 + \lmbd)]$.

If $\alf \left[a' + \left(1 - \left(1 + \frac{1}{\alf} \right) \lmbd \right) a \right] \geq 2 \lmbd + (1 + \lmbd) f + c$, the warehouse stock increases by at least $cw_i$.
It suffices to have
\[ k \geq (1 + \lmbd) f + \frac{2\alf} {\left(1 - \left(1 + \frac{1}{\alf} \right)\lmbd \right)} \left[ 2 \lmbd + (1 + \lmbd) f + c \right]. \]
By assumption $\left(1 + \frac{1}{\alf} \right)\lmbd \le \frac12$,
so $k \ge (1 + \lmbd) \left(1 + \frac{4}{\alf} \right) f + \frac{8}{\alf} \lmbd + \frac{4}{\alf} c$ suffices.
\end{proof}

In order to apply Lemma \ref{lem:war-too-empt-impr}, we need to bound the
length of the initial time interval after which all demands satisfy $x_i \le 2 w_i$.
We will need several preliminary lemmas.
In the next section, we will avoid this difficulty by modifying the price
update rule for the case that demands are large.

\Xomit{
\begin{definition}
\label{def:bdd-prices-wrhs}
Prices are $f$-bounded if $p_i$ always remains in the range $pe^{\pm f\lmbd}$ for all $i$,
and demands are $d$-bounded if $x_i \le d w_i$ for all $i$.
\end{definition}

\begin{definition}
\label{def:var-dem-equil}
$p^{(c)}$ denotes the equilibrium prices for supplies $c w_i$.
\end{definition}

\begin{definition}
\label{def:dmd-bdd-prices}
Prices $p$ are \emph{$c$-demand bounded} if $p_i\in [p^{(c)}_i, p^{(1/c)}_i] $
for all $i$.
\end{definition}
}

\noindent
{\bf Lemma \ref{prices-stay-dmd-bdd}.}~%
\emph{%
If $\frac{\lE}{1 - \lE} \le \frac 16$, and if
initially the prices have all been $c$-demand bounded for a full day
for some $c \ge 2$, they remain $c$-demand bounded thereafter.
}
%
\begin{proof}
First we analyze the bounds for low prices.

For any good $ i $, if $ p_i \le p_i^{(c)} (1 + \lmbd)^ 2 $, then by WGS and
bounded elasticity,
$ x_i \ge c (1 + \lmbd)^{-2E} w_i \ge c (1 -2\lE ) w_i  $,
as $2\lE \le 1 $, and hence
$ \xbi \ge  c (1 -2\lE ) w_i $ also,
which means that the next update to $ p_i$
will be not be a decrease so long as $ (1 -2\lE ) w_i \ge \wti$.
Applying Lemma \ref{lem:wti-constr}
to give $\wti \le \frac 43 w_i$, shows
$\lE \le \frac 16$ suffices.

Otherwise, a price decrease decreases $p_i$ to at most
$p_i^{(c)} (1 + \lmbd)^ 2 (1 - \lmbd) \ge p_i^{(c)}$; that is, $p_i$
remains above the lower bound.

Similarly, if $ p_i \ge p_i^{(1/c)} (1 - \lmbd)^ 2 $,
then $ x_i \le \frac1c (1 - \lmbd)^{-2E} w_i \le \frac1c (1 +2\frac{\lE}{1 - \lE} ) w_i  $,
as $2 \frac{\lE}{1 - \lE} \le 1 $, and hence
$ \xbi \ge  \frac1c (1 -2 \frac{\lE}{1 - \lE} ) w_i $ also,
which means that the next update to $ p_i$
will be not be an increase so long as $\frac1c (1 + 2 \frac{\lE}{1 - \lE} ) w_i \le \wti$.
Here $\frac{\lE}{1 - \lE} \le \frac 16$ suffices.

Otherwise, a price increase increases $p_i$ to at most
$p_i^{(1/c)} (1 - \lmbd)^ 2 (1 + \lmbd) \le p_i^{(1/c)}$; that is, $p_i$
remains below the upper bound.
\end{proof}

\begin{lemma}
\label{lem:supply-cost-bdd-war-bdd}
$\sum_i w_i p_i \le 3(\phi +M)$, where $M$ is the daily supply of money.
\end{lemma}
\begin{proof}
We note that if $x_i \le \frac13 w_i$,
then by Lemma \ref{lem:wti-constr}
$x_i \le \frac12 \wti$, and $\wti - x_i \ge \frac12 \wti$.
Now
\begin{eqnarray*}
\sum_i w_i p_i & = & \sum_{x_i\le w_i/3} w_i p_i + \sum_{x_i\ge w_i/3} w_i p_i \\
& \le & \sum_i |\wti -x_i| p_i +\sum_i 3x_i p_i \le \phi + M.
\end{eqnarray*}
\end{proof}

\begin{lemma}
\label{lem:misspend-bdd}
Let $i=\argmax\{ \frac{p^*_i} {p_i} \}$.
If $p_i < p^*_i$, then the total misspending,
$\sum_h |x_h - w_h| p_h$,
is at least $w_i (p^*_i - p_i)$. Similarly, let $j = \argmax\{ \frac{p_j} {p^*_j} \}$.
If $p_j > p^*_j$, then the total misspending
is at least $w_i (p_j - p^*_j)$.
\end{lemma}
\begin{proof}
By WGS, when the price of a good is reduced, the spending on that good only increases.
Consider reducing the prices from their equilibrium values, one by one, for each of the goods
whose price is below its equilibrium value. Each price reduction can be viewed as occuring in two stages,
first a spending neutral change, which leaves the spending on every good unchanged, and
second a stage which increases the spending on the good whose price is being reduced.
Clearly, there is no reduction of spending in total, on the goods whose prices are reduced.
Thus the excess spending is at least that which would be obtained were all the prices spending neutral.
Consider good $i$ in this spending neutral scenario.
At price $p_i$, the excess spending on good $i$ would be
$p_i(x_i - w_i) = p^*_i w_i - p_i w_i$, as claimed.

The proof of the second claim is analogous.
\end{proof}

\begin{lemma}
\label{lem:init-price-conv-time-bdd}
Let $\phi_{\text{init}}$ be the initial value for $\phi$.
If $\alt > 1$ and $|\wti - w_i| \le \alr w_i = 4 \alf w_i$ for all $i$, where
\[
\alr \le \frac{(1 - \loa) (1 - 6(1 + 2\alt))} {12\gam(1 + 2\alt)}
  \]
and $\gam = \max_i \frac{M}{p^*_i w_i}$,
then after
\[
D = \frac{16(1 + \alt)}{\loa} \log\frac{2 \phi_{\text{init}}} {(1 - \loa)\min_i p^*_i w_i}
\]
days all demands satisfy $x_i \le 2 w_i$ henceforth.
\end{lemma}
\begin{proof}
We want to show that $\phi \ge 2(1 + 2\alt) \sum_i |\wti - w_i| p_i$, while
some $p_i \notin [\frac12 p^*_i , \frac32 p^*_i ]$,
for then we can apply Corollary \ref{cor:phi-rate-of-change} to obtain that
$\frac{d\phi}{dt}  \leq - \frac{\loa}{8(1 +\alt)}\phi $ (as $\chi_i=0$ for all $i$
in the present setting).
The condition on $\phi$ holds if $\phi \ge 2(1 + 2\alt) \alr \sum_i w_i p_i$,
and by Lemma \ref{lem:supply-cost-bdd-war-bdd}, this holds if
$\phi \ge 6(1 + 2\alt) \alr (\phi + M)$; in turn, it suffices that
\[
\phi \ge \frac{6(1 + 2\alt) \alr M}{1 - 6(1 + 2\alt) \alr}.
\]
Now
\begin{eqnarray*}
\phi & \ge &  \sum_i [(1 - \loa)|x_i - \wti| + \alt |\wti - w_i|] p_i \\
& \ge & \sum_i  [(1 - \loa)(|x_i - w_i| - |\wti - w_i|) + \alt |\wti - w_i|] p_i \\
& \ge & (1 - \loa) \sum_i (|x_i - w_i|p_i ~~~~\text{as $\alt > 1$} .
\end{eqnarray*}

Suppose there is a good $h$ such that $p_h \notin [\frac12 p^*_h , \frac32 p^*_h]$;
then  let $h =\argmax\{\frac{p_i} {p^*_i}, \frac {p^*_i} {p_i} \}$.
By Lemma \ref{lem:misspend-bdd},
$\sum_i |x_i - w_i| p_i \ge w_h |p_h -p^*_h| > \frac12 w_h p^*_h \ge \frac{1}{2\gam}M$.
Thus $\phi \ge \frac12 w_h p^*_h \ge \frac{1}{2\gam}M$.

So the condition holds if
\[
\frac{1 - \loa} {2\gam} M
  \ge  \frac{6(1 + 2\alt)}{1 - 6(1 + 2\alt) \alr} \alr  M.
\]
It suffices that
\[
\alr \le \frac{(1 - \loa) (1 - 6(1 + 2\alt))} {12\gam(1 + 2\alt)}
\]
to ensure that $\frac{d\phi} {dt} \le - \frac{\loa}{8(1 +\alt)}\phi $.
Thus after at most
\[
D = \frac{16(1 + \alt)}{\loa}\log\frac{\phi_{\text{init}}} {\frac{1 - \loa}{2}\min_i w_i p^*_i  }
\]
days, $\phi \le  \frac12 (1 - \loa) p^*_i w_i$ for all $i$.

Suppose, for a contradiction, that after $D$ days some $x_h > 2w_h$. Then
$p_h(x_h - w_h) > \frac12 p^*_h w_h$.
Now
$\phi  \ge  (1 - \loa) \sum_i |x_i - w_i| p_i \ge (1 - \loa) (x_h - w_h) p_h
> \frac12 (1 - \loa) w_h p^*_h \ge \phi$.
\end{proof}

\Xomit{
\begin{definition}
\label{def:d-bound}
Suppose the prices are always $c$-demand bounded.
Let $f=f(c)$ be the corresponding $f$-bound on the prices and given the prices
are $f$-bounded let $d(f)$ be the demand bound (conceivably, $d >> c$).
\end{definition}
}

We now obtain a bound on the depletion of the warehouse stock while the prices are decreasing
from their initial values to a value that guarantees 2-demand bound henceforth,
which allows Lemmas \ref{lem:war-too-full-impr} and \ref{lem:war-too-empt-impr}
to be applied.

\begin{lemma}
\label{lem:warhs-bdd-frst-phse}
During the time that $\phi$ reduces to at most $\frac12 (1 - \loa) w_h p^*_h \ge \phi$,
for each $i$,
the stock of warehouse $i$ reduces by at most
\[
(d - 1) D w_i = (d - 1) w_i \frac{16(1 + \alt)}{\loa}\log\frac{\phi_{\text{init}}}{\frac{1 - \loa}{2}\min_i w_i p^*_i}.
\]
\end{lemma}
\begin{proof}
Each day, the stock can reduce by at most $(d-1)w_i$.
By Lemma \ref{lem:init-price-conv-time-bdd}, the stated reduction in potential occurs within
$D$ days.
\end{proof}

Now we are ready to bound how large a warehouse suffices, given an upper bound
on $\phi_{\text{init}}$.
Recall that we view the warehouse as having 8 equal sized zones of fullness.

\Xomit{
\begin{definition}
\label{def:wrhs-zones}
The four zones above the half way target are called the \emph{high} zones, and the other four
are the \emph{low} zones.
Going from he center outward, the zones are called the \emph{safe} zone, the \emph{inner buffer},
the \emph{middle buffer}, and the \emph{outer buffer}.
\end{definition}
}

Note that if the $i$th warehouse were completely full or empty, then
$|\wti - w_i| = \frac12 \kpi c_i$, and as we assumed that $|\wti - w_i| \le \alr w_i$, this
implies that $\frac12 \kpi c_i \le \alr w_i$.
Note that Constraint \ref{cst:wti-bdd}
implies $\frac12 \kpi c_i \le \frac{1}{3} w_i$, which is ensured
if $\alr \le \frac13$.

\medskip
\noindent
{\bf Theorem \ref{lem:good-wrhs}.}~%
\emph{%
Suppose that the prices are always $f$-bounded and let $d= d(f)$.
Also suppose that each price is updated at least once a day.
Suppose further that the warehouses are initially all in their safe or inner buffer zones.
Finally, suppose that $\lmbd \left( 1 + \frac{1} {\alf} \right) \le \frac12$.
Then the warehouse stocks never go outside their outer buffers \emph{(}i.e.\ they never overflow
or run out of stock\emph{)} if
$\frac{\alf}{\kpi} = \frac{c_i}{8w_i}  \ge
  \max \left\{ (d-1) D, 2\left( 1 + \frac{4}{\alf} \right)\frac{f}{\lmbd} + \frac{8 \lmbd} {\alf} \right\}$;
furthermore,
after $D + 2\left( 1 + \frac{4}{\alf} \right) \frac{f}{\lmbd} + \frac{8 \lmbd} {\alf} + \frac{8}{\kpi}$
days the warehouses will be in their safe or inner buffer zones thereafter, where
\[
D = \frac{16(1 + \alt)}{\loa}\log\frac{\phi_{\text{init}}} {\frac{1 - \loa}{2}\min_i w_i p^*_i  },
\]
and $\phi_{\text{init}}$ is the initial value of $\phi$.
}

\emph{%
If the fast updates rule is followed, then it suffices to have
$\frac{\alf}{\kpi} = \frac{c_i}{8w_i}  \ge  2\left( 1 + \frac{4}{\alf} \right)f + \frac{8 \lmbd} {\alf} $,
and then after $\left( 1 + \frac{4}{\alf} \right) f + \frac{8 \lmbd} {\alf} \lmbd + \frac{8}{\kpi}$
days the warehouses will be in their safe or inner buffer zones thereafter.
}

\medskip

Again, substituting our bounds on $\lmbd, \alo, \alf$ suggests $D$ is rather large.
The result should be viewed as indicating that there is a bound on the
needed warehouse size, but not that this is a tight bound (in terms of constants).

\begin{proof}
We will consider warehouse $i$.
After an initial $D$ days (as defined in Lemma \ref{lem:init-price-conv-time-bdd}),
the demands are henceforth 2-bounded.
In this time, by Lemma \ref{lem:warhs-bdd-frst-phse}, the warehouse stocks can decrease by at most
$(d - 1) D w_i \le \frac18 c_i$.
As a result the warehouses are all in their middle buffers or nearer the center at this point.

We show that henceforth the tendency is to improve, i.e. move toward the safe zone,
but there can be fluctuations of up to one zone width. The result is that every
warehouse remains within its outer buffer, and after a suitable time they will all
be in either their inner buffer or safe zone.

As it is $f$-bounded, price $p_i$ can decrease by at most $e^{-2(f/\lmbd)\lmbd}$ from its initial value.
If $s_i(t)$ is in middle buffer at time $t$, then by Lemma \ref{lem:war-too-full-impr}
(taking $c=0$),
within $2\left( 1 + \frac{2}{\alf} \right) \frac{f}{\lmbd}$ 
days the value of $s_i$ will have returned to $s_i(t)$ or remained
below this value.
During this interval a decrease of $\omega$ in the warehouse stock takes
at least $\omega$ days to replenish; consequently the stock can increase by at
most $2\left( 1 + \frac{2}{\alf} \right) \frac{f}{\lmbd} w_i$. 
It follows that the warehouse never runs out of stock as each portion of
the outer buffer has width at least 
$ \frac{1}{8} c_i \ge 2\left( 1 + \frac{2}{\alf} \right) \frac{f}{\lmbd} w_i$.
Again, by  Lemma \ref{lem:war-too-full-impr}, taking $c=\frac{c_i}{4 w_i}$, after
$2\left( 1 + \frac{2}{\alf} \right) \frac{f}{\lmbd} + \frac{c_i}{2\alf w_i}$ 
days it will have reached the safe zone.
The preceding argument shows that it can never go above the inner buffer henceforth.

We apply the same argument to the low zones using Lemma \ref{lem:war-too-empt-impr}.
Now, with $c=0$, we see that within 
$2\left( 1 + \frac{4}{\alf} \right) \frac{f}{\lmbd}  + \frac{8 \lmbd} { \alf} \lmbd$
days $s_i$ returns to the inner buffer. Likewise after
$2\left( 1 + \frac{4}{\alf} \right) \frac{f}{\lmbd} + \frac{8 \lmbd} { \alf} + \frac{c_i}{\alf w_i}$ days
it will never again go below the inner buffer.

If the fast update rule is followed then the first phase in which the potential reduces is not needed,
and then the warehouse will never go beyond its middle buffer (so in fact the outer buffers are then
not needed). The analysis of the second phase is as before,
except that the time needed to first reach the inner buffer has a bound that is half as large.
\end{proof}

\noindent
{\bf Comment}. We note that the constraints on $\lmbd$, $\alf$, $c_i$ and $\kpi$
can be met in turn.
Lemma \ref{lem:init-price-conv-time-bdd}
constrains $\alr = 4 \alf$ (while $\lmbd$ appears in this constraint
we already know from Lemma \ref{lem:war-progr-bdd-dem} that we want $\loa \le \frac12$, which bound can safely be used).
In turn, the bound $\lmbd\left( 1 + \frac{1}{\alr} \right) \le \frac 12$ and
Theorem \ref{lem:war-progr-bdd-dem}
constrain $\lmbd$.
Finally, $c_i$ is upper bounded (and hence $\kpi$ lower bounded)
by the conditions in Theorems \ref{lem:war-progr-bdd-dem} and \ref{lem:good-wrhs}.

\subsection{The Effect of  Inaccuracy}
\label{sec:non-acc}

We begin by analyzing Case (i), where the parameter $\rho$ is not known to the price-setters.

\begin{lemma}
\label{lem:non-acc-cont-progress}
When $p_i$ is updated,
if Assumption \ref{ass:revised-dem-bdd} and Constraint \ref{cst:wti-bdd} hold,
$\frac{\alt}{2} + \alpha_1  \max\{\frac32,d-1\} \le 1$, and
$\loa + \frac43 \lmbd \left(1 + \frac{2E d}{1-\lE}  + \frac12 \alt \right) \le 1$, then
$\phi_i$ increases by at most
$\frac43\lmbd \rho (2b +\kpi)\left(1 + \frac{2E d}{1-\lE}  + \frac12 \alt \right) w_i p_i$.
\end{lemma}
\begin{proof}
We use Lemma \ref{lem:war-updates-only-help-templ}.
Were $\Del_i p_i$ accurate, Corollary \ref{cor:war-updates-only-help} assures us that
$\phi$ only decreases.
With a toward $\wti$ update $|\Del_i  p_i|$ is too small by
at most $\lmbd \rho (2b +\kpi)p_i$.
Thus the increase in $\phi$ is at most
$\wti \lmbd \rho (2b +\kpi)p_i$;
using the bound $\wti \le \frac43 w_i$ from Lemma \ref{lem:wti-constr},
yields that this is at most
$\frac43 \lmbd \rho (2b +\kpi) w_i p_i$.

With an away from $\wti$ update
$|\Del_i p_i|$ may be too large by at most $\lmbd \rho (2b + \kpi)p_i$.
Thus the increase in $\phi$ is at most
$\wti \lmbd  \rho p_i \left(1 + \frac{2E d}{1-\lE}  + \frac12 \alt \right)(2b + \kpi)$,
which is at most
$\frac43 \lmbd  \rho \left(1 + \frac{2E d}{1-\lE}  + \frac12 \alt \right) (2b + \kpi) w_i p_i $.
\end{proof}

The result of Lemma \ref{lem:war-cont-progress} holds unchanged,
namely that $d\phi/dt\le -\frac{\kmin(\alt -1)}{2} \phi$ except
when a price increase occurs. We will deduce that so long as $\phi$ is large enough it will
decrease by a $1 - \frac{\kmin(\alt -1)}{8})$ rate daily.

The following preliminary lemma is helpful.

\begin{lemma}
\label{lem:supply-cost-bdd}
$\sum_i w_i p_i \le \phi /(1 - \loa) +M$, where $M$ is the daily
supply of money.
\end{lemma}
\begin{proof}
Note that $M=\sum_i x_i p_i$ and that
$\phi_i \ge |\wti - x_i|(1 - \loa) + \alt |\wti-w_i|
\ge |\wti - x_i|(1 - \loa) + |\wti-w_i|$.
Now
\begin{eqnarray*}
\sum_i w_i p_i  & \le  & \sum_i p_i(x_i + |\wti -x_i| + |w_i -\wti|)
\\
 & \le &  \sum_i x_i p_i  + 2\sum_i p_i[|\wti - x_i|(1 - \loa) + |\wti-w_i|]/(1 - \loa) \\
 & \le  & M + \phi/(1 - \loa).
\end{eqnarray*}
\end{proof}

Before obtaining a bound on the rate of decrease of $\phi$, we state a more
general result which will be used for a series of results of this nature.

\begin{lemma}
\label{lem:genrl-progr-bdd-dem}
Let $D=[\tau,\tau+1]$ be a day.
Let $\phi^+=max_{t\in D} \phi(t)$.
Suppose that except when a price increase occurs,
$\frac{d\phi}{dt} \le -\nu \phi + \gamma$, for  parameters $0 < \nu \le 1$, $\gamma \ge 0$.
Further suppose that the price updates over the course of day $D$
collectively increment $\phi$ by at most
$\mu(a_1\phi^+ + a_2\phi(\tau) +a_3 M)$, where $M$ is the daily supply of money,
and $\mu,a_1,a_2,a_3\ge 0$ are suitable parameters.

If
$\mu\left[ \frac{a_1(1+\mu a_2)}{1-\mu a_1} + a_2 \right] \le \frac{\nu}{8}$ and
$\phi(\tau) \ge \frac{8\mu}{\nu} a_3 M \frac{1}{1-\mu a_1}$,
then
$\phi(\tau+1) \le \left(1 - \frac{\nu}{4}\right)\phi(\tau) +\frac{\gamma}{1-\mu a_1}$.
\end{lemma}
\begin{proof}
$\phi(\tau+1) \le e^{-\nu} \phi(\tau) + \gamma + \mu(a_1\phi^+ + a_2\phi(\tau) +a_3 M)$ as
$\frac{d\phi}{dt} \le -\nu \phi + \gamma$.
So
$\phi(\tau+1) \le (1 - \nu /2) \phi(\tau) + \gamma + \mu(a_1\phi^+ + a_2\phi(\tau) +a_3 M)$,
if $\nu \le 1$ (using the power series expansion for $e^{-x}$).

Now
$\phi^+ \le \phi(\tau) + \mu(a_1 \phi^+ + a_2\phi(\tau) + a_3 M + \gamma)$;
hence
\[
\phi^+ \le \left[ \frac{\phi(\tau)(1 + \mu a_2)  +\mu a_3 M + \gamma}{1 -\mu a_1} \right].
\]
Thus
\begin{eqnarray*}
\phi(\tau+1) & \le & \phi(\tau)(1 - \nu/2)
+ \mu \phi(\tau)\left[\frac{a_1 (1 + \mu a_2)}{1 -\mu a_1} + a_2 \right]
+ \mu a_3 M \left[\frac{ \mu a_1}{1 -\mu a_1} + 1 \right] + \frac{\gamma}{1-\mu a_1}
\\
& \le & \phi(\tau)(1 - \nu/4) +\frac{\gamma}{1-\mu a_1}
\end{eqnarray*}
\[
\mbox{if } \mu \phi(\tau) \left[\frac{a_1(1 + \mu a_2)}{1 -\mu a_1} + a_2 \right] \le \frac {\nu}{8}\phi(\tau)
~~\mbox{i.e.}~~ \mu \left[ \frac{a_1(1 + \mu a_2)}{1 -\mu a_1} + a_2 \right] \le \frac{\nu}{8}
\]
\[
\mbox{and }
\mu a_3 M \left[\frac{ \mu a_1}{1 -\mu a_1} + 1 \right] \le \frac{\nu \phi(\tau)}{8}
~~\mbox{i.e.}~~
 \phi(\tau) \ge \frac{8\mu}{\nu} a_3 M  \frac{ 1}{1 -\mu a_1} .
\]
\end{proof}

\noindent
{\bf Theorem \ref{lem:non-acc-progr-bdd-dem}.}~%
\emph{%
If Assumption \ref{ass:revised-dem-bdd} and Constraint \ref{cst:wti-bdd} hold,
$\frac{\alt}{2} + \alpha_1  \max\{\frac32,(d-1)\} \le 1$,
$\loa + \frac43 \lmbd \left(1 + \frac{2E d}{1-\lE}  + \frac12 \alt \right) \le 1$,
$\frac{\kpi (\alt - 1)}{2} \le 1$,
$ 4\kpi(1 + \alt) \le \loa \leq \frac{1}{2}$,
and each price is updated at least once every day, and at most every $1/b$ days,
and if
$\phi \ge \frac{16 \mu M} {\kpi (\alt - 1)}
 \frac{1 -\loa}{1 - \loa - \mu}$
at the start of the day,
then $\phi$ decreases by at least a $1 - \frac{\kmin(\alt -1)}{8}$ factor by the end of the day,
where $M$ is the daily
supply of money
and $\mu = \frac43\lmbd \rho b (2b +\kmax)\left(1 + \frac{2E d}{1-\lE}  + \frac12 \alt \right)$,
supposing that $\kmin(\alt -1) \ge 16\mu/[1-\loa -\mu]$.
}
\begin{proof}
We will apply Lemma \ref{lem:genrl-progr-bdd-dem}.
By Lemma \ref{lem:war-cont-progress},
$d\phi/dt \le  - \frac{\kmin(\alt -1)\phi}{2}$, except when there is a price increase,
so $\nu = \frac{\kmin(\alt -1)}{2}$ and $\gamma = 0$.

We now determine the other parameters.
Consider a time interval of length $1/b$ days.
There is at most one price increase per good during this interval.
By Lemma \ref{lem:non-acc-cont-progress},
they increase the potential by at most
$\sum_i \frac43\lmbd \rho (2b +\kpi)\left(1 + \frac{2E d}{1-\lE}  + \frac12 \alt \right) w_i p_i =
 \frac43\lmbd \rho (2b +\kmax)\left(1 + \frac{2E d}{1-\lE}  + \frac12 \alt \right) \left[\frac{\phi^+}{(1 - \loa)} +M \right]$
(on using Lemma \ref{lem:supply-cost-bdd}
to bound $\sum_i w_i p_i$).
Thus $\mu = \frac43\lmbd \rho b (2b +\kmax)\left(1 + \frac{2E d}{1-\lE}  + \frac12 \alt \right)$,
$a_1=1/(1 - \loa)$, $a_2=0$, and $a_3=1$.
\end{proof}

To analyze Case (ii) we will use the potential
\[
\phi_i = p_i[\span(x_i, \xbi, \wti) -
4\kpi(1 + \alt) (t - \tau_i) | \xbi- \wti |
+ \alpha_2 | \wti - w_i |].
\]
By Lemma \ref{lem:war-cont-progress}, $\frac{d\phi_i}{dt}
 \leq - \frac{\kpi(\alt - 1)}{2} \phi_i$ at any time when no price update is occurring
(this follows by setting $\lmbd =  4\kpi(1 + \alt)$ in the statement of  Lemma \ref{lem:war-cont-progress}).

We say that an attempted update of $p_i$ is a \emph{null update} if $p_i$ is unchanged
due to $\zbi^r$ being too small.

\begin{lemma}
\label{lem:non-acc-cont-progress-ii}
When $p_i$ undergoes a null update,
$\phi_i$ increases by at most $8\kpi(1 + \alt)  \rho (2b +\kpi) w_i p_i(t-\tau_i)$,
where $t$ is the current time and $\tau_i$ the time of the last update or
attempted update.
\end{lemma}
\begin{proof}
The cost for a null update is $4\kpi(1 + \alt) (t - \tau_i) | \xbi- \wti |p_i$.
As this is a null update, $| \xbi- \wti | \le 2(2b+\kpi)\rho w_i$,
so we see that the cost is at most
$ 8\kpi(1 + \alt)  (2b+\kpi)\rho w_i p_i   (t - \tau_i) $.
\end{proof}

\begin{lemma}
\label{lem:non-acc-cont-progress-ii'}
When $p_i$ undergoes an actual update,
if  Assumption \ref{ass:revised-dem-bdd} and Constraint \ref{cst:wti-bdd} hold,
$\frac{\alt}{2} + \alpha_1  \max\{3,2(d-1)\} \le 1$,
$\frac{\kpi(1 + \alt)}{2} \le 1$,
and
$\loa + 2\lambda \left(1 + \frac{2E d}{1-\lE}  + \frac12 \alt \right) \leq 1$, then
$\phi_i$ only decreases.
\end{lemma}
\begin{proof}
We use Lemma \ref{lem:war-updates-only-help-templ}.

Since $|\zbi^c-\zbi^r| \ge \frac12 \zbi^r$,
we know that $\frac12 \zbi^r \le \zbi^c \le \frac32 \zbi^r$.
Consequently, the magnitude of the calculated $\Del_i p_i$,
ranges between $\frac12$ and $\frac32$ times the ideal value,
the value that would be
obtained were $\zbi^r = \zbi^c$.

In the bounds of Lemma \ref{lem:war-updates-only-help-templ},
we need to replace the ideal
value of $\Del_i p_i$ by the worst case for the calculated value.
This yields the following constraints to ensure $\phi_i$ only decreases.\\
\\
Case 1. A toward $\wti$ update:
\\
Case 1.1. $|\xbi - \wti| \le w_i$.

Here $|\Del_i p_i| \ge \frac 12 \lmbd |\xbi - \wti|p_i/w_i$.
So $\phi$ only decreases if
\begin{equation}
\label{eqn:cstrt-1v-lem-updates-only-help}
\loa |\xbi - \wti| p_i  \leq \frac12 (1 - \frac12\alt)\wti \lmbd |\xbi - \wti|p_i/w_i.
\end{equation}
By Lemma \ref{lem:wti-constr}, $\wti/w_i \ge \frac23$,
so $\frac32\alpha_1 + \frac14 \alt \le \frac12$ suffices.
\\
Case 1.2. $|\xbi - \wti| > w_i$.
\\
Here $|\Del_i p_i| \ge \frac 12 \lmbd p_i$.
So $\phi$ only decreases if
\begin{equation}
\label{eqn:cstrt-2v-lem-updates-only-help}
 \loa (d-1)\wti  p_i  \leq
\frac12 (1 - \frac12  \alt) \wti \frac 12 \lmbd p_i.
\end{equation}
The condition $\alpha_1(d-1) + \frac14 \alt \le \frac12$ suffices.
\\
Case 2. An away from $\wti$ update:
\\
Here $|\Del_i p_i| \le \frac 32 \lmbd p_i |\xbi - \wti|/w_i$.
So $\phi$ only decreases if
\begin{equation}
\label{eqn:cstrt-4v-lem-updates-only-help}
\left[ \wti + \frac{2E}{1-\lE} x_i  +\frac12 \alt \wti \right] \frac 32 \lmbd p_i |\xbi - \wti|/w_i
\leq (1 - \loa) p_i | \xbi - \wti|.
\end{equation}
By Lemma \ref{lem:wti-constr},
$ \wti \le \frac43 w_i$,
so this is subsumed by
$\loa + 2\lambda \left(1 + \frac{2E d}{1-\lE}  + \frac12 \alt \right) \leq 1$.
\end{proof}

\noindent
{\bf Theorem \ref{lem:non-acc-progr-bdd-dem-2}.}~%
\emph{%
If Assumption \ref{ass:revised-dem-bdd} and Constraint \ref{cst:wti-bdd} hold,
$\frac{\alt}{2} + \alpha_1  \max\{3,2(d-1)\} \le 1$,
$\loa + 2\lambda \left(1 + \frac{2E d}{1-\lE}  + \frac12 \alt \right) \leq 1$,
$ 4\kpi(1 + \alt) \le \loa \leq \frac{1}{2}$ for all $i$,
$\frac{\kpi (\alt - 1)}{2} \le 1$,
$\mu \left[ \frac{1 + \mu/(1 - \loa)} {1 - \mu/(1 - \loa)} + \frac{1}{(1 - \loa)} \right]
\le \frac{\kpi(\alt - 1)}{2}$,
if $\phi \ge \frac{32\mu M} {[1 - \frac{\mu}{1 -\loa}] [(\kmin(\alt - 1)]}$
at the start of the day,
where $\mu = 8\kmax(1 + \alt) (2b +\kmax) \rho $,
and if each price is updated at least once every day,
then $\phi$ decreases by at least a $1 - \frac{\kmin(\alt -1)}{8}$ factor daily.
}
\begin{proof}
Again, we apply Lemma \ref{lem:genrl-progr-bdd-dem}.

By Lemma \ref{lem:war-cont-progress},
$d\phi/dt \le  - \frac{\kmin(\alt -1)}{2}$, except when there is a price increase,
so $\nu = \frac{\kmin(\alt -1)}{2}$ and $\gamma = 0$.

We now determine the other parameters.

To total the increase due to null updates, we spread the potential increase due to
each null update uniformly across the
time interval between the present null update and the predeceding update
to the same price (null or otherwise).
By Lemma \ref{lem:non-acc-cont-progress-ii},
the potential increase due to a null update to $p_i$ at time $t$ is
at most $8\kmax(1 + \alt) (2b +\kmax) \rho w_i p_i(t)\Del= \mu w_i p_i(t)\Del$,
where $\Del$ is the length of the time interval over which the increase is being spread.
Thus the total potential increase due to null updates during $D$ is bounded by
$\mu \int_D \sum_i w_ip_i dt$ plus the following term which covers
the intervals that extend into the day preceding $D$:
$\mu \sum_i w_ip_i(\tau)$.

Then, using Lemma \ref{lem:supply-cost-bdd} to bound $\sum_i p_i w_i$,
shows that the potential increase over $D$ is bounded by
$\mu[M+\phi^+/(1-\loa)]  + \mu[M+\phi(\tau)/(1-\loa)]$.
So we can set $a_1=a_2=/1/(1 - \loa)$ and $a_3=2$.
\end{proof}

\noindent
{\bf Remark}.
To combine the fast update rule with the rule of at most one update every $1/b$ days,
we restate them as follows: an update to $p_i$ occurs after a $w_i$ reduction to
the corresponding warehouse stock, and if that does not occur, following
a period of between $1/b$ and 1 day since the last update to $p_i$, null or
otherwise.

Small changes will be needed in the proof of Theorem \ref{lem:non-acc-progr-bdd-dem}.

\subsection{Discrete  Goods and Prices}
\label{sec:discrete}

The analysis of the discrete case is similar to that for inaccurate data.
We begin by examining the sources of ``error'' in the discrete setting.

\smallskip
\noindent{\bf Sources of Error}.

As $|\xbi^I - \ybi| \le 1$, and $|\wti^A - \wti^I| \le \kpi$
the previous bound of  $2(b + \kpi)\rho w_i$
on the error in calculating $\zbi$ ($=\ybi -\wti$ here)
is replaced by $1 + \kpi$.
This amounts to setting $\rho = \min_i (1 + \kpi)/[2(b + \kpi) w_i]$.
Note that this implies that a price update occurs only if $|\xbi^A -\wti| \ge 2(1 + \kpi)$.
\footnote{We could enforce a condition $|\xbi^A - \ybi| \le \frac1{(1 + \kpi)}$
(or any other convenient
positive bound), and then
price updates would occur so long as $|\xbi^A -\wti| \ge 1$;
however, the bound on the elasticity for the $y_i$ demands would increase correspondingly.}

We also need to take account of the fact that the prices $p_i$
are integral, but the calculated updates need not be.
To avoid overlarge updates, we conservatively round down the
the magnitude of each update.
Note that this reduces the value of the update by at most a factor of 2.
This introduces a second source of error,
which we need to incorporate in the analysis.
In addition, this has the following implication:
no price can be less than $1/\lmbd$,
for a smaller price would never be updated; furthermore, the price $1/\lmbd$
may not be reduced even if the update rule so indicated.

\Xomit{
\smallskip
\noindent{\bf Indivisibility Parameters}.
We measure the indivisibility of the market in terms of two parameters, $r$ and $s$.
$r=M/\sum_i w_i$, where $M$ is the daily supply of money; it provides an upper bound
on the weighted average price for an item at equilibrium. This can be thought of as the
granularity of money at the equilibrium.
$s=\min_i w_i$ is the minimum size for the daily supply for any item,
and thus indicates the granularity of the least divisible good.

\smallskip

Next, we restate Constraint \ref{cst:wti-bdd} and Lemma \ref{cst:wti-bdd},
replacing $\wti$ with $\wti^I$.

\begin{constraint}
\label{cst:wti-bdd-disc}
$|\wti^I - w_i| \le \frac13 w_i$.
\end{constraint}
}

\begin{lemma}
\label{lem:wti-constr-disc}
If Constraint \ref{cst:wti-bdd-disc} holds, then
$\frac34 \wti^I \le w_i \le \frac32 \wti^I$ and
$|\wti^I - w_i| \le \frac12 \wti^I$.
In addition,
Constraint \ref{cst:wti-bdd-disc} holds for all possible warehouse contents if
$\kpi\max\{|c_i - s_i^*|,s_i^*\}\le \frac13 w_i$, and if
in addition $s_i^* = c_i/2$, the condition becomes $\kpi \le \frac23\frac{w_i}{c_i}$.
\end{lemma}

The analysis will follow scenario (ii) from Section \ref{sec:non-acc},
since ``$\rho$'' is known.
First, we note that Lemma \ref{sec:non-acc}
still holds, as $M \ge \sum_i x_i p_i$.

We use a potential $\phi_i$ expressed in terms of $y_i$:
\[
\phi_i = p_i \left[\span(y_i, \ybi, \wti^I) -
4\kpi(1 + \alt) (t - \tau_i) | \ybi- \wti^I |
+ \alpha_2 | \wti^I - w_i | \right].
\]

We restate Lemma \ref{lem:war-updates-only-help-templ}
in terms of the demands $y_i$ in the corresponding divisible market.

\begin{lemma}
\label{lem:ind-updates-only-help-templ}
If Constraint \ref{cst:wti-bdd-disc} holds and $y_i \le d\wti^I$ for all $i$,
then when $p_i$ is updated, $\phi$ increases by at most the following:
\\
(i) With a toward $\wti$ update:
\begin{equation}
\label{eqn:cstrt-1-5-lem-updates-only-help}
\loa |\ybi - \wti^I| p_i + \frac12\alt \wti^I|\Del_i p_i|  - \wti^I|\Delta_i p_i|.
\end{equation}
(ii) With an away from $\wti$ update:
\begin{equation}
\label{eqn:cstrt-2-5-lem-updates-only-help}
\left(1 + \frac{2E d}{1-\lE} \right) \wti^I |\Delta_i p_i|
 + \frac12\alt \wti^I|\Del_i p_i|
- (1- \loa )  p_i| \overline{y}_i  - {\wti^I}|.
\end{equation}
\end{lemma}
\begin{proof}
The only changes were to replace $\xbi$ by $\ybi$ and $\wti$ by $\wti^I$
which reflect the corresponding changes
to the potential function.
\end{proof}

Again, we say that an update is null if it leaves the price unchanged.
Its cost is bounded in the following lemma.
\begin{lemma}
\label{lem:null-cont-progress-disc}
When $p_i$ undergoes a null update,
$\phi_i$ increases by at most
$6(1 + \kpi)\kpi(1 + \alt) (t - \tau_i)  w_i/\lmbd$ if  $| \xbi^A - \wti^A | \ge 2(1 + \kpi)$
and by at most $12(1 + \kpi)\kpi(1 + \alt) (t - \tau_i) p_i$
if $| \xbi^A - \wti^A | < 2(1 + \kpi)$,
where $t$ is the current time and $\tau_i$ the time of the last update or
attempted update.
\end{lemma}
\begin{proof}
The cost for a null update is $4\kpi(1 + \alt) (t - \tau_i) | \ybi- \wti^I |p_i$.
As this is a null update, either $\lmbd| \xbi^A- \wti^A |p_i/w_i < 1$ or
$| \xbi^A - \wti^A | < 2(1 + \kpi)$.
Now $ | \ybi- \wti^I | \le | \xbi^I- \wti^I | + 1 = | \xbi^A - \wti^A | + 1$.
If $| \xbi^A - \wti^A | \ge 2(1 + \kpi)$,
$ | \ybi- \wti^I | \le \frac32 (1 + \kpi) | \xbi^A - \wti^A | $.
So we see that the cost is at most either
$ 6(1 + \kpi)\kpi(1 + \alt) (w_i/\lmbd)   (t - \tau_i) $ if  $| \xbi^A - \wti^A | \ge 2(1 + \kpi)$
and  $12(1 + \kpi)\kpi(1 + \alt) (t - \tau_i) p_i$ otherwise.
\end{proof}

Next, we bound $\sum_i w_i$ and $\sum_i p_i$.

\begin{lemma}
\label{lem:disc-bound-w-p}
If $w_i \ge 6$ for all $i$,
$\sum_i w_i =  M/r$, and
$\sum_i   p_i \le \frac{3}{s}\left( \frac{\phi}{ (1 - \loa)} + M \right)$.
\end{lemma}
\begin{proof}
The first claim simply restates the definition of $r$.
For the second claim, we argue as follows.
\[
\sum_i p_i  = \sum_{y_i^I \le w_i/3} p_i + \sum_{y_i^I > w_i/3} p_i
\le \sum_{y_i^I \le w_i/3}\frac{ |\wti^I - y_i^I|p_i} {w_i/3} + \frac3s \sum_{y_i^I > w_i/3} y_i^Ip_i
\]
The first term in the inequality following from Lemma \ref{lem:wti-constr},
as $\wti^I \ge \frac23 w_i$.
The above total is bounded by $\frac{3}{s}\left(\frac{\phi_i}{(1 - \loa)} +M \right)$.
\end{proof}

\begin{lemma}
\label{lem:non-acc-cont-progress-disc}
When $p_i$ undergoes an actual update,
if Constraint \ref{cst:wti-bdd-disc} holds and $y_i \le d\wti^I$ for all $i$,
$\frac{\alt}{2} + \alpha_1  \max\{\frac92,2(d - 1))\} \le 1$,
$\loa + \frac83 \lmbd \left(1 + \frac{2E d}{1-\lE}  + \frac12 \alt \right) \le 1$,
and each $w_i \ge 2$, then
$\phi_i$ only decreases.
\end{lemma}
\begin{proof}
In the same way that the proof of Lemma \ref{lem:non-acc-cont-progress} is based on
Lemma \ref{lem:war-updates-only-help-templ},
we use Lemma \ref{lem:ind-updates-only-help-templ}.

Since the potential function is based on $y_i$ and $\ybi$ values,
we need to bound $\Del_i p_i$ in terms of $\ybi$.
In fact, it is calculated in terms of $\xbi$.

We note that $ |\xbi^A - \wti^A| - 1 - \kpi \le |\ybi -\wti^I| \le |\xbi^A - \wti^A| + 1 + \kpi$.

\noindent
Case 1. A toward $\wti^I$ update.
\\
Case 1.1. $|\xbi^A - \wti^A| \le w_i$.

Here $|\Del_i p_i| \ge \frac 12 \lmbd |\xbi^A - \wti^A|p_i/w_i$,
as by assumption $\lmbd |\xbi^A - \wti^A|p_i/w_i \ge 2$ for an actual update.
So $\phi$ is non-increasing as a result of the update if
\begin{equation}
\label{eqn:cstrt-1v-5-lem-updates-only-help}
\loa |\ybi - \wti^I| p_i  \leq (1 - \frac12\alt)\wti^I\cdot  \frac12 \lmbd |\xbi^A - \wti^A| \frac{p_i}{w_i}.
\end{equation}
By Lemma \ref{lem:wti-constr}, $\wti^I/w_i \ge \frac23$, and as already noted
$|\ybi -\wti^I| \le |\xbi^A - \wti^A| + 1 + \kpi$,
so
$\loa (|\xbi^A - \wti^A| + 1+ \kpi) p_i  \leq  (1 - \frac12\alt)\cdot \frac13 \lmbd |\xbi^A - \wti^A| p_i $
suffices.
As $|\xbi^A - \wti^A| \ge 2(1 + \kpi)$ for an actual update,
$|\xbi^A - \wti^A|  + 1 + \kpi \le \frac32 |\xbi^A - \wti^A| $,
and $\frac92 \alo + \frac12\alt \le 1$ suffices.
\\
Case 1.2. $|\xbi^A - \wti^A|  > w_i$.
\\
Here $|\Del_i p_i| \ge \frac 12 \lmbd p_i$.
So $\phi$ is non-increasing due to the update if
$\loa (d-1)\wti^I p_i \le   (1 - \frac12\alt)\wti^I \cdot \frac12 \lmbd p_i$.
In turn, $2\alo  (d-1) +  \frac12\alt \le 1$ suffices.

\noindent
\\
Case 2. An away from $\wti^I$ update.
\\
Here $|\Del_i p_i| \le  \lmbd p_i |\xbi^A - \wti^A|/w_i$.
Recall that an update occurs only if $|\xbi^A - \wti^A| \ge 2(1 + \kpi)$, and hence
$|\ybi - \wti^I| \ge \frac 12 (\xbi^A - \wti^A)$.
Thus $|\Del_i p_i| \le  2\lmbd p_i |\ybi - \wti^I|/w_i$.
So $\phi$ is non-increasing due to the update if
\begin{equation}
\label{eqn:cstrt-4v-5-lem-updates-only-help}
2\lmbd p_i |\ybi - \wti^I| \left(1 + \frac{2E d}{1-\lE}  + \frac12 \alt \right) \frac{\wti^I}{w_i}
\leq (1 - \loa) p_i | \ybi - \wti^I|.
\end{equation}
As $\wti^I/w_i \le \frac 43$ by Lemma \ref{lem:wti-constr-disc},
this is subsumed by
$\loa + \frac83 \lambda \left(1 + \frac{2E d}{1-\lE}  + \frac12 \alt \right) \leq 1$.
\end{proof}

Next, we apply Lemma \ref{lem:war-cont-progress}
to get a bound on $d\phi/dt$.

\begin{corollary}
\label{cor:discr-cont-progress}
If $ 4\kpi(1 + \alt) \le \loa \leq \frac{1}{2}$,
$\frac{d\phi_i}{dt}
 \leq - \frac{\kpi(\alt - 1)}{2} \phi_i + \kpi p_i$
 at any time when no price update is occurring
\emph{(}to any $p_j$\emph{)};
this bound also holds for the one-sided derivatives when a
price update occurs.
\end{corollary}
\begin{proof}
This is Lemma \ref {lem:war-cont-progress},
with $x_i$ replaced by $y_i$ and $\wti$ by $\wti^I$.
So $\chi_i = \frac{d\wti^I}{dt}+\kpi(y_i - w_i)=-\kpi(x_i^I - w_i) +\kpi(y_i - w_i)
=-\kpi(x_i-y_i)$, and hence $0 \le -p_i \chi_i \le -\kpi p_i$.
\end{proof}

\noindent
{\bf Theorem \ref{lem:discr-progr-bdd-dem-2}.}~%
\emph{%
If Constraint \ref{cst:wti-bdd-disc} holds and $y_i \le d\wti$ for all $i$,
$\frac{\alt}{2} + \alpha_1  \max\{\frac92,2(d-1)\} \le 1$,
$\loa + \frac{8}{3} \lambda \left(1 + \frac{2E d}{1-\lE}  + \frac12 \alt \right) \le 1$,
each $w_i \ge 6$,
$ 4\kpi(1 + \alt) \le \loa \leq \frac{1}{2}$ for all $i$,\\
$s \ge \frac{48} {(\alt -1) (1 - \loa)}
\left[ 1 + 6(1 + \alt)  + (1 + \alt) \frac{1 - \loa + \frac{18}{s} \kpi (1 + \alt)  + \frac{3\kpi}{s}} { 1 - \loa - \frac{18}{s} \kpi (1 + \alt)}  \right]$,\\
if $\phi(\tau) \ge \frac{48 } {(\alt - 1)}\left[ (1 + \alt) \left( \frac{4} {\lmbd r} + \frac{24}{s} \right) + \frac1s \right]
 \frac{1 - \loa} {1- \loa - \frac{18}{s} \kpi (1  + \alt)} M$
at the start of the day,
and if each price is updated at least once every day,
then $\phi$ decreases by at least a $1 - \frac{\kmin(\alt -1)}{8}$ factor
over the course of the day.
}
\begin{proof}
We prove the result for a day $D$ starting at time $\tau$,
using Lemma \ref{lem:genrl-progr-bdd-dem},
as in the proof of Lemma \ref{lem:war-progr-bdd-dem}.

First, we note that by applying Lemma \ref{lem:war-cont-progress},
we can set $\nu= \kpi (\alt - 1)/2$ and $\gamma = \sum_i p_i$.

To total the increase due to null updates during $D$, we spread the potential increase due to
each null update uniformly across the
time interval between the present null update and the preceding update
to the same price (null or otherwise).
By Lemma \ref{lem:null-cont-progress-disc},
the potential increase due to a null update to $p_i$ at time $t$ is
at most $6\kpi(1 + \alt) \frac{w_i}{\lmbd}$ if $|\xbi - \wti| \ge 2$,
and $12\kpi(1 + \alt) p_i \Del$ otherwise,
where $\Del$ is the length of the time interval over which the increase is being spread.
Thus the total potential increase due to null updates during $D$ is bounded by
$6\kpi(1 + \alt) \int_D [\sum_i \frac{w_i}{\lmbd} +  \sum_{|\xbi - \wti| <2} 2 p_i(t) ] dt$
plus the following term which covers
the intervals that extend into the day preceding $D$:
$6\kpi(1 + \alt) \sum_i \frac{w_i}{\lmbd} + \sum_{|\xbi - \wti| < 2} 2 p_i (\tau)$.

The potential increase during $D$ due to the ``$\gamma$'' term is at most
$\kpi \sum_i p_i$.

As shown in Lemma \ref{lem:disc-bound-w-p}, $\sum_i w_i \le M/r$ and
$\sum_i p_i \le 3\phi/[s(1 - \loa)] + 3M/s$.
It follows that the potential increase over $D$ is bounded by
$6\kpi(1 + \alt)\{ M/(\lmbd r)  + 6M/s+3\phi^+/[s(1-\loa)]  + M/(\lmbd r) + 6M/s + 3\phi(\tau)/[s(1-\loa)]\} +\kpi [3\phi/[s(1 - \loa)] + 3M/s]$.
So we can apply Lemma \ref{lem:genrl-progr-bdd-dem}
with $\mu = 3\kpi$, $a_1= 6(1 + \alt)/[s(1 - \loa)]$,
$a_2= [6(1 + \alt) + 1]/[s(1 - \loa)]$
and $a_3= (1 + \alt) [4/(\lmbd r)  + 24/s]  + 1/s$.
The condition $\phi(\tau) \ge \frac{8\mu} {\nu} a_3 M \frac{1}{1-\mu a_1}$ in
Lemma \ref{lem:disc-bound-w-p} becomes
\[
\phi(\tau) \ge \frac{48 } {(\alt - 1)}\left[ (1 + \alt) \left( \frac{4} {\lmbd r} + \frac{24}{s} \right) + \frac1s \right]
 \frac{1 - \loa} {1- \loa - \frac{18}{s} \kpi (1  + \alt)} M.
\]
Similarly, the condition
$\mu\left[ \frac{a_1(1 + \mu a_2)} {1 - \mu a_1)} + a_2 \right] \le \frac{\nu}{8}$ becomes
\[
\frac{3 (1 + \alt)} {s(1 - \loa)} \left[ \frac{1 - \loa + \frac{18}{s} \kpi (1 + \alt)  + \frac{3\kpi}{s}} {1 - \loa - \frac{18}{s} \kpi (1 + \alt)}
 \right] +\frac{18 (1 + \alt)} {s(1 - \loa)} + \frac{3} {s(1 - \loa)}  \le \frac{\alt -1} {16};
\]
equivalently,
\[
s \ge \frac{48} {(\alt -1) (1 - \loa)}
\left[ 1 + 6(1 + \alt)  + (1 + \alt) \frac{1 - \loa + \frac{18}{s} \kpi (1 + \alt)  + \frac{3\kpi}{s}} { 1 - \loa - \frac{18}{s} \kpi (1 + \alt)}  \right].
\]
\end{proof}

\noindent
{\bf Theorem \ref{thm:discr-lwr-bdd}}.~%
\emph{%
In the discrete setting there are markets with $ \Omega (E/r) $
misspending at any pricing.
}
\begin{proof}
The construction uses one good plus money. Let $p_g$ by the price of
the good, $x_g$ the demand for the good and $n_\$$ the demand for
money. We define a continuous CES utility $[x_g (r + 1/2)]^{1 - 1/E} +
n_\$^{1 - 1/E}$ for the aggregate demand (it can be due to one buyer),
with a supply that at $p_g = r + 1/2$ spends $M/2$ money on
$\frac{M}{2r + 1}$ units of good $g$. It is not hard to check that at
$p_g = r$ and $r + 1$ there is already a change of $\theta(En_g / r)$
in the demand for good $g$.
\end{proof}

\subsubsection{Construction of Virtual Demands $\bf y$}

Our goal is to construct virtual demands $y_i$, with $x_i -1 < y_i \le x_i$,
defined for all $p$ for which $x_i \ge 1$.
Further the spending on $y_i$ will be a non-increasing function of $p_i$,
will obey WGS at integer price points (which is where it is defined),
and will have elasticity parameter $2E$.

$y_i$ is constructed as follows.

For each collection of prices $p_{-i}$, let $m_i(j)$ be the spending on
good $i$ when $p_i=j$; we also write it as $m_j$ for short. Consider the
sequence $m_1, m_2, m_3,\cdots$ and let $m_{l_1} = m_1, m_{l_2},
m_{l_3},\cdots$ be the maximal subsequence of non-increasing nonzero spending.
We create a new sequence $m'_1, m'_2, \cdots$ of virtual spending, such that
$m'_j\ge m'_{j+1}$ for all $j$.
First, we set $m'_{l_a}=m_{l_a}$ for all $a$.
If $m_{l_a}=m_{l_{a+1}}$,
then we set $m'_j=m_{l_a}$ for all $j\in (l_a, l_{a+1})$.
Also, if $x_i( l_{a+1}-1) \ge x_i( l_{a+1})+2$, then
we set $m'_j=m_{l_a}$ for all $j\in(l_a, l_{a+1})$.
Likewise, if $x_i(j) < x_i(j-1)$ for some $j\in(l_a, l_{a+1})$, we set
$m'_{l_a+1},m'_{l_a+2}, \cdots, m'_j= m_{l_a}$.
If $m_{l_b}$ is the last term in the non-increasing sequence, we set
$m'_j=m_{l_b}$ for all $j>l_b$ with $m_j>0$.
The only other alternative is that there is an $h\ge l_a$ with
$m'_{l_a}=m'_h$, and $x_i(h)=x_i(h+1) =\cdots=x_i(l_{a+1}-1)=x_i(l_{a+1})+1$.
Then, for $j\in(h, l_{a+1})$,
$y'_i(j)$ is defined to interpolate
between $y_i(h)$ and $y_i(l_{a+1})$ as follows:
$y'(j) =y_i(h)\left(\frac{h}{j}\right)^c$, where $c$ is given by
$y_i(h) \left(\frac{h}{l_{a+1}}\right)^c = y_i(l_{a+1})$.
Then $y_i(p) = y_i(p_{-i},j)$ is given by $y_i(p_{-i},j)=\max_{q_{-i} \leq p_{-i}}  \{y'_i(p_{-i},j), y_i(q_{-i}, j)\}$.
Note that $jy'(j) =hy_i(h)\left( \frac{h}{j}\right)^{c-1}$, for $h \le j \le l_{a+1}$.
As we know that $ hy_i(h) = m'_{l_a} > m_{l_{a+1}} = l_{a+1} y_i(l_{a+1})$,
in this last case, on setting $j =  l_{a+1}$, we can conclude that $c-1>0$, i.e.\ that $c>1$.

\begin{lemma}
\label{lem:forallj}
$y_i(j) > x_i(j)-1$ for all $j$.
\end{lemma}
\begin{proof}
{\bf Case 1}~ $m'_j = m_{l_a}$ for some $l_a < j$.

By Discrete WGS, $m_{l_a} \geq m_j - (l_a - 1)$. Thus
\[
y_i(j) = \frac{m'_j}{j} = \frac{m_{l_a}}{j} \geq \frac{m_j}{j} - \frac{l_a - 1}{j} > x_i (j) - 1.
\]

\medskip
\noindent{\bf Case 2}~ $m'_j < m_j$ but $m_j \neq m_{l_a}$ for any
$a$.

Then there is a maximal sequence $m_h, m_{h+1}, \cdots, m_k$,
with $j \in(h,k]$ and $k =  l_{a+1} - 1$ for some $a$,
for which $x_i(h) = x_i(h + 1) = \cdots = x_i(k) = x_i(k+1) + 1$.  Case 1
shows that $y_i(h) > x_i(h)-1$. By construction, $y_i(k) > x_i(k+1) =
x_i(k)-1$, and as the values $y'_h, y'_{h+1}, \cdots , y'_k$ form
a decreasing sequence by construction, it follows that
for $j$ with $h < j \le k$, $y_i(j) \ge y_i(k) >
x_i (k) - 1 = x_i(j) - 1$.
\end{proof}

\begin{lemma}
\label{lem:thespending}
$m_i(j)\ge m_i(j+1)$ for all $i$ and $j$.
\end{lemma}
\begin{proof}
The only case we need to check is when $y'_i(j+1)$ is defined by interpolation
for $j\in(h, l_{a+1})$ say.
By construction,
$p_iy'_i$ is a decreasing function of $p_i$ on the interpolating interval.
Now the result follows by induction on $p_{-i}$. The claim holds for $p_{-i}=1$,
as then $y'_i(p) = y_i(p)$. For larger $p_{-i}$, $y_i$ as a function of $p_i$ is simply
a max of noncreasing functions and hence is itself non-increasing.
\end{proof}

\begin{lemma}
\label{lem:thedemands}
At discrete prices for which $x_i>0$, the demands $y_i$ obey WGS.
\end{lemma}
\begin{proof}
First, we note that as $p_i$ increases the virtual spending $m'_i$
on good $i$ only decreases, and hence $y_i$ only decreases.

Next, we consider the change to $y_i$ as $p_j$ increases, for $j
\neq i$.  Let $p_i = h$. Consider prices $p_j$ and $p_j + 1$ for
good $j$.  Suppose that $y_i (h,p_j + 1)$ is set by a rule making
$m'_h(p_j + 1) = m_{l_a}(p_j+ 1)$, where $h \geq l_a$. By Discrete
WGS, $m_{l_a}(p_j + 1) \geq m_{l_a}(p_j)$. By construction,
$m_{l_a}(p_j) \geq m'_{l_a}(p_j) \geq m'_h(p_j)$. So $m'_h(p_j + 1)
\geq m'_h(p_j)$ and hence $y_i(h, p_j + 1) \geq y_i(h, p_j)$ in this
case.

Otherwise, by construction, $y_h(p_j +1) \geq y_h(p_j)$.
\end{proof}

\begin{lemma}
\label{lem:bdd-elasticity}
The demands $y_i$ obey bounded elasticity:
$$y_i(p_{-i}, p_i) \leq y_i (p_{-i}, p_i + d)
\left(1 + \frac{d}{p_i}\right)^{2E} \ \mbox{if} \ y_i (p_{-i},
p_i + d) > 0.$$
\end{lemma}
\begin{proof}
Note that it suffices to prove the bound for
$d = 1$, as
$$\left(1 + \frac{a}{p_i}\right) \left(1 + \frac{b}{a + p_i}\right) =
\frac{p_i + a + b}{p_i} = \left(1 + \frac{a + b}{p_i}\right).$$

\noindent{\bf Case 1}
$m'_{p_i} = m'_{p_i+1}$.

Then $y_i (p_i) = \frac{m'_{p_i}}{p_i}$, $y_i (p_i + 1) = \frac{m'_{p_i+1}}{p_i + 1}$.
So, $p_iy_i (p_i) = (p_i + 1) y_i (p_i + 1)$, or $y_i (p_i) \leq y_i (p_i + 1) \left(1 + \frac{1}{p_i} \right)$.

\smallskip
\noindent{\bf Case 2} $m'_h(p_l) = m'_h(p_l-1)$.

We note that $m'_{h+1}(p_l) \geq m'_{h+1}(p_l-1)$, so
$\frac{m'_h(p_l)}{m'_{h+1}(p_l)} \leq
\frac{m'_h(p_l-1)}{m'_{h+1}(p_l-1)}$,
and $\frac{y'(h, p_l)}{y'(h+1, p_l)} \leq
\frac{y'(h, p_l-1)}{y'(h+1, p_l-1)}$.
Hence the claim follows by
induction on $p_{-i}$ since Case 2 does not apply at the base cases when $p_{-i} = 1$.

\smallskip
\noindent{\bf Case 3} $x_i(h) \ge x(l_{a+1})+2$, where $l_a \le h < l_{a+1}$.

This is an intermediate bound, used in Case 4, except when $h+1 = l_{a+1}$.

By Discrete WGS,
$x_i(h) \leq \left\lceil x_i(l_{a+1}) \left(1 + \frac{l_{a+1} - h}{h}\right)^E \right\rceil$.
Hence $x_i(h) -1 <  x_i(l_{a+1}) \left(1 + \frac{l_{a+1} - h}{h}\right)^E$,
or
\[
\frac{x_i(l_{a+1}) + [x_i(h) - x_i(l_{a+1}) -1]} {x_i(l_{a+1})} \le \left(1 + \frac{l_{a+1} - h}{h}\right)^E.
\]
As $x_i(h)  - x_i(l_{a+1}) \ge 2$, $x_i(h)  - x_i(l_{a+1}) -1 \ge 1$.
Let $\Del$ denote $x_i(h)  - x_i(l_{a+1}) -1$.
So $1 + \frac{\Del}{x_i(l_{a+1})} \le \left(1 + \frac{l_{a+1} - h}{h}\right)^E$.
Hence $1 + \frac{2\Del}{x_i(l_{a+1})} \le \left(1 + \frac{l_{a+1} - h}{h}\right)^{2E}$,
or
\[
\frac{x_i(h) + [x_i(h) - x_i(l_{a+1}) -2]} {x_i(l_{a+1})} \le \left(1 + \frac{l_{a+1} - h}{h}\right)^{2E}.
\]
Hence $y_i(h) \le x_i(h) \le x_i(l_{a+1}) \left(1 + \frac{l_{a+1} - h}{h}\right)^{2E}$.

\smallskip
\noindent{\bf Case 4} $y_j$ is defined by interpolation.

Then there are an $h$ and an $l_a$ with $j\in(h, l_{a+1})$, such that
$x_i(h -1) > x_i(h) = x_i (h+1) = \cdots =
x_i(k) > x_i(k + 1)$, with $k+1 =  l_{a+1}$,
where $y'_{h+1}, y'_{h+2},
\cdots, y'_k$ are obtained by interpolation between $y'_h$ and
$y'_{k+1} = x_{k+1}$.

By Case 3, $y_i(h) \le y_i(l_{a+1})\left(1 + \frac{l_{a+1} - h}{h}\right)^{2E}$.
Let $c>1$ satisfy $y_i(h) = y_i(l_{a+1})\left(1 + \frac{l_{a+1} - h}{h}\right)^{c}$;
so $0< c \le 2E$ (in fact, $1 < c \le 2E$).
By construction, for $h\le j < l_{a+1} $,
$y_i(j) = y_i(l_{a+1})\left(1 + \frac{l_{a+1} - j}{j}\right)^{c}$;
thus $y_i(j) = y_i(j+1)\left(1 + \frac{1}{j}\right)^{c}
\le y_i(j+1)\left(1 + \frac{1}{j}\right)^{2E}$.
\end{proof} 
\subsection{Extensions}

Property \ref{prop:progress-1}, while implicit, underpins our analysis, but so far we have limited ourselves
to the case $\alpha = 1$.
Essentially, reducing $\alpha$ will slow the daily reduction in potential by an $\alpha$ factor.
Now we consider two extensions of the ongoing Fisher market with WGS and bounded elasticity which will have
an $\alpha$ in the range $(0,1)$.

\subsubsection{Beyond WGS}

Instead of having demands for goods $j \ne i$ only increase when $p_i$ increases by $\Del_i p_i>0$,
we could allow some of the demands to decrease, but with the constraint that the total of all reductions be
bounded as follows:
\begin{equation}
\label{eqn:beyond-wgs}
\sum_{x'_j < x_j} |x_j - x'_j| p_j \le (1 - \alpha) w_i \Del_i p_i,
\end{equation}
which we call the \emph{$\alpha$-bounded complements property}.
A symmetric condition applies when $\Del_i p_i < 0$.

\begin{lemma}
\label{lem:bdd-compl}
If demands obey the $\alpha$-bounded complements property then Property
\ref{prop:progress-1} holds.
\end{lemma}
\begin{proof}
The one change to the proof of Lemma \ref{lem:progress-1} is that the total additional detrimental
change in spending
when $p_i$ is updated is given by (\ref{eqn:beyond-wgs}), and it
increases the potential by at most $(1 - \alpha) w_i |\Del_i p_i|$.

The prior offsetting reduction in the potential by $I_i=  w_i |\Del_i p_i|$ continues to hold, so there
is a net reduction by at least $\alpha I_i$.
\end{proof}

The resulting change to the overall analysis is that the reductions to the potential by
$\alo w_i |\Del_i p_i|$, which drive the prior analysis,
now become $\alpha \alo w_i |\Del_i p_i|$, resulting in a
$1/\alpha$ slowdown in the convergence rates.

\subsubsection{Ongoing $\mathbf{\alpha}$-Exchange Markets}

At this point we review the standard definition of (One-Time) Exchange Markets.

\medskip
\noindent
{\bf The One-Time Exchange Market}\
A market comprises two sets, goods
$G$, with $|G| = n$, and  traders $T$, with $|T| = m$. The goods are
assumed to be infinitely divisible.
Each trader $l$ starts with an allocation $w_{il}$ of good $i$. Each
trader $j$ has a utility function $u_j(x_{1j},\cdots,x_{nj})$
expressing its preferences: if $j$ prefers a basket with $x_{ij}$
units (possibly a real number) of good $i$, to the basket with
$y_{ij}$ units, for $1\le i\le n$, then $u_j(x_{1j},\cdots,x_{nj})>
u_j(y_{1j},\cdots,y_{nj})$. Each trader $j$ intends to trade goods so
as to achieve a personal optimal combination (basket) of goods given
the constraints imposed by their initial allocation. The trade is
driven by a collection of prices $p_i$ for good $i$, $1\le i\le n$.
Agent $j$ chooses $x_{ij}$, $1\le i\le n$, so as to maximize $u_j$,
subject to the basket being affordable, that is: $\sum_{i=1}^n
x_{ij}p_i \le \sum_{i=1}^n w_{ij}p_i$. Prices ${\bf
p}=(p_1,p_2,\cdots,p_n)$ are said to provide an \emph{equilibrium}
if, in addition, the demand for each good is bounded by the supply:
$\sum_{j=1}^m x_{ij} \le \sum_{j=1}^m w_{ij}$. The market problem is
to find equilibrium prices.

\medskip

We consider Exchange markets which repeat on a daily basis in the same way as the ongoing Fisher markets.
For each good, there will be an agent with a warehouse for that good and this trader is the one
who meets demand fluctuations from the warehouse stock.

More specifically, each day there is a collection of identical traders who come to market each
with the same basket of goods,
including money, to trade. Furthermore, for each good, there is a trader with a warehouse, who is the price-setter for that
good and whose goal is to bring the warehouse to being half full.

We further assume that the price-setter for good $i$ brings at least $\alpha w_i$ of  good $i$ to the market and
is only seeking money for this portion of the supply.

We call this an $\mathbf{\alpha}$-Exchange market.

Then in the same way as for Lemma \ref{lem:bdd-compl}, one can show:

\begin{lemma}
\label{lem:bdd-exchange-mkts}
In  an $\mathbf{\alpha}$-Exchange market Property
\ref{prop:progress-1} holds.
\end{lemma}

Again, we obtain convergence slowed down by an $\alpha$ factor.

In \cite{cole-fleischer}
a related constraint on the buyer side was used. It specified that at equilibrium, traders coming
to market with money (buyers), would be buying at least $\alpha w_i$ of the $i$th good, for each $i$.
To ensure adequate demand among buyers for goods at other prices, they introduced an additional
assumption, termed the\emph{ wealth effect}. This assumption was a bound on the elasticity of wealth,
namely that when all prices dropped by a factor $f$, all demands increased by a factor
$f^{\beta}$ for some $\beta$, $0 <\beta < 1$.

It would be interesting to consider a version of the ongoing exchange market
in which the price-setter only meets bounded fluctuations in demand, and demands
beyond this range, along with unspent money, are carried forward to the next day.
The challenge, with respect to our approach for the analysis, is to ensure a
relatively small change in the computed demands for the preponderance of the goods
(in terms of their contribution to a suitable potential) for otherwise it is not clear that
the price updates will reduce the potential.  By computed demands, we have in mind something
analogous to the calculations used
in the ongoing Fisher market with warehouses:
this gives a low weight (a factor $\alt \kpi$) to the demand contribution used
to account for the warehouse excesses.

\section*{Acknowledgements}

Richard Cole and Lisa Fleischer are grateful to Nimrod Meggido for his
encouragement during the early stages of this work.
Richard Cole also thanks Jess Benhabib for pointers to the literature.

\bibliographystyle{plain}
\bibliography{markets} 

\end{document}